\documentclass[journal,onecolumn,10pt,twoside]{IEEEtran}
\pdfoutput=1
\usepackage{cite}
\usepackage{paralist,comment,color,fancybox,framed}
\usepackage{array}
\usepackage{amsmath,amsthm,amstext,amsfonts,amssymb}
\usepackage{url}
\usepackage[mathscr]{eucal}
\usepackage{cite}
\usepackage{xcolor,colortbl}
\usepackage[colorlinks   = true,urlcolor     = blue, 
  linkcolor    = blue, 
  citecolor   = red 
  ]{hyperref}
  
  \usepackage{lscape}
\usepackage{tikz}
\usepackage{pgfplots}
\pgfplotsset{compat=1.14} 
\usepackage{float}
\usepackage{booktabs,arydshln,tabu}
\usepackage{caption}
\usepackage{subcaption}
\usepackage{diagbox}

\usepackage{pdflscape} 
\usepackage{afterpage} 
\usepackage{rotating}

\usepackage{varwidth}
\usepackage{adjustbox,xcolor}
\usepackage{mathabx,mathtools}
\usepackage{multicol}
\usepackage{multirow}
\usepackage{array,makecell}
\usepackage{breqn}

\newcolumntype{?}{!{\vrule width 1pt}}
\newcolumntype{+}{!{\vrule width 0.15mm}}
\makeatother

\usepackage{enumitem}

\usepackage{algorithm}
\usepackage[noend]{algpseudocode}

\usetikzlibrary{matrix,shapes,decorations.pathreplacing,fit,backgrounds,calc,graphs,quotes,arrows,positioning}

\newtheorem{theorem}{Theorem}

\newtheorem{lemma}[theorem]{Lemma}
\newtheorem{corollary}[theorem]{Corollary}

\theoremstyle{definition}
\newtheorem{example}{Example}[section]
\theoremstyle{definition}

\definecolor{shadecolor}{gray}{.95}%
\definecolor{red2}{HTML}{DBA8B1}%

\newcommand{\x}{\Asterisk}
\newcommand{\tRnd}{\text{RandBC}}
\newcommand{\nix}[1]{}

\newenvironment{sysmatrix}[1]
 {\left(\begin{array}{@{}#1@{}}}
 {\end{array}\right)}

\newlength{\rowidth}%
\AtBeginDocument{\setlength{\rowidth}{4em}}

\begin{document}
\title{Lifting Constructions of PDAs for Coded Caching with Linear Subpacketization}
\author{Aravind~V~R,
        Pradeep~Kiran~Sarvepalli
        and~Andrew~Thangaraj,~\IEEEmembership{Senior Member,~IEEE}%
\thanks{All authors are with the Department
of Electrical Engineering, Indian Institute of Technology Madras,
Chennai 600036, India. 
Email: \{ee13d205, pradeep, andrew\}@ee.iitm.ac.in}%
}
\maketitle

\begin{abstract}
Coded caching is a technique where multicasting and coding opportunities are utilized to achieve better rate-memory tradeoff in cached networks. A crucial parameter in coded caching is subpacketization, which is the number of parts a file is to be split into for coding purposes. The original Maddah-Ali-Niesen scheme has order-optimal rate at a subpacketization growing exponentially with the number of users. In contrast, placement and delivery schemes in coded caching, designed using placement delivery arrays (PDAs), can have linear subpacketization with a penalty in rate. In this work, we propose several constructions of efficient PDAs through \textit{lifting}, where a base PDA is expanded by replacing each entry by another PDA. By proposing and using the notion of Blackburn-compatibility of PDAs, we provide multiple lifting constructions with increasing coding gains. We compare the constructed coded caching schemes with other existing schemes for moderately high number of users and show that the proposed constructions are versatile and achieve a good rate-memory tradeoff at low subpacketizations. 
\end{abstract}

\renewcommand{\arraystretch}{0.75}
\section{Introduction}
Communication networks get overburdened with data traffic during peak hours and underutilized in off-peak hours. Caching is a technique to alleviate the high transmission load of a server in a communication network during peak hours, and it involves prefetching popular content and storing it nearer to or at the user's device during off-peak hours. Depending on the limitations on memory, a part of these files would be prefetched and once the user makes a demand, the rest of the requested file will be transmitted. The fundamental problem in caching is the optimal trade-off between the cache memory with each user versus the transmission load.

Maddah-Ali and Niesen had shown in their seminal paper that coding can achieve significant gain over uncoded caching by making use of multicast opportunities \cite{maddah2014fundamental}.
They showed their scheme to be order optimal with an information-theoretic lower bound on the number of files needed to be transmitted (known as \textit{rate}).
This scheme achieves a \textit{coding gain} (also known as \textit{global caching gain}) in addition to the \textit{caching gain}.
Asymptotically, its coding gain is proportional to the number of users and that results in a rate independent of the number of users.
A version of Maddah-Ali-Niesen (M-N) scheme, optimal for uncoded prefetching, was presented in \cite{yu2017exact}. Though the exact capacity expression for rate of coded caching is still an open problem, several bounds have been presented in  \cite{sengupta2015improved,ghasemi2017improved,wang2016new,yu2018characterizing,wan2016optimality}. The problem has been studied in several settings like decentralized caching \cite{maddah2015decentralized}, non-uniform demands \cite{niesen2016coded}, hierarchical caching \cite{karamchandani2016hierarchical}, coded prefetching \cite{chen2016fundamental, tian2018caching, gomez2018fundamental}, content security \cite{sengupta2014fundamental,ravindrakumar2016fundamental}, demand privacy \cite{wan2020coded,aravind2020subpacketization,kamath2020demand} to name a few.

An important parameter of interest in coded caching is \textit{subpacketization}. It is the number of parts a file will be split into, for the purpose of coding.
In the standard Maddah-Ali-Niesen scheme \cite{maddah2014fundamental}, the subpacketization, denoted $f$, grows exponentially with $K$ {and is given by $f= \binom{K}{KM/N} \approx 2^{K H\left(M/N\right)}$}.
This limits the utility of the scheme in practical scenarios where there may be a large number of users (large $K$). Hence, reducing subpacketization is important in coded caching schemes.
It was shown in \cite{shangguan2018centralized} that for rate independent of the number of users, the subpacketization should be superlinear in the number of users.
A few other bounds relating subpacketization with other parameters were proposed in \cite{cheng2017coded,chittoor2020subexponential}.
The pursuit towards lower subpacketization had lead to formulating coded caching in combinatorial frameworks.  
Under certain constraints, the coded caching problem is equivalent to the design of {placement delivery arrays} \cite{yan2017placement}, caching matrices \cite{agrawal2019coded}, partial Latin {rectangles} with Blackburn property \cite{shangguan2018centralized,wanless2004partial}, certain 3-uniform 3-partite hypergraphs \cite{shangguan2018centralized}, resolvable designs from linear block codes \cite{tang2018coded}, induced matchings of a Ruzsa-Szeméredi graph \cite{shanmugam2017coded,shanmugam2017unified}, strong edge coloring of the bipartite graph \cite{yan2017bipartite} or a clique cover for the complement of the square of the associated line graph \cite{krishnan2018coded}.
These frameworks require the cache contents to be uncoded and symmetric with respect to all files.
When the number of users is very large, rate of $O(K^\delta)$ for small $\delta$ is achievable with linear subpacketization (in $K$) from schemes based on dense Ruzsa-Szeméredi graphs \cite{alon2012nearly,shanmugam2017coded}.
For reducing subpacketization in practical scenarios, a few families of constructions have been built based on the combinatorial frameworks we have discussed.
Some of these schemes are summarised in Table~\ref{tab:schemes}.

\begin{table}[t]
\caption{Parameters of some known coded caching schemes.}
\label{tab:schemes}
\centering
\adjustbox{max width=\textwidth}{
\begin{tabular}{|c|c|c|c|c|c|}
\hline
\textbf{Construction}                                                                & $\boldsymbol{K}$    & $\boldsymbol{M/N}$                             & $\boldsymbol{R}$                          & $\boldsymbol{f}$           & \textbf{Constraints}                                       \\ \hline
Maddah-Ali-Niesen \cite{maddah2014fundamental}                                                                   & $K$    & $\frac{t}{K}$                          & $\approx \frac{N}{M}$                 & $\approx 2^{K H(M/N)}$ &    $K,t \in \mathcal{Z}^+, K>t$                                                        \\ \hline
    Grouping by Shanmugam \textit{et al.} \cite{shanmugam2016finite}    & $ck$     & $\frac{t}{k}$                          & $\approx c\frac{N}{M}$                 & $\approx 2^{k H(M/N)}$ &    $c,k,t \in \mathcal{Z}^+, k>t$                                                        \\ \hline
Tang-Ramamoorthy   \cite{tang2018coded}                                                                  & $qn$            & $\frac{1}{q}$                                      & $\frac{(q-1)n}{k}+1$                  & $q^k$                  &                                                            \\ \hline
Shangguan \textit{et al.}   \cite{shangguan2018centralized}                                        & $\binom{n}{b}$  & $\frac{\binom{n}{a}-\binom{n-b}{a}}{\binom{n}{a}}$ & $\frac{\binom{n}{a+b}}{\binom{n}{a}}$ & $\binom{n}{a}$         &                                                            \\ \hline
\multirow{2}{*}{Agrawal \textit{et al.} \cite{agrawal2019coded}                }                           & $v$  & $1-\frac{k}{v}$ & $\frac{k(k-1)}{v-1}$ & $\frac{v(v-1)}{k(k-1)}$         & $\exists$ a $(v,k,1)$-BIBD.  
     \\ \cline{2-6}
& $v$  & $1-\frac{k-1}{v}$ & $1$ & $kv$         &       $\exists$ a simple symmetric $(v,k,2)$-BIBD.
               \\ \hline
    Alon \textit{et al.} \cite{alon2012nearly}, Shanmugam \textit{et al.} \cite{shanmugam2017coded}
 & $\ge K(\delta)$ & $\ge K^{-\epsilon(\delta)}$ & $K^\delta$ & $K$ & $\epsilon(\delta) \rightarrow 0$ and $K(\delta) \rightarrow \infty$ as $\delta \rightarrow 0$.
 \\ \hline
\end{tabular}}
\end{table}

In this work, we present a few construction schemes for placement delivery arrays.
\textit{Placement delivery array} or PDA developed by Yan \textit{et al.} captures the placement and delivery schemes as non-integer and integer entries in an array that satisfies some conditions \cite{yan2017placement}.
The number of columns of a PDA indicates the number of users while the number of rows indicates the subpacketization.
We focus on PDAs where the number of rows is linear with the number of columns.
Drawing inspiration from the lifting constructions for low-density parity-check codes \cite{thorpe2003low}, we propose lifting constructions for PDAs, where we use PDAs of small size to obtain larger PDAs.
We introduce a new technical notion of Blackburn-compatibility for PDAs that enable these lifting constructions. We propose a variety of constructions for PDAs which satisfy this constraint. This includes algebraic and randomized constructions. With the lifting constructions using these PDAs, we obtain good memory-rate tradeoffs with linear subpacketization. In particular, when the number of users has many divisors or is a power of $2$, the memory-rate tradeoffs are close to that of the standard Maddah-Ali and Niesen scheme. The random construction ensures PDAs satisfying Blackburn compatibility with good performance for arbitrary parameters. Our methods perform well both in terms of obtaining good coding gains and versatility.
Our specific contributions are as follows:
\begin{compactenum}[i)]
\item We propose construction schemes for $2$-regular placement delivery arrays for a wide range of parameters. The coded caching schemes from these PDAs are efficient in low memory regimes. This also gives a family of base PDAs for the subsequent lifting constructions.
\item We present a basic lifting construction that takes any two PDAs to construct larger PDAs with higher coding gain.
\item We propose the notion of Blackburn compatibility between PDAs and present a general lifting construction that uses a base PDA and a set of Blackburn compatible PDAs.
\item We present several constructions of Blackburn compatible PDAs for the lifting constructions which includes both algebraic and randomized constructions.
\item Using a combination of the constructions we propose, we demonstrate that significant coding gain and good memory-rate tradeoffs can be achieved with linear subpacketization.
\end{compactenum}
{The coded caching schemes that are constructed target the moderate (non-asymptotic) regime of parameters and are shown to be competitive with other existing schemes, particularly in terms of subpacketization.} We recently became aware of some existing works \cite{zhong2020placement,michel2019placement} that use the idea of combining PDAs to obtain new PDAs. Our constructions are more general in nature and follow a different approach.

The rest of the paper is organized as follows.
In Section~\ref{sec:ps}, we describe the system setup and the problem statement. 
We present the constructions for $2$-PDAs for a range of parameters in Section~\ref{sec:construction}.
In Section~\ref{sec:lifting}, we propose the lifting constructions for placement delivery arrays  from which we can construct coded caching schemes with low subpacketization. In Section~\ref{sec:blackburn}, we propose several constructions for Blackburn compatible PDAs which are the building blocks for the lifting constructions.
We present our results and compare them with existing works in Section~\ref{sec:results} and conclude with a brief summary of our work in Section~\ref{sec:conc}. An earlier version of this paper also contained some additional material on demand privacy and other algebraic constructions. Interested readers may refer to \cite{aravind2020coded}.

\section{Preliminaries}\label{sec:ps}
\subsection{Coded caching}
Consider a server, holding $N$ files $W_i$, $i\in[N]=\{0,1,\ldots, N-1\}$, of $F$ bits each, connected to $K$ users via a multicast link. 
User $k$, $k\in [K]$, has a cache $Z_k$ of size $MF$ bits.
Coded caching works in two phases. In the first phase, called the \textit{placement phase}, the cache $Z_k$ of User $k$ is populated with content by the server, while being unaware of the files demanded by the users. 
In the second phase, called the \textit{delivery phase}, User $k$ demands file $D_k \in [N]$ from the server.
Let $\boldsymbol{D}= (D_0, D_1, \ldots, D_{K-1})$. 
Based on the demands and stored cache contents, the server multicasts packets of the same size.
The entire multicast transmission from the server is denoted $X^{\boldsymbol{D}}$ for a demand vector $\boldsymbol{D}$, and we suppose that the length of $X^{\boldsymbol{D}}$ is $RF$ bits. 
The quantities $M$ and $R$ are measures of cache size and rate of transmission, respectively. 

The main requirement in a coded caching scheme is that User~$k$ should be able to decode the file $W_{D_k}$ using $Z_k$ and $X^{\boldsymbol{D}}$. 
We denote a coded caching scheme with $K$ users, $N$ files, local cache size $M$, and rate $R$ as a $(K,N;M,R)$ coded caching scheme, or as a $(K,N)$ scheme.

\subsection{Coded caching schemes from PDAs}
We use the framework of placement delivery arrays for centralized coded caching schemes \cite{yan2017placement}. 
For positive integers $K$, $f$, $Z$ and a set of integers $\mathcal{S}$, a $(K, f, Z, \mathcal{S})$ \textit{placement delivery array} is an $f\times K $ matrix  $P=[p_{j,k}]$, $j\in[f],k\in[K]$, containing either a ``$\x$'' or integers from $\mathcal{S}$ in each array cell such that they satisfy the following conditions.
\begin{enumerate}[label=C\arabic*., ref=C\arabic*]
    \item The symbol $\x$ appears $Z$ times in each column.\label{cond:equalZ}
    \item Each integer $s\in \mathcal{S}$ occurs at least once in the array.\label{cond:everysOnce}
    \item (Blackburn property) If the entries in two distinct cells $p_{j_1,k_1}$ and $p_{j_2,k_2}$ are the same integer $s\in\mathcal{S}$, then $p_{j_{1},k_{2}}=p_{j_{2},k_{1}}=\Asterisk$.\label{cond:blackburn}
\end{enumerate}
If there is no ambiguity, we will use the notation $(K,f,Z,S)$, where $S$ is a positive integer, for a PDA implicitly assuming $\mathcal{S}=\{0,1,\ldots,S-1\}$. The construction of a coded caching scheme from a PDA was proved by Yan \textit{et al.} \cite{yan2017placement}, and is reproduced below for reference.
\begin{theorem}[Coded caching schemes from PDAs \cite{yan2017placement}]\label{th:yanpda}
For a given $(K,f,Z,\mathcal{S})$ PDA, $P = [p_{j,k}]_{f\times K}$,
there exists a corresponding  $(K,N;M,R)$ caching system with subpacketization $f$, $M/N = Z/f$ and $R = |\mathcal{S}|/f$.
\end{theorem}
In the caching scheme corresponding to a PDA, each file $W_i$ is split into $f$ subfiles $W_{i,j}$, $i\in[N]$, $j\in[f]$. Each row of the PDA corresponds to a subfile label and if $p_{j,k}=\Asterisk$, then User~$k$'s cache is loaded with the $j$-th subfile of every file. So the cache $Z_k$ of User $k$ is given by
\begin{align}
    Z_k = \{W_{i,j}:\forall i\in [N], p_{j,k}=\Asterisk \}.
\end{align}
For the demand vector $\boldsymbol{D}=\{D_k: k \in [K] \}$, the server transmits
\begin{align}
    X^{\boldsymbol{D}}=\left\{\bigoplus_{j\in [f], k\in [K], p_{j,k}=s} W_{D_k,j}: s\in \mathcal{S}\right\}.
\end{align}
In the packet of $X^{\boldsymbol{D}}$ corresponding to $s$, User~$k_0$ has all the subfiles occurring in the XOR in its cache except $W_{D_{k_0},j_0}$, where  $j_0$ is such that $p_{j_0,k_0}=s$. This is because, if $p_{j,k}=s$ for $k\ne k_0$, then $j\ne j_0$ and $p_{j,k_0}=\Asterisk$ by \ref{cond:blackburn}. So, User $k_0$ recovers $W_{D_{k_0},j}$ whenever $p_{j,k_0}\ne\Asterisk$.
In this manner, each user can recover the entire file demanded using cache contents and transmissions.

A placement delivery array P is said to be a $g$-regular $(K,f,Z,\mathcal{S})$ PDA or $g$-$(K,f,Z,\mathcal{S})$ PDA or $g$-PDA for short, if each integer in $\mathcal{S}$ appears $g$ times in $P$. This implies that $|\mathcal{S}|g=K(f-Z)=$ number of integer cells in the PDA, or $|\mathcal{S}|=K(f-Z)/g$. The constant $g$, called the \textit{coding gain} \cite{yan2017placement} (or \textit{global caching gain} in \cite{maddah2014fundamental}), indicates the number of users that recover a subfile from a single transmission. 
The rate of the coded caching scheme obtained from a $g$-regular PDA is
\begin{align}
    R &= \frac{|\mathcal{S}|(F/f)}{F} = \frac{K(f-Z)}{fg}. \label{eq:regularPdaRate}
\end{align}
In terms of parameters, our focus is on the regime where subpacketization $f$ is linear in $K$, i.e. $f=mK$ for a positive integer $m$, and the number of users $K$ is \emph{moderately} high. In this regime, we provide explicit constructions for $(K,f,Z,\mathcal{S})$ PDAs that achieve different trade-offs between the per-file cache size $Z/f$ and the rate $|\mathcal{S}|/f$. 

\section{2-regular PDAs}\label{sec:construction}
In this section, we will present some methods for constructing 2-regular PDAs.
Most of these constructions are direct and, in some cases, they are special cases of other known constructions. We include them here for reference as we use them later as ingredients in lifting constructions.

We will consider $(n,n,Z,\mathcal{S})$-PDAs that are $2$-regular. The corresponding $(K,N;M,R)$ coded caching scheme can have arbitrary number of files $N$, $K=n$ users, $M=(NZ/n)$, and the rate $R=(n-Z)/2$ (using Eq.~\eqref{eq:regularPdaRate}). 
We begin with a direct construction of a $2$-regular PDA corresponding to the Maddah-Ali-Niesen scheme \cite{niesen2016coded} for $t=\frac{KM}{N}=1$. 
\begin{lemma}[Dense 2-regular PDAs]\label{lm:diagPDA}
For an integer $n \geq 2$, there exists a $2$-regular $(n,n,1,n(n-1)/2)$ PDA.
\end{lemma}
\begin{proof}
Take an $n \times n$ array. Set all diagonal (or anti-diagonal) entries to $\x$. There are $S=n(n-1)/2$ cells below the diagonal. Fill them with the integers from 0 to $S-1$. Symmetrically fill the cells above the diagonal. 
\end{proof}
For a set $\mathcal{S}=\{s_1\,\ldots,s_{n(n-1)/2}\}$ of $n(n-1)/2$ integers, we let $G_n(\mathcal{S})$ and $H_n(\mathcal{S)}$ denote the $(n,n,1,\mathcal{S})$ PDAs obtained using Lemma \ref{lm:diagPDA} with anti-diagonal and diagonal cells set as $\x$, respectively. The integers in $\mathcal{S}$ are arranged row-wise in the specified order above the anti-diagonal in $G_n(\mathcal{S})$ or below the diagonal in $H_n(\mathcal{S})$. For example,
    \begin{align}
    G_3(\{4,6,5\})=\begin{pmatrix}
    4  & 6  & \x \\
    5  & \x & 6  \\
    \x & 5  & 4
    \end{pmatrix}, G_3(\{6,4,5\})=\begin{pmatrix}
    6  & 4  & \x \\
    5  & \x & 4  \\
    \x & 5  & 6
    \end{pmatrix},
    H_3(\{3,2,1\})=\begin{pmatrix}
    \x & 3 & 2 \\
    3 & \x & 1 \\
    2 & 1 & \x
    \end{pmatrix}.\label{eq:GS}
    \end{align}

We now provide methods to modify a 2-regular PDA $P$ from Lemma \ref{lm:diagPDA} to lower $|\mathcal{S}|$, while retaining regularity of 2. The basic idea is to replace integers in $P$ with $\x$ without affecting regularity. 
For this purpose, given a 2-regular $(n,n,Z,\mathcal{S})$ PDA $P$ with rows/columns indexed from 0 to $n-1$ and $s\in \mathcal{S}$ at locations $(i_1(s),j_1(s))$ and $(i_2(s),j_2(s))$, we associate a graph $G(P)=(V,E)$ with vertex set $V=\{0,1,\ldots,n{-}1\}$ representing the columns of $P$ and edge set $E=\{(j_2(s),j_1(s)):s\in S\}$. We refer the reader to \cite{Harary1969} for definitions and basic results in graph theory. The edge $e(s)=(j_2(s),j_1(s))\in E$ is labelled with the triple $(i_1(s),i_2(s),s)$. For a symmetric PDA $P$ (such as the one from Lemma \ref{lm:diagPDA}), since $i_2(s)=j_1(s)$ and $j_2(s)=i_1(s)$, we have $e(s)=(i_1(s),j_1(s))$ and the edge label is shortened to $s$.
\begin{lemma}
Let $P$ be the 2-regular $(n,n,1,n(n-1)/2)$ PDA from Lemma \ref{lm:diagPDA}. Then, the associated graph $G(P)$ is equal to $K_n$, the complete graph on $n$ vertices.
\label{lem:completegraphPDA}
\end{lemma}
\begin{proof}
By symmetry, the first column of $P$ is equal to the first row, if diagonal is set to $\x$, or {the reverse of} the last row, if anti-diagonal is set to $\x$. So, in $G(P)$, vertex 1 is connected to all other vertices $\ge 2$.
Now, delete the first column and its corresponding first or last symmetric row, and see that vertex 2 is connected to vertices $\ge 3$. Proceed iteratively to complete the proof.
\end{proof}
A spanning subgraph has the same vertex set as the original graph and a subset of its edges with no isolated vertices. A graph or subgraph is said to be $r$-regular if every vertex has degree equal to $r$. An $r$-regular spanning subgraph is very useful for modifying PDAs as shown in the following lemma. 
\begin{theorem}
Consider a 2-regular $(n,n,z,\mathcal{S})$ PDA $P$ with associated graph $G(P)$. Suppose $G(P)$ has an $r$-regular spanning subgraph with its $nr/2$ edges being $\{e(s):s\in \mathcal{S}_r\}$, where $\mathcal{S}_r\subset\mathcal{S}$. The array obtained by setting $s\in\mathcal{S}_r$ as $\x$ in $P$ is a 2-regular $(n,n,z+r,\mathcal{S}\setminus\mathcal{S}_r)$ PDA.
\label{lem:rregular}
\end{theorem}
\begin{proof}
Since the subgraph is spanning and $r$-regular, exactly $r$ integers in $\mathcal{S}_r$ are present in each column of $P$. Setting the integers in $\mathcal{S}_r$ to $\x$ results in the modified PDA as claimed.
\end{proof}
To find regular spanning subgraphs, the notions of 1-factors and 1-factorization are useful \cite{wallis2013one}. A matching in a graph is a set of non-intersecting (or parallel) edges. A matching is said to be a 1-factor if it covers all vertices. A 1-factor is clearly a 1-regular spanning subgraph.

A complete graph $K_n$, for $n$ even, has multiple 1-factors each with $n/2$ edges \cite{wallis2013one}. A 1-factorization of $K_n$, $n$ even, is a partition of its $n(n{-}1)/2$ edges into $n{-}1$ edge-disjoint 1-factors. It is well known that 1-factorizations exist for $K_n$ when $n$ is even \cite{factorizationsurvey85}. The union of $r$ different 1-factors in a 1-factorization is clearly an $r$-regular spanning subgraph of $K_n$, which can be used in Lemma \ref{lem:rregular} as follows.
\begin{corollary}
Let $P$ be the 2-regular $(n,n,1,n(n{-}1)/2)$ PDA from Lemma \ref{lm:diagPDA} with associated graph $K_n$, $n$ even. Let $\{M_1,M_2,\ldots,M_{n-1}\}$ be a 1-factorization of $K_n$ with the $i$-th 1-factor $M_i=\{e(s_{ij}): j\in\{1,2,\ldots,n/2\}\}$. For $z\in\{1,2,\ldots,n{-}2\}$, let $P_z$ be the array obtained by setting $s_{ij}$ to $\x$ in $P$ for $1\le i\le z$ and $j\in\{1,2,\ldots,n/2\}$. Then, $P_z$ is a 2-regular $(n,n,z+1,n(n{-}z{-}1)/2)$ PDA. 
\label{lem:modify_even}
\end{corollary}
\begin{proof}
Since $M_1\cup M_2\cup\cdots\cup M_z$ is a $z$-regular spanning subgraph of $K_n$, the result follows by the use of Lemma \ref{lem:rregular}.
\end{proof}
For $n=4$, consider the $(4,4,1,6)$ PDA $H_4([6])$. A 1-factorization for the associated graph is $M_1=\{e(0),e(5)\}$, $M_2=\{e(1),e(4)\}$, $M_3=\{e(2),e(3)\}$ as shown in Fig.~\ref{fig:pda_graph}. The modified PDAs obtained using this 1-factorization in Lemma \ref{lem:modify_even} are easy to write down.
\begin{figure}
    \centering
    \includegraphics{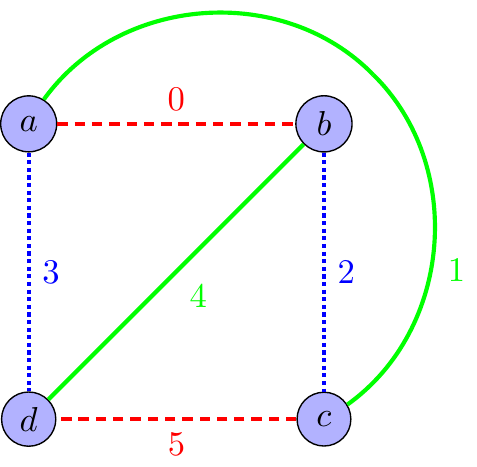}
    \caption{A 1-factorization of the graph associated with the $(4,4,1,6)$ PDA $H_4([6])$.}
    \label{fig:pda_graph}
\end{figure}
For $n$ odd, there are no 1-factors in $K_n$, and the smallest regular spanning subgraph is a Hamiltonian cycle, which is a cycle with $n$ edges passing through all $n$ vertices \cite{Harary1969}. It is well-known that $K_n$ for $n$ odd can be decomposed into $(n-1)/2$ edge-disjoint Hamiltonian cycles \cite{wallis2013one,AlspachWalecki08}. This decomposition leads to the following lemma.
\begin{corollary}
Let $P$ be the 2-regular $(n,n,1,n(n{-}1)/2)$ PDA from Lemma \ref{lm:diagPDA} with associated graph $K_n$, $n$ odd. Let $\{H_1, H_2,\allowbreak \ldots, H_{(n-1)/2}\}$ be a set of edge-disjoint Hamiltonian cycles of $K_n$ with $H_i=\{e(s_{ij}): j\in\{1,2,\ldots,n\}\}$. For $z\in\{1,2,\ldots,(n{-}3)/2\}$, let $P_z$ be the array obtained by setting $s_{ij}$ to $\x$ in $P$ for $1\le i\le z$ and $j\in\{1,2,\ldots,n\}$. Then, $P_z$ is a 2-regular $(n,n,2z+1,n(n{-}2z{-}1)/2)$ PDA.
\label{lem:modify_odd}
\end{corollary}
\begin{proof}
Since $H_1\cup H_2\cup\cdots\cup H_z$ is a $2z$-regular spanning subgraph of $K_n$, the result follows by the use of Lemma \ref{lem:rregular}.
\end{proof}
For $n=5$, consider the $(5,5,1,10)$ PDA in $H_5([10])$. A Hamiltonian cycle decomposition for the associated graph is $H_1=\{e(0),e(1),e(5),e(6),e(7)\}$,  $H_2=\{e(2),e(8),e(9),e(4),e(3)\}$. The modified PDAs obtained using this decomposition in Cor.~\ref{lem:modify_odd} are easy to write down.

We remark that the cache memory versus rate trade-off obtained by the 2-PDAs of Corollaries \ref{lem:modify_even} and \ref{lem:modify_odd} are the same as those obtained by memory-sharing between the full-storage scheme corresponding to the all-$\x$ $(n,n,n,0)$ PDA and the $(n,n,1,n(n-1)/2)$ PDA of Lemma \ref{lm:diagPDA}. However, explicit 2-PDAs that achieve the same trade-off without memory-sharing have not appeared earlier in the literature. In addition, the connection between PDAs and graphs appears to be new as well, and this connection could possibly lead to some interesting generalizations.
\section{Lifting or protograph-type constructions}\label{sec:lifting}
Constructions of PDAs with coding gain 2 and low subpacketization were briefly described in the previous section. To increase coding gain and obtain $g$-regular PDAs for $g>2$ without a significant increase in subpacketization, we employ the idea of lifting or protograph construction. Similar to the popular notion of protograph or lifted LDPC codes, we start with a base PDA, and replace each entry with another PDA. The PDAs that replace an integer or a $\x$ in the base PDA are called as \emph{constituent} PDAs. 

An important requirement when lifting PDAs is that we have to ensure that the Blackburn property is preserved during the lifting.
For this purpose, we define a constraint called \emph{Blackburn compatibility} which needs to be satisfied by the constituent PDAs for the lifting to be valid. 

We propose several deterministic and a randomized construction for lifting of PDAs and compare them with other existing PDAs in terms of their memory-rate tradeoff and subpacketization.  

\subsection{Notation for PDAs}\label{sec:notation}
The following PDAs are used repeatedly in lifting constructions.
\begin{enumerate}
    \item For an integer $t$, $I_n(t)$ denotes the $(n,n,n-1,1)$ PDA with the integer $t$ on the main diagonal and $\x$ in all other cells. $\tilde{I}_n(t)$ denotes the $(n,n,n-1,1)$ PDA with the integer $t$ on the main anti-diagonal and $\x$ in all other cells. For example, 
    \begin{align}
    I_3(1)=\begin{pmatrix}
    1  & \x & \x \\
    \x & 1  & \x \\
    \x & \x & 1
    \end{pmatrix},\quad \tilde{I}_3(0)=\begin{pmatrix}
    \x & \x & 0  \\
    \x & 0  & \x \\
    0  & \x & \x
    \end{pmatrix}.\label{eq:It}
    \end{align}
    \item For a set $\mathcal{S}=\{s_1,\ldots,s_{n(n-1)/2}\}$ of $n(n-1)/2$ integers, the PDAs $G_n(\mathcal{S})$ and $H_n(\mathcal{S})$ are as defined in Lemma \ref{lm:diagPDA}.
    \item For a set $\mathcal{S}=\{s_1,\ldots,s_{n^2}\}$ of $n^2$ integers, $J_n(\mathcal{S})$ denotes the $(n,n,0,\mathcal{S})$ PDA obtained by filling all the cells in the array with distinct integers from $\mathcal{S}$ row-wise in the specified order. For example,
    \begin{align}
    J_3([9])=\begin{pmatrix}
    0  & 1  & 2 \\
    3  & 4  & 5 \\
    6  & 7  & 8
    \end{pmatrix}.\label{eq:JS}
    \end{align}
\end{enumerate}
\subsection{Basic lifting}
In the first lifting method, we start with a base PDA and replace $\x$'s with an all-$\x$ array, and replace integers with PDAs that contain disjoint sets of integers. 
\begin{theorem}[Basic lifting]\label{th:basicTiling}
Let $P_b$ be a $(K,f,Z,\mathcal{S}_b)$ PDA. Let $\mathcal{P}=\{P_i:i\in\mathcal{S}_b\}$, where $P_i$ is an $(m,n,e,\mathcal{S}_i)$ PDA and $\mathcal{S}_i$, $\mathcal{S}_j$ are disjoint if $i\ne j$. Let an array $B_{\mathcal{P}}(P_b)$ be defined as follows:
\begin{enumerate}
    \item Each $\x$ in $P_b$ is replaced by a $n\times m$ all-$\x$ array.
    \item Each integer $i\in\mathcal{S}_b$ is replaced by $P_i\in\mathcal{P}$.
\end{enumerate}
Then, $B_{\mathcal{P}}(P_b)$ is a $(Km,fn,Zn+(f-Z)e,\mathcal{S})$ PDA, where $\mathcal{S}=\bigcup_{i\in\mathcal{S}_b}\mathcal{S}_i$. 
\end{theorem}
\begin{proof}
See the proof of a more general version in Theorem \ref{th:generalTiling}.
\end{proof}
 
\begin{example}
Two PDAs obtained using Theorem~\ref{th:basicTiling} are shown in Table \ref{tab:basicTiling}. The first one is with ${P_{b_1}}=\begin{pmatrix}
0&1&2\\
3&4&5\\
6&7&8
\end{pmatrix}$ and integer $i$ replaced with $I_3(i)$ resulting in a 3-regular $(9,9,6,9)$ PDA. The second one uses ${P_{b_2}}=\begin{pmatrix}
\x&0&1\\
0&\x&2\\
1&2&\x
\end{pmatrix}$ with 0 replaced by $G_3(\{0,1,2\})$, 1 replaced by $H_3(\{3,4,5\})$ and 2 replaced by $G_3(\{6,7,8\})$ resulting in a 4-regular $(9,9,5,9)$ PDA.
\begin{table}[htb]
    \centering
    \caption{Examples of lifting using Lemma \ref{th:basicTiling}. {The PDAs shown below are obtained by lifting $P_{b_1}$ and $P_{b_2}$ respectively.}}
    \label{tab:basicTiling}
    \begin{tabular}{?c+c+c?c+c+c?c+c+c?}
\Xhline{2\arrayrulewidth}
0 & $\x$ & $\x$ & \cellcolor{gray!50} 1 & \cellcolor{gray!50} $\x$ & \cellcolor{gray!50} $\x$ & 2 & $\x$ & $\x$ \\ \hline
$\x$ & 0 & $\x$ & \cellcolor{gray!50} $\x$ & \cellcolor{gray!50} 1 & \cellcolor{gray!50} $\x$ & $\x$ & 2 & $\x$ \\ \hline
$\x$ & $\x$ & 0 & \cellcolor{gray!50} $\x$ & \cellcolor{gray!50} $\x$ & \cellcolor{gray!50} 1 & $\x$ & $\x$ & 2 \\ \Xhline{2\arrayrulewidth}
\cellcolor{gray!50} 3 & \cellcolor{gray!50} $\x$ & \cellcolor{gray!50} $\x$ & 4 & $\x$ & $\x$ & \cellcolor{gray!50} 5 & \cellcolor{gray!50} $\x$ & \cellcolor{gray!50} $\x$ \\ \hline
\cellcolor{gray!50} $\x$ & \cellcolor{gray!50} 3 & \cellcolor{gray!50} $\x$ & $\x$ & 4 & $\x$ & \cellcolor{gray!50} $\x$ & \cellcolor{gray!50} 5 & \cellcolor{gray!50} $\x$ \\ \hline
\cellcolor{gray!50} $\x$ & \cellcolor{gray!50} $\x$ & \cellcolor{gray!50} 3 & $\x$ & $\x$ & 4 & \cellcolor{gray!50} $\x$ & \cellcolor{gray!50} $\x$ & \cellcolor{gray!50} 5 \\ \Xhline{2\arrayrulewidth}
6 & $\x$ & $\x$ & \cellcolor{gray!50} 7 & \cellcolor{gray!50} $\x$ & \cellcolor{gray!50} $\x$ & 8 & $\x$ & $\x$ \\ \hline
$\x$ & 6 & $\x$ & \cellcolor{gray!50} $\x$ & \cellcolor{gray!50} 7 & \cellcolor{gray!50} $\x$ & $\x$ & 8 & $\x$ \\ \hline
$\x$ & $\x$ & 6 & \cellcolor{gray!50} $\x$ & \cellcolor{gray!50} $\x$ & \cellcolor{gray!50} 7 & $\x$ & $\x$ & 8
\\ \Xhline{2\arrayrulewidth}
\end{tabular}
\hspace{0.5cm}
\begin{tabular}{?c+c+c?c+c+c?c+c+c?}
\Xhline{2\arrayrulewidth}
$\x$ & $\x$ & $\x$ & \cellcolor{gray!50} 0 & \cellcolor{gray!50} 1 & \cellcolor{gray!50} $\x$ & $\x$ & 3 & 4 \\ \hline
$\x$ & $\x$ & $\x$ & \cellcolor{gray!50} 2 & \cellcolor{gray!50} $\x$ & \cellcolor{gray!50} 1 & 3 & $\x$ & 5 \\ \hline
$\x$ & $\x$ & $\x$ & \cellcolor{gray!50} $\x$ & \cellcolor{gray!50} 2 & \cellcolor{gray!50} 0 & 4 & 5 & $\x$ \\ \Xhline{2\arrayrulewidth}
\cellcolor{gray!50} 0 & \cellcolor{gray!50} 1 & \cellcolor{gray!50} $\x$ & $\x$ & $\x$ & $\x$ & \cellcolor{gray!50} 6 & \cellcolor{gray!50} 7 & \cellcolor{gray!50} $\x$ \\ \hline
\cellcolor{gray!50} 2 & \cellcolor{gray!50} $\x$ & \cellcolor{gray!50} 1 & $\x$ & $\x$ & $\x$ & \cellcolor{gray!50} 8 & \cellcolor{gray!50} $\x$ & \cellcolor{gray!50} 7 \\ \hline
\cellcolor{gray!50} $\x$ & \cellcolor{gray!50} 2 & \cellcolor{gray!50} 0 & $\x$ & $\x$ & $\x$ & \cellcolor{gray!50} $\x$ & \cellcolor{gray!50} 8 & \cellcolor{gray!50} 6 \\ \Xhline{2\arrayrulewidth}
$\x$ & 3 & 4 & \cellcolor{gray!50} 6 & \cellcolor{gray!50} 7 & \cellcolor{gray!50} $\x$ & $\x$ & $\x$ & $\x$ \\ \hline
3 & $\x$ & 5 & \cellcolor{gray!50} 8 & \cellcolor{gray!50} $\x$ & \cellcolor{gray!50} 7 & $\x$ & $\x$ & $\x$ \\ \hline
4 & 5 & $\x$ & \cellcolor{gray!50} $\x$ & \cellcolor{gray!50} 8 & \cellcolor{gray!50} 6 & $\x$ & $\x$ & $\x$
\\ \Xhline{2\arrayrulewidth}
\end{tabular}
\end{table}
As a comparison, a $9\times 9$ PDA constructed using Corollary \ref{lem:modify_odd} has parameters $(9,9,5,18)$ and is 2-regular. The $(9,9,5,9)$ PDA using basic lifting is 4-regular and provides a lower rate coded caching scheme at the same memory when compared to the $(9,9,5,18)$ PDA.
\label{ex:basic}
\end{example}
The requirements imposed by Theorem \ref{th:basicTiling} on $P_i$ are strictly not necessary for the lifting to result in a valid PDA. Firstly, $P_i$ need not be valid PDAs by themselves. Secondly, the $P_i$ need not contain disjoint integers. A simple example is the lifting of $P_b=\begin{pmatrix}
    0\\
    1
\end{pmatrix}$ with $\mathcal{P}=\{P_0=\begin{pmatrix}
    0&\x\\
    \x&\x
\end{pmatrix}, P_1=\begin{pmatrix}
    \x&\x\\
    \x&0
\end{pmatrix}\}$, which results in a trivial 2-regular $(2,4,3,1)$ PDA. While for simplicity of general constructions and for dense base PDAs, constituent PDAs with disjoint sets of integers appear to be a good choice, further optimizations of the lifted PDAs will be possible for improving parameters and trade-offs. 
\subsubsection{Regular basic lifting}
The basic lifting construction is simple, and provides PDAs of various sizes with higher coding gains in a direct manner. The simplest $g$-regular construction by basic lifting is captured in the following corollary to Theorem~\ref{th:basicTiling}.
\begin{corollary}
Let $P_b$ be a $g_b$-regular $(K_b,f_b,Z_b,\frac{K_b(f_b-Z_b)}{g_b})$ PDA. Let $P_l$ be a $g_c$-regular $(m,n,e,\frac{m(n-e)}{g_c})$ PDA, and $\mathcal{P}=\{P_i:i\in[\frac{K_b(f_b-Z_b)}{g_b}]\}$, where $P_i$ are copies of $P_l$ with its integers replaced by another disjoint set of integers. Then, $B_{\mathcal{P}}(P_b)$, which is denoted simply as $B_{P_l}(P_b)$ in this case, is a $g_bg_c$-regular $(K_bm,f_bn,Z_bn+(f_b-Z_b)e,\frac{K_b(f_b-Z_b)}{g_b}\frac{m(n-e)}{g_c})$ PDA.
\label{cor:basicreg}
\end{corollary}
\begin{proof}
See proof of a more general case in Corollary \ref{cor:genreg}.
\end{proof}
As stated earlier, 1-PDAs and 2-PDAs are the easiest to construct and use as base PDAs and constituent PDAs in lifting. Let 1-PDA$(n,z)$ denote a $(n,n,z,n(n-z))$ 1-PDA, and let 2-PDA$(n,z)$ denote the $(n,n,z,n(n-z)/2)$ 2-PDA obtained using Corollaries \ref{lem:modify_even} and \ref{lem:modify_odd} for applicable values of $n,z$.

In Corollary \ref{cor:basicreg}, using a 1-PDA as $P_b$ or $P_l$, we do not obtain an increase in coding gain. An increase in coding gain is obtained if we consider 2-PDAs as $P_b$ and $P_l$. 
\begin{example}[2-PDA to 4-PDA]
In Corollary \ref{cor:basicreg}, let $P_b$ be 2-PDA$(K_b,Z_b)$ and $P_l$ be 2-PDA$(n,e)$, where $z$ and $e$ are chosen so that the 2-PDAs exist (see Corollaries \ref{lem:modify_even} and \ref{lem:modify_odd}). $B_{P_l}(P_b)$ is a $(K_bn,K_bn,Z_bn+(K_b-Z_b)e,K_bn(K_b-Z_b)(n-e)/4)$ 4-PDA. 
\end{example}
While each $P_i$ in $\mathcal{P}$ being regular is sufficient for the lifted PDA to be regular, it is not strictly necessary. To see this through an example, let $\mathcal{P}_3=\{P_0,P_1\}$ be the set of PDAs defined as
\begin{align}
    P_0 = \left(
        \begin{array}{ccc}
            \x & 0 & 2 \\
            0 & \x & 1 \\
            3 & 1 & \x
        \end{array}
        \right),\quad
    P_1 = \left(
        \begin{array}{ccc}
            1 & 2 & \x \\
            3 & \x & 2 \\
            \x & 3 & 0
        \end{array}
        \right).
        \label{eq:A3ex}
\end{align}
For $P_b=\begin{pmatrix}
0&\x\\
\x&1
\end{pmatrix}
$, $B_{\mathcal{P}_3}(P_b)$ is 3-regular, while neither $P_0$ nor $P_1$ are individually regular. 
\subsubsection{Recursive basic lifting}
Recursive application of basic lifting is useful in obtaining multiple lifted PDAs with different coding gains. This procedure is particularly effective if the final target number of users $K$ can be factored into a product of $r$ numbers as $K=K_1K_2\cdots K_r$, and is captured in the following lemma for reference (proof is skipped).
\begin{lemma}[Recursive basic lifting for $K=K_1K_2\cdots K_r$ users]\label{lm:recTiling}
Let $P_b^{(1)}$ be a $K_1\times K_1$ $g_1$-PDA, and let $P_l^{(i)}$ for $i=2,\ldots,r$ be $K_i\times K_i$ $g_i$-PDAs. Consider the recursion $P_b^{(i)}=B_{P_{l}^{(i)}}(P_b^{(i-1)})$ for $i=2,\ldots,r$. The result of the recursion $P_b^{(r)}$ is a $(g_1\cdots g_r)$-regular $K\times K$ PDA. 
\end{lemma}
We illustrate the above lemma with 1-PDAs and/or 2-PDAs used as constituent PDAs.
\begin{example}[1-PDA and 2-PDA recursive basic lifting  for $K=K_1K_2\cdots K_r$ users]
Let $g_i\in\{1,2\}$ for $i=1,\ldots,r$. If $g_i=1$, let $z_i\in[K_i]$. If $g_i=2$, let $z_i\in\{1,2,\ldots,K_i-1\}$ if $K_i$ is even, or let $z_i\in\{1,3,\ldots,K_i-2\}$ if $K_i$ is odd. Let $P_b^{(1)}$ be a $K_1\times K_1$ $g_1$-PDA, and let $P_l^{(i)}$ for $i=2,\ldots,r$ be $K_i\times K_i$ $g_i$-PDAs (these PDAs exist by Corollaries \ref{lem:modify_even} and \ref{lem:modify_odd}). Consider the recursion $P_b^{(i)}=B_{P_{l}^{(i)}}(P_b^{(i-1)})$ for $i=2,\ldots,r$. The number of $\x$s per column of $P_b^{(i)}$, denoted $Z_i$, is given by the recursion 
\begin{equation}
Z_i=Z_{i-1}K_i+(K_1\cdots K_{i-1}-Z_{i-1})z_i,\, i=2,\ldots,r,
\label{eq:zrecbasic}
\end{equation}
initialised with $Z_1=z_1$. The result of the recursive lifting $P_b^{(r)}$ is a $(g_1\cdots g_r)$-regular $K\times K$ PDA with $Z_r$ given by \eqref{eq:zrecbasic}.

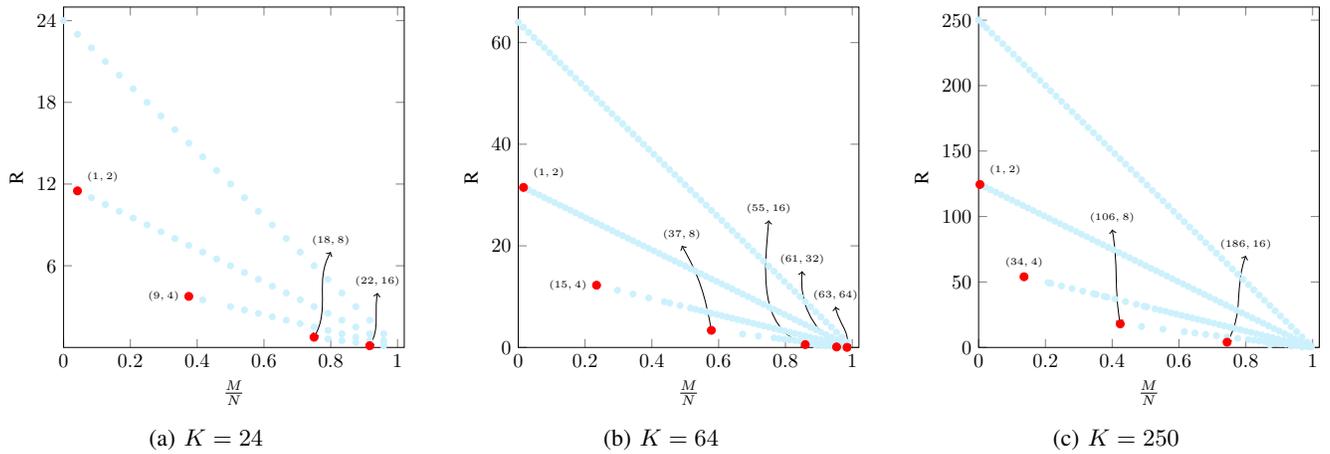
\begin{figure}[H]
\begin{subfigure}{0.32\textwidth}
    \centering
    \definecolor{mycolor1}{rgb}{0.80000,0.94700,0.99100}%
\begin{tikzpicture}[scale=0.75]

\begin{axis}[%
clip=false,
width=3in,
height=3in,
at={(0.758in,0.481in)},
xmin=0,
xmax=1.02,
ymin=0,
ymax=25,
ytick = {6,12,18,24},
axis background/.style={fill=white},
xlabel = {$\frac{M}{N}$},
ylabel = {R},
legend style={legend cell align=left, align=left, draw=white!15!black}
]
\addplot[only marks, mark=*, mark options={}, mark size=1.5000pt, draw=mycolor1,fill=mycolor1] table[row sep=crcr]{%
x	y\\
0	24\\
0.0416666666666667	11.5\\
0.0416666666666667	23\\
0.0833333333333333	11\\
0.0833333333333333	22\\
0.125	10.5\\
0.125	21\\
0.166666666666667	10\\
0.166666666666667	20\\
0.208333333333333	9.5\\
0.208333333333333	19\\
0.25	9\\
0.25	18\\
0.291666666666667	8.5\\
0.291666666666667	17\\
0.333333333333333	8\\
0.333333333333333	16\\
0.375	3.75\\
0.375	7.5\\
0.375	15\\
0.416666666666667	3.5\\
0.416666666666667	7\\
0.416666666666667	14\\
0.458333333333333	6.5\\
0.458333333333333	13\\
0.5	3\\
0.5	6\\
0.5	12\\
0.541666666666667	2.75\\
0.541666666666667	5.5\\
0.541666666666667	11\\
0.583333333333333	2.5\\
0.583333333333333	5\\
0.583333333333333	10\\
0.625	2.25\\
0.625	4.5\\
0.625	9\\
0.666666666666667	2\\
0.666666666666667	4\\
0.666666666666667	8\\
0.708333333333333	1.75\\
0.708333333333333	3.5\\
0.708333333333333	7\\
0.75	0.75\\
0.75	1.5\\
0.75	3\\
0.75	6\\
0.791666666666667	0.625\\
0.791666666666667	1.25\\
0.791666666666667	2.5\\
0.791666666666667	5\\
0.833333333333333	0.5\\
0.833333333333333	1\\
0.833333333333333	2\\
0.833333333333333	4\\
0.875	0.375\\
0.875	0.75\\
0.875	1.5\\
0.875	3\\
0.916666666666667	0.125\\
0.916666666666667	0.25\\
0.916666666666667	0.5\\
0.916666666666667	1\\
0.916666666666667	2\\
0.958333333333333	0.125\\
0.958333333333333	0.25\\
0.958333333333333	0.5\\
0.958333333333333	1\\
};
\node[label={[outer sep=-2pt]45:\tiny{$(1,2)$}}] at (axis cs: 0.0416666666666667,11.5) {} ;
\node[label={[outer sep=-2pt]180:\tiny{$(9,4)$}}] at (axis cs: 0.375, 3.75) {} ;
\path[<->, draw] (axis cs: 0.75, 0.75) to[out = 60, in = 240]
        (axis cs: 0.8, 7) node[above] {\tiny{$(18,8)$}};
\path[<->, draw] (axis cs: 0.916666666666667, 0.125) to[out = 60, in = 260]
        (axis cs: 0.94, 4) node[above] {\tiny{$(22,16)$}};
\addplot[only marks, mark options={solid,draw=red,fill=red}]
    coordinates {
    (0.0416666666666667,11.5) (0.375, 3.75) (0.75, 0.75) (0.916666666666667, 0.125) 
    };
\end{axis}
\end{tikzpicture}%
    \caption{$K=24$}
\end{subfigure}%
~
\begin{subfigure}{0.32\textwidth}
    \centering
    \definecolor{mycolor1}{rgb}{0.80000,0.94700,0.99100}%
\begin{tikzpicture}[scale=0.75]

\begin{axis}[%
clip=false,
width=3in,
height=3in,
at={(0.758in,0.481in)},
xmin=0,
xmax=1.02,
ymin=0,
ymax=67,
axis background/.style={fill=white},
xlabel = {$\frac{M}{N}$},
ylabel = {R},
legend style={legend cell align=left, align=left, draw=white!15!black}
]
\addplot[only marks, mark=*, mark options={}, mark size=1.5000pt, draw=mycolor1, fill=mycolor1] table[row sep=crcr]{%
x	y\\
0	64\\
0.015625	31.5\\
0.015625	63\\
0.03125	31\\
0.03125	62\\
0.046875	30.5\\
0.046875	61\\
0.0625	30\\
0.0625	60\\
0.078125	29.5\\
0.078125	59\\
0.09375	29\\
0.09375	58\\
0.109375	28.5\\
0.109375	57\\
0.125	28\\
0.125	56\\
0.140625	27.5\\
0.140625	55\\
0.15625	27\\
0.15625	54\\
0.171875	26.5\\
0.171875	53\\
0.1875	26\\
0.1875	52\\
0.203125	25.5\\
0.203125	51\\
0.21875	25\\
0.21875	50\\
0.234375	12.25\\
0.234375	24.5\\
0.234375	49\\
0.25	24\\
0.25	48\\
0.265625	23.5\\
0.265625	47\\
0.28125	23\\
0.28125	46\\
0.296875	11.25\\
0.296875	22.5\\
0.296875	45\\
0.3125	22\\
0.3125	44\\
0.328125	21.5\\
0.328125	43\\
0.34375	10.5\\
0.34375	21\\
0.34375	42\\
0.359375	20.5\\
0.359375	41\\
0.375	20\\
0.375	40\\
0.390625	9.75\\
0.390625	19.5\\
0.390625	39\\
0.40625	19\\
0.40625	38\\
0.421875	18.5\\
0.421875	37\\
0.4375	9\\
0.4375	18\\
0.4375	36\\
0.453125	8.75\\
0.453125	17.5\\
0.453125	35\\
0.46875	17\\
0.46875	34\\
0.484375	8.25\\
0.484375	16.5\\
0.484375	33\\
0.5	16\\
0.5	32\\
0.515625	7.75\\
0.515625	15.5\\
0.515625	31\\
0.53125	7.5\\
0.53125	15\\
0.53125	30\\
0.546875	7.25\\
0.546875	14.5\\
0.546875	29\\
0.5625	7\\
0.5625	14\\
0.5625	28\\
0.578125	3.375\\
0.578125	6.75\\
0.578125	13.5\\
0.578125	27\\
0.59375	6.5\\
0.59375	13\\
0.59375	26\\
0.609375	6.25\\
0.609375	12.5\\
0.609375	25\\
0.625	6\\
0.625	12\\
0.625	24\\
0.640625	5.75\\
0.640625	11.5\\
0.640625	23\\
0.65625	5.5\\
0.65625	11\\
0.65625	22\\
0.671875	2.625\\
0.671875	5.25\\
0.671875	10.5\\
0.671875	21\\
0.6875	5\\
0.6875	10\\
0.6875	20\\
0.703125	4.75\\
0.703125	9.5\\
0.703125	19\\
0.71875	2.25\\
0.71875	4.5\\
0.71875	9\\
0.71875	18\\
0.734375	4.25\\
0.734375	8.5\\
0.734375	17\\
0.75	4\\
0.75	8\\
0.75	16\\
0.765625	1.875\\
0.765625	3.75\\
0.765625	7.5\\
0.765625	15\\
0.78125	1.75\\
0.78125	3.5\\
0.78125	7\\
0.78125	14\\
0.796875	1.625\\
0.796875	3.25\\
0.796875	6.5\\
0.796875	13\\
0.8125	1.5\\
0.8125	3\\
0.8125	6\\
0.8125	12\\
0.828125	1.375\\
0.828125	2.75\\
0.828125	5.5\\
0.828125	11\\
0.84375	1.25\\
0.84375	2.5\\
0.84375	5\\
0.84375	10\\
0.859375	0.5625\\
0.859375	1.125\\
0.859375	2.25\\
0.859375	4.5\\
0.859375	9\\
0.875	1\\
0.875	2\\
0.875	4\\
0.875	8\\
0.890625	0.4375\\
0.890625	0.875\\
0.890625	1.75\\
0.890625	3.5\\
0.890625	7\\
0.90625	0.375\\
0.90625	0.75\\
0.90625	1.5\\
0.90625	3\\
0.90625	6\\
0.921875	0.3125\\
0.921875	0.625\\
0.921875	1.25\\
0.921875	2.5\\
0.921875	5\\
0.9375	0.25\\
0.9375	0.5\\
0.9375	1\\
0.9375	2\\
0.9375	4\\
0.953125	0.09375\\
0.953125	0.1875\\
0.953125	0.375\\
0.953125	0.75\\
0.953125	1.5\\
0.953125	3\\
0.96875	0.0625\\
0.96875	0.125\\
0.96875	0.25\\
0.96875	0.5\\
0.96875	1\\
0.96875	2\\
0.984375	0.015625\\
0.984375	0.03125\\
0.984375	0.0625\\
0.984375	0.125\\
0.984375	0.25\\
0.984375	0.5\\
0.984375	1\\
};
\node[label={[outer sep=-2pt]45:\tiny{$(1,2)$}}] at (axis cs: 0.015625,31.5) {} ;
\node[label={[outer sep=-2pt]180:\tiny{$(15,4)$}}] at (axis cs: 0.234375,12.25) {} ;
\path[->, draw] (axis cs: 0.578125,3.375) to[out = 100, in = 300]
        (axis cs: 0.49, 20) node[above] {\tiny{$(37,8)$}};
\path[<->, draw] (axis cs: 0.859375,0.5625) to[out = 160, in = 265]
        (axis cs: 0.75, 25) node[above] {\tiny{$(55,16)$}};
\path[<->, draw] (axis cs: 0.953125,0.09375) to[out = 140, in = 265]
        (axis cs: 0.85, 15) node[above] {\tiny{$(61,32)$}};
\path[<->, draw] (axis cs: 0.984375,0.015625) to[out = 80, in = 280]
        (axis cs: .95, 8) node[above] {\tiny{$(63,64)$}};
\addplot[only marks, mark options={solid,draw=red,fill=red}]
    coordinates {
    (0.015625,31.5) (0.234375,12.25) (0.578125,3.37) (0.859375,0.5625) (0.953125,0.09375) (0.984375,0.015625)
    };

\end{axis}
\end{tikzpicture}%
    \caption{$K=64$}
\end{subfigure}%
~
\begin{subfigure}{0.32\textwidth}
    \centering
    \definecolor{mycolor1}{rgb}{0.80000,0.94700,0.99100}%
\begin{tikzpicture}[scale=0.75]

\begin{axis}[%
clip=false,
width=3in,
height=3in,
at={(0.758in,0.481in)},
xmin=0,
xmax=1.02,
ymin=0,
ymax=260,
axis background/.style={fill=white},
xlabel = {$\frac{M}{N}$},
ylabel = {R},
legend style={legend cell align=left, align=left, draw=white!15!black}
]
\addplot[only marks, mark=*, mark options={}, mark size=1.5000pt, draw=mycolor1, fill=mycolor1] table[row sep=crcr]{%
x	y\\
0.996	1\\
0.984	4\\
0.968	8\\
0.952	12\\
0.936	16\\
0.92	20\\
0.904	24\\
0.888	28\\
0.872	32\\
0.856	36\\
0.84	40\\
0.824	44\\
0.808	48\\
0.792	52\\
0.776	56\\
0.76	60\\
0.744	64\\
0.728	68\\
0.712	72\\
0.696	76\\
0.68	80\\
0.664	84\\
0.648	88\\
0.632	92\\
0.616	96\\
0.6	100\\
0.584	104\\
0.568	108\\
0.552	112\\
0.536	116\\
0.52	120\\
0.504	124\\
0.488	128\\
0.472	132\\
0.456	136\\
0.44	140\\
0.424	144\\
0.408	148\\
0.392	152\\
0.376	156\\
0.36	160\\
0.344	164\\
0.328	168\\
0.312	172\\
0.296	176\\
0.28	180\\
0.264	184\\
0.248	188\\
0.232	192\\
0.216	196\\
0.2	200\\
0.184	204\\
0.168	208\\
0.152	212\\
0.136	216\\
0.12	220\\
0.104	224\\
0.088	228\\
0.072	232\\
0.056	236\\
0.04	240\\
0.024	244\\
0.012	247\\
0	250\\
0.996	0.5\\
0.98	2.5\\
0.964	4.5\\
0.948	6.5\\
0.932	8.5\\
0.916	10.5\\
0.9	12.5\\
0.884	14.5\\
0.868	16.5\\
0.852	18.5\\
0.836	20.5\\
0.82	22.5\\
0.804	24.5\\
0.788	26.5\\
0.772	28.5\\
0.756	30.5\\
0.74	32.5\\
0.724	34.5\\
0.708	36.5\\
0.692	38.5\\
0.676	40.5\\
0.66	42.5\\
0.644	44.5\\
0.628	46.5\\
0.612	48.5\\
0.596	50.5\\
0.58	52.5\\
0.564	54.5\\
0.548	56.5\\
0.532	58.5\\
0.516	60.5\\
0.5	62.5\\
0.484	64.5\\
0.468	66.5\\
0.452	68.5\\
0.436	70.5\\
0.42	72.5\\
0.404	74.5\\
0.388	76.5\\
0.372	78.5\\
0.356	80.5\\
0.34	82.5\\
0.324	84.5\\
0.308	86.5\\
0.292	88.5\\
0.276	90.5\\
0.26	92.5\\
0.244	94.5\\
0.228	96.5\\
0.212	98.5\\
0.196	100.5\\
0.18	102.5\\
0.164	104.5\\
0.148	106.5\\
0.132	108.5\\
0.116	110.5\\
0.1	112.5\\
0.084	114.5\\
0.068	116.5\\
0.052	118.5\\
0.036	120.5\\
0.02	122.5\\
0.004	124.5\\
0.992	0.5\\
0.976	1.5\\
0.96	2.5\\
0.944	3.5\\
0.928	4.5\\
0.912	5.5\\
0.896	6.5\\
0.88	7.5\\
0.864	8.5\\
0.848	9.5\\
0.832	10.5\\
0.816	11.5\\
0.8	12.5\\
0.784	13.5\\
0.768	14.5\\
0.752	15.5\\
0.736	16.5\\
0.72	17.5\\
0.704	18.5\\
0.688	19.5\\
0.672	20.5\\
0.656	21.5\\
0.64	22.5\\
0.624	23.5\\
0.608	24.5\\
0.592	25.5\\
0.576	26.5\\
0.56	27.5\\
0.544	28.5\\
0.528	29.5\\
0.504	31\\
0.472	33\\
0.44	35\\
0.408	37\\
0.376	39\\
0.344	41\\
0.312	43\\
0.28	45\\
0.248	47\\
0.216	49\\
0.496	31.5\\
0.352	40.5\\
0.384	38.5\\
0.208	49.5\\
0.136	54\\
0.984	0.5\\
0.968	1\\
0.936	2\\
0.952	1.5\\
0.904	3\\
0.872	4\\
0.92	2.5\\
0.84	5\\
0.808	6\\
0.888	3.5\\
0.776	7\\
0.744	8\\
0.856	4.5\\
0.712	9\\
0.68	10\\
0.824	5.5\\
0.648	11\\
0.616	12\\
0.552	14\\
0.488	16\\
0.424	18\\
0.968	0.5\\
0.936	1\\
0.872	2\\
0.744	4\\
};
\node[label={[outer sep=-2pt]45:\tiny{$(1,2)$}}] at (axis cs: 0.004,124.5) {} ;
\node[label={[outer sep=-2pt]90:\tiny{$(34,4)$}}] at (axis cs: 0.136,54) {} ;
\path[<->, draw] (axis cs: 0.424,18) to[out = 100, in = 280]
        (axis cs: 0.4, 90) node[above] {\tiny{$(106,8)$}};
\path[<->, draw] (axis cs: 0.744,4) to[out = 60, in = 250]
        (axis cs: 0.8, 70) node[above] {\tiny{$(186,16)$}};
\addplot[only marks, mark options={solid,draw=red,fill=red}]
    coordinates {
    (0.004,124.5) (0.136,54) (0.424,18) (0.744,4) 
    };

\end{axis}
\end{tikzpicture}%
    \caption{$K=250$}
\end{subfigure}
        
\caption{Memory-rate tradeoff for $K=24,64,250$ obtained using Lemma~\ref{lm:recTiling}. The points related to PDAs with minimum $Z$ for a coding gain $g$ are highlighted in red and are labelled using $(Z,g)$.}
\label{fig:tiling24}
\end{figure}

\end{example}
\subsection{General lifting}
In basic lifting, every $\x$ in the base PDA is replaced with the all-$\x$ array. We next consider a generalized version of basic lifting by allowing more general PDAs to replace a $\x$. For coded caching schemes with low cache memory, since we need the number of $\x$s in each column in the lifted PDA to be low, replacing $\x$ in the base PDA with non-trivial PDAs is beneficial. However, to ensure that the Blackburn property for the lifted PDA is not violated, the constituent PDAs that are used to replace the integers and $\x$s need to satisfy some additional constraints.
We introduce the notion of Blackburn compatibility of PDAs to capture such constraints on the constituent PDAs.

\subsubsection{Blackburn compatibility}
Two $n\times n$ PDAs $P_0=[p^{(0)}_{ij}]$ and $P_1=[p^{(1)}_{ij}]$ are said to be \emph{Blackburn-compatible} with respect to (w.r.t.) a third $n\times n$ PDA $P_{\x}=[p^{({\x})}_{ij}]$ if, whenever  $p^{(0)}_{i_1j_1}{=}p^{(1)}_{i_2j_2}{\ne}\Asterisk$, we have  $p^{({\x})}_{i_1j_2}{=}p^{({\x})}_{i_2j_1}=\Asterisk$. In other words, if two entries in $P_0$ and $P_1$ are a common integer $s$, the \emph{mirrored} locations of $s$ in $P_*$ are $\x$s. For $g\ge2$, we say $P_0,\ldots,P_{g-1}$ are Blackburn compatible w.r.t. $P_*$ when they are pairwise Blackburn compatible with $P_*$. 

To see the connection between Blackburn compatibility and lifting, consider an integer $s$ occurring $g$ times in a base PDA. The rows and columns containing $s$ in the PDA, after permutations, can be rearranged into the PDA $I_g(s)$. So, any valid lifting of the base PDA needs to necessarily include a valid lifting of $I_g(s)$. Validity of a certain lifting of $I_g(s)$ and Blackburn compatibility are shown to be equivalent in the following lemma. 
\begin{lemma}[{Equivalence between Blackburn-compatibility and lifting}]
Suppose $P_*,P_0,\ldots,P_{g-1}$ are PDAs of the same size. Let $P_*^{(i,j)}$ for $i,j=0,1,2,\ldots$ be copies of $P_*$ containing integers that are disjoint from each other and from the integers in $P_0,\ldots,P_{g-1}$. Then, the set $\mathcal{P}=\{P_0,\ldots,P_{g-1}\}$ is a set of PDAs Blackburn-compatible w.r.t. $P_*$ if and only if the following lifting of $I_g$ is a valid PDA. 
$$L_{\mathcal{P},P_*}(I_g)\triangleq\begin{pmatrix}
    P_0&P_*^{(0,1)}&\cdots&P_*^{(0,g-1)}\\
    P_*^{(1,0)}&P_1&\cdots&P_*^{(1,g-1)}\\
    \vdots&\vdots&\ddots&\vdots\\
    P_*^{(g-1,0)}&P_*^{(g-1,1)}&\cdots&P_{g-1}
\end{pmatrix},$$ 
\end{lemma}
\begin{proof}
The conditions for validity of the above PDA and the definition of Blackburn compatibility are readily seen to be equivalent. Because of the disjointness properties of the integers in $P_*^{(i,j)}$, no additional conditions arise.
\end{proof}
The above lemma, beyond establishing the connection between lifting and Blackburn compatibility, provides a way to visualize the mirrored locations and aids in constructions of Blackburn-compatible PDAs.  

We see that any two PDAs are Blackburn-compatible w.r.t. the trivial all-$\x$ PDA. If $P_*$ is not all $\x$, the Blackburn compatibility needs to be established more carefully. Before presenting tests for Blackburn compatibility and general constructions, we show some illustrative examples. 
\begin{example}
$P_0=\begin{pmatrix}
    0&1\\
    3&4
\end{pmatrix}$, $P_1=\begin{pmatrix}
    4&1\\
    3&0
\end{pmatrix}$ are Blackburn-compatible w.r.t. $I_2(t)=\begin{pmatrix}
    t&\x\\
    \x&t
\end{pmatrix}$ for $t\notin\{0,1,3,4\}$. We see that
\begin{align*}
L_{\{P_0,P_1\},I2}(I_2)=\left(\begin{array}{cc:cc}
    0&1&2&\x\\
    3&4&\x&2\\ \hdashline
    5&\x&4&1\\
    \x&5&3&0
\end{array}\right)
\end{align*}
is the $(4,4,1,6)$ 2-PDA that we have earlier denoted as $G_4([6])$.
\end{example}
\begin{example}
$P_0=\begin{pmatrix}
    \x&0&2\\
    0&\x&1\\
    3&1&\x
\end{pmatrix}$, $P_1=\begin{pmatrix}
    1&2&\x\\
    3&\x&2\\
    \x&3&0
\end{pmatrix}$ are Blackburn-compatible w.r.t. $I_3(t)$ for $t\notin[4]$. We see that
\begin{align*}
L_{\{P_0,P_1\},I_3}(I_2)=\left(\begin{array}{ccc:ccc}
    \x&0&2&4&\x&\x\\
    0&\x&1&\x&4&\x\\
    3&1&\x&\x&\x&4\\\hdashline
    5&\x&\x&1&2&\x\\
    \x&5&\x&3&\x&2\\
    \x&\x&5&\x&3&0
\end{array}\right)
\end{align*}
is a $(6,6,3,6)$ 3-PDA.
\end{example}
\begin{example}
Consider the following rectangular arrays.
\begin{align}
    P_0 = \begin{pmatrix}
        0 & \x & \x & 1 & 6\\
        2 & 1 & \x & \x & 7\\
        \x & 5 & 2 & \x & 3\\
        \x & 4 & 9 & 3 & \x\\
        \x & \x & 0 & 5 & 4\\
        \x & \x & 1 & \x & 8\\
        9 & \x & \x & 7 & \x\\
        3 & 8 & \x & \x & \x\\
        4 & \x & 6 & \x & \x\\
        \x & 0 & \x & 2 & \x        
    \end{pmatrix},
    P_1 = \begin{pmatrix}
        \x & \x & \x & 1 & 6\\
        2 & \x & \x & \x & 7\\
        \x & 5 & \x & \x & 3\\
        \x & 4 & 9 & \x & \x\\
        \x & \x & 0 & 5 & \x\\
        5 & \x & 1 & \x & 8\\
        9 & 6 & \x & 7 & \x\\
        3 & 8 & 7 & \x & \x\\
        4 & \x & 6 & 8 & \x\\
        \x & 0 & \x & 2 & 9        
    \end{pmatrix},
    P_{\x}(\{s,t\}) = \begin{pmatrix}
        s & \x & \x & \x & \x \\
        \x & s & \x & \x & \x \\
        \x & \x & s & \x & \x \\
        \x & \x & \x & s & \x \\
        \x & \x & \x & \x & s \\
        t & \x & \x & \x & \x \\
        \x & t & \x & \x & \x \\
        \x & \x & t & \x & \x \\
        \x & \x & \x & t & \x \\
        \x & \x & \x & \x & t \\
    \end{pmatrix}.
\label{eq:10by5}
\end{align}
We can check that $P_0$ and $P_1$ are Blackburn-compatible w.r.t. $P_*(\{s,t\})$ for $s,t\notin[10]$ by confirming that all mirrored locations of integers 0 to 9 are $\x$s. Also, $$L_{\{P_0,P_1\},P_*}(I_2)=\begin{pmatrix}
    P_0&P_*(\{10,11\})\\
    P_*(\{12,13\})&P_1
\end{pmatrix}$$
is a $(10,20,13,14)$ 5-PDA. 
\end{example}
\begin{example}
Given,
\begin{align}
     P_*(\{t_0,t_1,t_2,t_3\})=\begin{pmatrix}
         I_3(t_0)\\
         I_3(t_1)\\
         I_3(t_2)\\
         I_3(t_3)
     \end{pmatrix}, P_0=\begin{pmatrix}
         J_3(0:8)\\
         H_{3,0}(9,13,17)\\
         H_{3,1}(10,14,15)\\
         H_{3,2}(11,12,16)
     \end{pmatrix}, P_1=\begin{pmatrix}
         J_3(9:17)\\
         H_{3,0}(0,4,8)\\
         H_{3,1}(1,5,6)\\
         H_{3,2}(2,3,7)
     \end{pmatrix}
\label{eq:12by3}
\end{align}
$P_0$ and $P_1$ are Blackburn-compatible w.r.t. $P_*(\{t_0,t_1,t_2,t_3\})$ for $t_0,t_1,t_2,t_3 \notin [18]$ and
$$L_{\{P_0,P_1\},P_*}(I_2)=\begin{pmatrix}
    P_0&P_*(\{18,19,20,21\})\\
    P_*(\{22,23,24,25\})&P_1
\end{pmatrix}$$
is a $(6,24,11,26)$ 3-PDA.
\end{example}
The above examples are for two PDAs Blackburn-compatible w.r.t. a third non-trivial PDA. When $P_*$ is all-$\x$, an arbitrary number of copies of a PDA, $\{P,P,\ldots\}$, are Blackburn-compatible w.r.t. the all-$\x$ array, and this is used in the basic lifting of Theorem \ref{th:basicTiling}. However, for a general $P_*$, we require integer-disjoint copying of PDAs to ensure Blackburn compatibility. Since this is a repeatedly occurring step in constructions, we record it as a lemma.
\begin{lemma}[{Replication of Blackburn-comptabible PDAs}]\label{lm:bcReplicate}
Given a set $\mathcal{P}$ of $b$ PDAs Blackburn-compatible w.r.t. $P_*$, the set of $mb$ PDAs formed by $m$ integer-disjoint copies of the PDAs in $\mathcal{P}$ is Blackburn-compatible w.r.t. $P_*$, for any integer $m > 0$.
\end{lemma}
\begin{proof}
The proof is immediate by the disjointness of the integers.
\end{proof}
The use of integer-disjoint copying in the above lemma results in a multiplicative increase in the number of integers with every replication. If a decrease in the number of integers is desirable (to change the memory-rate trade-off), other methods have to be considered.
\subsubsection{General lifting theorem}
Using Blackburn-compatible PDAs, a generalization of basic lifting is presented in the following theorem.
\begin{theorem}[General lifting]\label{th:generalTiling}
Let $P_b$ be a $(K,f,Z_b,\mathcal{S}_b)$ PDA. Let $g_s$ be the frequency of integer $s$ in $P_b$.
For $s\in \mathcal{S}_b$ and $t\in\{1,\ldots,g_s\}$, let $P_{s,t}$ be an $(m,n,Z_c,\mathcal{S}_{s,t})$ PDA such that for any $s$, $P_{s,1}, \ldots , P_{s,g_s}$ are Blackburn compatible w.r.t. an $(m,n,Z_*,\mathcal{S}_*)$ PDA $P_*$, and, for distinct integers $s,s'\in\mathcal{S}_b$, $\mathcal{S}_{s,t}$ and $\mathcal{S}_{s',t'}$ are disjoint.
Let $P_{*,r}$, $r\in[KZ_b]$, be integer-disjoint copies of $P_*$, which are integer-disjoint with $P_{s,t}$ as well.
Let an array $P$ be defined as follows:
\begin{enumerate}
    \item $r$-th $\x$ in $P_b$ is replaced by $P_{\x,r}$ for $r\in[KZ_b]$.
    \item $t$-th occurrence of integer $s\in\mathcal{S}_b$ in $P_b$ is replaced by $P_{s,t}$ for $t=1,\ldots,g_s$.
\end{enumerate}
Then, $P$ is a $(Km,fn,Z_bZ_*+(f-Z_b)Z_c,\mathcal{S})$ PDA, where $\mathcal{S}=\left(\bigcup_{r\in[KZ_b]}\mathcal{S}_{*,r}\right)\bigcup\left(\bigcup_{s\in\mathcal{S}_b}\bigcup_{t\in[g_s]}\mathcal{S}_{s,t}\right)$.
\end{theorem}
\begin{proof}
Clearly, $P$ is an $fn\times Km$ array. Each column of $P_b$ has $Z_b$ $\x$'s and $f-Z_b$ integers. So, each column of $P$ has $Z_bZ_*+(f-Z_b)Z_c$ $\x$'s satisfying \ref{cond:equalZ}. 

Since $P_b$ is a PDA, $P$ has all integers in $\mathcal{S}$ occurring at least once, satisfying \ref{cond:everysOnce}.

Finally, we need to verify the Blackburn property for $P$. Let $p_{i,j}$ and $q_{i,j}$ denote the $(i,j)$-th elements of $P_b$ and $P$, respectively. Let $i/n$ denote the quotient when $i$ is divided by $n$, and let $j/m$ be defined similarly. If $q_{i_1,j_1}=q_{i_2,j_2}=s$ in the lifted PDA $P$, we necessarily have that $p_{i_1/n,j_1/m}=p_{i_2/n,j_2/m}$ in the base PDA $P_b$ because different integers are expanded to PDAs containing disjoint sets of integers. So, we have $p_{i_1/n,j_2/m}=p_{i_2/n,j_1/m}=\Asterisk$ by the Blackburn property. Since an $\x$
 is replaced by $P_{0,t}$ and an entry in $P_{0,t}$ is $\x$ iff the corresponding entry in $P_{0}$ is $\x$, $q_{i_1,j_2}=q_{i_2,j_1}=\Asterisk$ in $P$ due to the Blackburn compatibility of the PDAs replacing an integer w.r.t $P_{\x}$. Hence, the Blackburn property (\ref{cond:blackburn}) is satisfied.
\end{proof}
Clearly, basic lifting is a special case of general lifting, where $P_*$ is the all-$\x$ array and $P_{s,t}$ are arbitrary PDAs. However, if a non-trivial $P_*$ is to be used, then we require as many Blackburn-compatible PDAs w.r.t. $P_*$ as the largest integer frequency of the base PDA. We will see general methods to construct Blackburn-compatible PDAs in the next sections.

To obtain regular lifted PDAs, the constituent PDAs and the base PDAs will need to be chosen more carefully. The following corollary of the above general lifting theorem presents a sufficient condition for regular lifting using Blackburn compatibility.
\begin{corollary}[{General regular lifting}]\label{co:regularTiling}
Let $P_b$ be a $(K_b,f_b,Z_b,K_b(f_b-Z_b)/g_b)$ $g_b$-PDA. Let $P_*$ be an $(m,n,Z_*,m(n-Z_*)/g)$ $g$-PDA. Let  $\mathcal{P}=\{P_0, P_1, \ldots,P_{g_b-1}\}$ be a set of 
$n\times m$ PDAs satisfying the following conditions: 
\begin{itemize}
    \item the number of $\x$s in every column of every $P_i$ is equal to $e$,
    \item an integer occurring in any one $P_i$ occurs a total of $g$ times across all $P_i$'s,
    \item $\mathcal{P}$ is Blackburn-compatible w.r.t. $P_*$.
\end{itemize}
For $s\in[K(f-Z_b)/g_b]$, let $\mathcal{P}_s=\{P_{s,0},\ldots,P_{s,g_b-1}\}$ be an integer-disjoint copy of $\mathcal{P}$. Let the lifting of $P_b$ using Theorem \ref{th:generalTiling} with $P_{s,t}$ as constituent PDAs be denoted $L_{\mathcal{P},P_*}(P_b)$.

Then, $L_{\mathcal{P},P_*}(P_b)$ is a $g$-regular $(K_bm,f_bn,Z_bZ_*+(f_b-Z_b)e,\frac{K_bm(f_bn-Z_bZ_*-(f_b-Z_b)e)}{g})$ PDA.
\label{cor:genreg}
\end{corollary}
\begin{proof}
Only the parameters of $L_{\mathcal{P},P_*}(P_b)$ need to be established. Clearly, $L_{\mathcal{P},P_*}(P_b)$ is a $f_bn\times K_bm$ array. An integer in $P_b$ occurs $g_b$ times and the $i$-th occurrence is replaced by $P_i$. So, the number of times an integer occurs in $L_{\mathcal{P},P_*}(P_b)$ is $g$. Every column of $P_b$ contains $Z_b$ $\x$s and $f_b-Z_b$ integers. Since $\x$ is replaced by a copy of $P_*$ and an integer is replaced by a copy of $P_i$, the number of $\x$s per column after lifting is $Z_bZ_*+(f_b-Z_b)e$. 
\end{proof}

A crucial requirement for lifting with $P_*$ not being all-$\x$ is sets of Blackburn-compatible PDAs. Integer-disjoint copying is one simple method for constructing any number of Blackburn-compatible PDAs w.r.t. any $P_*$. For specific choices of $P_*$, other methods of construction could improve upon integer-disjoint copying, and we consider such constructions next.

\section{Constructions of Blackburn-compatible PDAs}\label{sec:blackburn}
We will present some general constructions for Blackburn-compatible PDAs using the following ideas - (1) permutation of integer elements and blocks to ensure mirrored locations are $\x$s, (2) tiling of identity/regular PDAs, (3) recursive methods, and (4) a randomized construction. 

{One strategy to construct Blackburn-compatible PDAs is to use existing PDAs from Section~\ref{sec:notation} for $P_*$ and $P_0$ and obtain the rest of $P_i$'s by transforming $P_0$ using permutations, transpose etc.}
{In Corollary~\ref{cor:genreg}, when $P_*$ is $g$-regular and each $P_i\in\mathcal{P}$ is $g_i$-regular with same set of integers, we have $g=g_bg_i$.
One natural choice for $P_*$ is $I_g(t)$ since it is a linear PDA with high coding gain.}
In most cases, we will consider the choice of $P_*$ as $I_g(t)$.
The following lemma provides a test for Blackburn compatibility of PDAs with respect to $I_g(t)$.
\begin{lemma}[Test for compatibility w.r.t. $I_g(t)$] \label{lm:InCompatible}
Two $g\times g$ PDAs $P_0$ and $P_1$ are Blackburn compatible with $I_g(t)$ for $t$ not appearing in $P_0$ or $P_1$ iff $p^{(0)}_{i_0j_0}{=}p^{(1)}_{i_1j_1}{\ne}\Asterisk$ implies $i_0 \ne j_1$ and $i_1 \ne j_0$. In words, mirrored locations of integers should be off-diagonal.
\end{lemma}
\begin{proof}
If $p^{(0)}_{i_0j_0}{=}p^{(1)}_{i_1j_1}{\ne}\Asterisk$ implies $i_0 \ne j_1$ and $i_1 \ne j_0$, then for no $i\in [g]$ we need $p_{i,i}^{(0)}=\x$. Hence $P_0$ and $P_1$ are Blackburn compatible w.r.t. any $g\times g$ PDA which has integers only in its diagonal. Now, let $P_0$ and $P_1$ are Blackburn compatible w.r.t. $I_g(t)$. Assume that $\exists s\in \mathcal{S}_0\cap \mathcal{S}_1$ such that $p^{(0)}_{i_0j_0}{=}p^{(1)}_{i_1j_1}=s$ and $i_0 = j_1$. This implies that cell $(i_0,i_0)$ of $I_g(t)$ is $\x$, But this is a contradiction. Hence $i_0 \ne j_1$. Similarly we can prove that $i_1 \ne j_0$.
\end{proof}

\subsection{Permutation constructions} 
{Now we introduce two permutation operations to obtain $P_i$'s for $i>0$ from $P_0$ when $P_0$ is either a $1$-PDA or a $2$-PDA as defined in Section~\ref{sec:notation}.
These $P_i$'s will be Blackburn compatible w.r.t. $I_g(t)$.}
The first permutation construction uses \emph{cyclic rotation of diagonal or anti-diagonal elements}, for which, we need the following notation. Given an $n\times n$ PDA $P=[p_{ij}]$, where $i$ and $j$ take values from $0$ to $n-1$, a PDA $\pi_{D,1}(P)$ is defined as
$$\pi_{D,1}(P)=\begin{pmatrix}
p_{n-1,n-1}&p_{0,1}&\cdots&p_{0,n-1}\\
p_{1,0}&p_{0,0}&\cdots&p_{1,n-1}\\
\vdots&\vdots&\ddots&\vdots\\
p_{n-1,0}&\cdots&\cdots&p_{n-2,n-2}
\end{pmatrix}.$$
Basically, $\pi_{D,1}(P)$ is identical to $P$ except for the diagonal entries, which are circularly shifted down by one position. For an integer $l$, $\pi^l_{D,1}(P)$ denotes the PDA obtained by $l$ applications of $\pi_{D,1}$ on $P$. For negative $l$, the diagonal entries are shifted up $l$ times.  
A similar cyclic rotation of anti-diagonal elements in $P$ is denoted by $\pi_{AD,1}(P)$, where $(i,j)$ on the anti-diagonal goes to $(i-1,j+1)\mod n$ and all other locations are retained. 

Given a $2n\times 2n$ PDA $P=\begin{pmatrix}P_{11}&P_{12}\\P_{21}&P_{22}\end{pmatrix}$, where $P_{ij}$ are $n\times n$ blocks, a PDA $\pi_{D,2}(P)$ is defined as
$$\pi_{D,2}(P)=\begin{pmatrix}
\pi_{D,1}(P_{11})&P_{12}\\
P_{21}&\pi^{-1}_{D,1}(P_{22})
\end{pmatrix}.$$
$\pi_{D,2}(P)$ is identical to $P$ except for the $2n$ diagonal entries - the first $n$ are circularly shifted down by one position, and the second $n$ are circularly shifted up by 1 position. A similar cyclic rotation of anti-diagonal elements in $P$ is denoted by $\pi_{AD,2}(P)$.

\begin{lemma}[Cyclic rotation]\label{lm:bc12pdas}
\begin{enumerate}
    \item (Construction C1) 
    \begin{enumerate}
        \item Given a 1-PDA $P$, $\{P,\pi_{D,1}(P)\}$ is Blackburn-compatible w.r.t. $I_g(t)$, and $\{P,\pi_{AD,1}(P)\}$ w.r.t. $\tilde{I}_g(t)$. 
        \item Letting $P=J_g([g^2])$, $\mathcal{P}_D=\{P,\pi_{D,1}(P),\ldots,\pi_{D,1}^{g-1}(P)\}$ is a set of $g$ $1$-PDAs Blackburn compatible w.r.t $I_g(t)$, and $\mathcal{P}_{AD}=\{P,\pi_{AD,1}(P),\ldots,\pi_{AD,1}^{g-1}(P)\}$ w.r.t. $\tilde{I}_g(t)$. 
        \item $L_{\mathcal{P}_D,I_g}(I_g)$ and $L_{\mathcal{P}_{AD},\tilde{I}_g}(I_g)$ are $g$-regular $(g^2,g^2,g^2-2g+1,g(2g-1))$ PDAs.
    \end{enumerate}
    \item (Construction C2) 
    \begin{enumerate}
        \item Letting $Q_0=G_{2g}([g(2g-1)])$ and $Q_1=H_{2g}([g(2g-1)])$, $\{Q_0,\pi_{D,2}(Q_0)\}$ is Blackburn-compatible w.r.t. $I_{2g}(t)$, and $\{Q_1,\pi_{AD,2}(Q)\}$ w.r.t. $\tilde{I}_{2g}(t)$. 
        \item Then $\mathcal{Q}_D=\{Q_0,\pi_{D,2}(Q_0),\ldots,\pi_{D,2}^{g-1}(Q_0)\}$ is a set of $g$ 2-PDAs Blackburn compatible w.r.t $I_{2g}(t)$, and $\mathcal{Q}_{AD}=\{Q_1,\pi_{AD,2}(Q_1),\ldots,\pi_{AD,2}^{g-1}(Q_1)\}$ w.r.t. $\tilde{I}_{2g}(t)$. 
        \item $L_{\mathcal{Q}_D,I_{2g}}(I_g)$ and $L_{\mathcal{Q}_{AD},\tilde{I}_{2g}}(I_g)$ are $2g$-regular $(2g^2,2g^2,2g^2-3g+2,g(3g-2))$ PDAs.
    \end{enumerate}
\end{enumerate}
The integer $t$ is chosen to be disjoint from the integers in $P$, $Q_0$ or $Q_1$.
\end{lemma}
\begin{proof}
We prove Part 1(a) as follows. Suppose $[P]_{ij}=[\pi_{D,1}(P)]_{ij}\ne\x$ for $i\ne j$ (off-diagonal). Then, the corresponding mirrored $(i,j)$-th entry of $I_g(t)$ for $i\ne j$ is a $\x$. Suppose $[P]_{ii}=[\pi_{D,1}(P)]_{jj}\ne\x$ (diagonal). Then, $i\ne j$ because of the rotation, and the mirrored $(i,j)$-th entry of $I_g(t)$ for $i\ne j$ is a $\x$. The claim for $P$ and $\pi_{AD,1}(P)$ can be proved in a similar fashion. Part 1(b) uses Part 1(a) with $P=J_g([g^2])$. Part 1(c) can be verified by a straight-forward calculation.

For Construction C2, the above proof can be modified in a minor way and we skip the details.
\end{proof}
\begin{example}
For $g=3$, a set of 3 PDAs $\mathcal{P}_3=\{P_1,P_2,P_3\}$ that are pairwise Blackburn-compatible w.r.t. $P_{\x}=I_3(t)$ constructed using C1 is shown below.
\begin{align}
\label{eq:1gP3}
\begin{split}
    P_{\x} = I_3(t) &= \left(
        \begin{array}{ccc}
            t & \x & \x \\
            \x & t & \x \\
            \x & \x & t
        \end{array}
        \right),\,
    P_0 = J_3([9]) = \left(
        \begin{array}{ccc}
            0 & 1 & 2 \\
            3 & 4 & 5 \\
            6 & 7 & 8
        \end{array}
        \right),\,\\
    P_1 &= \left(
        \begin{array}{ccc}
            8 & 1 & 2 \\
            3 & 0 & 5 \\
            6 & 7 & 4
        \end{array}
        \right),\,
    P_2 = \left(
        \begin{array}{ccc}
            4 & 1 & 2 \\
            3 & 8 & 5 \\
            6 & 7 & 0
        \end{array}
        \right).
\end{split}
\end{align}
Integer-disjoint copying would have resulted in 27 integers across $P_0$, $P_1$ and $P_2$, while the above construction uses 9. The lifted PDA $L_{\mathcal{P}_3,I_3}(I_3)$ is a 3-regular $(9,9,4,15)$ PDA.
\end{example}
\begin{example}
For $g=2$, a set of 2 PDAs $\mathcal{Q}_2=\{Q_0,Q_1\}$ that are Blackburn compatible w.r.t. $I_4(t)$ constructed using C2 is shown below.
\begin{align}
    Q_{\x}=I_4(t) = \left(
        \begin{array}{cccc}
            t & \x & \x & \x \\
            \x & t & \x & \x \\
            \x & \x & t & \x \\
            \x & \x & \x & t
        \end{array}
        \right),\,
    Q_0 = \left(
        \begin{array}{cccc}
            0 & 1 & 2 & \x \\
            3 & 4 & \x & 2 \\
            5 & \x & 4 & 1 \\
            \x & 5 & 3 & 0
        \end{array}
        \right),\,
    Q_1 = \left(
        \begin{array}{cccc}
            4 & 1 & 2 & \x \\
            3 & 0 & \x & 2 \\
            5 & \x & 0 & 1 \\
            \x & 5 & 3 & 4
        \end{array}
        \right).
        \label{eq:2gP4}
\end{align}
The lifted PDA $L_{\mathcal{Q}_2,I_4}(I_2)$ is a 4-regular $(8,8,4,8)$ PDA.
\end{example}

Next, we provide two transpose constructions for pairs of 1-PDAs Blackburn-compatible w.r.t. a $P_*$ with $\x$s on one or more diagonals.
\begin{lemma}[Transpose construction]\label{lm:bcTranspose}
For a $g\times g$ PDA $P$, the $i$-th diagonal of $P$ is the set of locations $\{(x,y): y=x+i\text{ mod }g\}$, $i\in[g]$. The $0$-th diagonal is the main diagonal. 
\begin{enumerate}[nosep]
    \item (Construction T1) A $g\times g$ 1-PDA $P$ and its transpose $P^T$ are Blackburn-compatible w.r.t. a $g\times g$ PDA $P_*$ with $\x$ on the diagonal cells.
    \item (Construction T2) For integers $g$ and $i$, a $g\times g$ 1-PDA $P$ and a cyclic permutation of its transpose $\pi^i_{AD,1}(P^T)$ are Blackburn-compatible w.r.t. a $g\times g$ PDA $P_*$ with $\x$ on the main and $i$-th diagonal cells.
\end{enumerate}
The integers in $P_*$ are assumed to be disjoint from those in $P$.
\end{lemma}
\begin{proof}
For proving T1, let $P_1=P$ and $P_2=P_1^T$. Since $p^{(1)}_{i_1j_1}{=}p^{(2)}_{i_2j_2}{\ne}\Asterisk$ implies $i_1 = j_2$ and $i_2 = j_1$ by construction and all diagonal entries in $P_*$ are $\x$, Blackburn compatibility is established.

For proving T2, let $P_1=P$ and $P_2=\pi^i_{AD,1}(P_1^T)$. If $p^{(1)}_{i_1j_1}{=}p^{(2)}_{i_2j_2}{\ne}\Asterisk$, there are two cases to consider. 
\begin{enumerate}
    \item $(i_1,j_1)$ is not on the anti-diagonal, which implies $i_1=j_2$ and $i_2=j_1$. So, the mirrored locations $(i_1,j_2)$ and $(i_2,j_1)$ are main diagonal entries of $P_*$, which are $\x$.
    \item $(i_1,j_1)$ is on the anti-diagonal, which implies $(i_2,j_2)$ is also on the anti-diagonal - transposed and cyclically rotated $i$ times. So, $i_2=j_1-i, j_2=i_1+i$ ($\text{mod } g$). So, the mirrored locations $(i_1,i_1+i\text{ mod }g)$ and $(j_1-i\text{ mod }g,j_1)$ are $i$-th diagonal entries of $P_*$, which are $\x$.
\end{enumerate}
This establishes Blackburn compatibility of $P$ and $\pi^i_{AD,1}(P^T)$ w.r.t. $P_*$ with $\x$ on the main and $i$-th diagonals.
\end{proof}
An example of the above lemma (T1) for $n=3$ is shown below.
\begin{align}
    P_* = \left(
        \begin{array}{ccc}
            \x & t_0 & t_1 \\
            t_0 & \x & t_2 \\
            t_1 & t_2 & \x 
        \end{array}
        \right),\,
    P_0 = \left(
        \begin{array}{ccc}
            0 & 1 & 2  \\
            3 & 4 & 5 \\
            6 & 7 & 8
        \end{array}
        \right),\,
    P_1 = \left(
        \begin{array}{ccc}
            0 & 3 & 6  \\
            1 & 4 & 7 \\
            2 & 5 & 8
        \end{array}
        \right)
        \label{eq:2P3}\\
L_{\mathcal{P},P_*}(I_2) = \begin{array}{c}
     (6,6,1,15) \textrm{ PDA} \\
\begin{sysmatrix}{cccccc}
             6 &  7 &  8 & \x & 0 & 1\\
             9 & 10 & 11 & 0 & \x & 2\\
            12 & 13 & 14 & 1 & 2 & \x\\
            \x & 3 & 4   & 6 & 9 & 12\\
            3 & \x & 5   & 7 & 10 & 13\\
            4 & 5 & \x   & 8 & 11 & 14 
\end{sysmatrix}
\end{array} 
\end{align}

{So far, we have seen constructions where we transform the elements of a PDA to achieve Blackburn compatibility. We will show that transformations can be applied at a block level to achieve the same.}
We now consider the case where $P_*$ is a tiling or block-concatenation of multiple PDAs. Specifically, for integers $g$ and $d$ with $d\mid g$, $P_*^{(g,d)}=B_{\tilde{I}_{2d}}(J_{g/d}([(g/d)^2]))$ (regular basic lifting from Corollary \ref{cor:basicreg}) is a $2g\times 2g$ PDA, which is a $g/d\times g/d$ block-concatenation of $\tilde{I}_{2d}(t)$, $t=1,\ldots,(g/d)^2$. In this case, we construct PDAs Blackburn-compatible w.r.t. $P_*^{(g,d)}$ block-wise with a combination of suitable transpositions and cyclic permutations.
\begin{lemma}[Construction BW1]\label{lm:bw1}
For positive integers $g$ and $d$ such that $d\mid g$, there exist 
a set $\mathcal{P}$ of $d$ $(2g,2g,1,g(2g-1))$ PDAs that are Blackburn compatible w.r.t $P_*^{(g,d)}=B_{\tilde{I}_{2d}}(J_{g/d}([(g/d)^2]))$ 
such that $L_{\mathcal{P},P_*^{(g,d)}}(I_{d})$ is a $2d$-regular $(2gd,2gd,2gd-3g+g/d+1,g(3g-g/d-1))$ PDA.
\end{lemma}
\begin{proof}
Consider the $2g\times 2g$ $2$-regular PDA $P_0=H_{2g}([g(2g-1)])$ with $\x$s along the diagonal. This PDA is partitioned into $2g/d\times 2g/d$ blocks $P_0^{(j,k)}$, $j,k\in[d]$ such that the diagonal blocks ($j=k$) are $2$-PDAs and others ($j\ne k$) are $1$-PDAs. Note that $P_0^{(j,k)}=(P_0^{(k,j)})^T$ when $j\ne k$. 

Construct $P_i$, $i=1,\ldots,d-1$, as a block-wise concatenation of $P_i^{(j,k)}$ defined as follows.
\begin{align*}
    P_i^{(j,k)}=\begin{cases}
    \pi_{AD,2}^i(P_0^{(j,j)}), & j=k,\\
    \pi_{AD,1}^{2i}(P_0^{(j,k)}), & j>k,\\
    (P_i^{(k,j)})^T, & j<k.
    \end{cases}
\end{align*}
We claim that $\mathcal{P}=\{P_0,\ldots,P_{d-1}\}$ is the desired set. Within $P_i$, Blackburn compatibility between $P_i^{(j,k)}$ and $P_i^{(k,j)}$ w.r.t. $P_i^{(j,j)}$ and $P_i^{(k,k)}$ is satisfied by Construction T1.
{An integer appearing in a diagonal block of $P_{i}$ will appear in a diagonal block of $P_{i'}$ whereas that appearing in a non-diagonal block will appear in two non-diagonal blocks of $P_{i'}$.}
Between $P_i$ and $P_{i'}$, Blackburn compatibility is satisfied for (1) diagonal blocks $\{P_i^{(j,j)},P_{i'}^{(j,j)}\}$ by Construction C2, (2) non-diagonal blocks $\{P_i^{(j,k)},P_{i'}^{(j,k)}\}$ by Construction C1, and (3) non-diagonal blocks $\{P_i^{(j,k)},P_{i'}^{(k,j)}\}$ by Construction T2. 
\end{proof}
\begin{example}
For $g=4$ and $d=2$ we have,
\begin{align*}
    P_* &= B_{\tilde{I}_4}(J_2[4]) = \left(\begin{array}{cccc:cccc}
\x & \x & \x & t_0 & \x & \x & \x & t_1 \\
\x & \x & t_0 & \x & \x & \x & t_1 & \x \\
\x & t_0 & \x & \x & \x & t_1 & \x & \x \\
t_0 & \x & \x & \x & t_1 & \x & \x & \x \\ \hdashline
\x & \x & \x & t_2 & \x & \x & \x & t_3 \\
\x & \x & t_2 & \x & \x & \x & t_3 & \x \\
\x & t_2 & \x & \x & \x & t_3 & \x & \x \\
t_2 & \x & \x & \x & t_3 & \x & \x & \x 
    \end{array}\right), 
\end{align*}
\begin{align*}
    P_0 &= 
{\left(\begin{array}{cc}
    P_0^{(0,0)} & P_0^{(0,1)} \\
    P_0^{(1,0)} & P_0^{(1,1)}
\end{array}\right)} =
\left(\begin{array}{cccc:cccc}
\x & 0 & 1 & 2 & 3 & 4 & 5 & 6 \\
0 & \x & 7 & 8 & 9 & 10 & 11 & 12 \\
1 & 7 & \x & 13 & 14 & 15 & 16 & 17 \\
2 & 8 & 13 & \x & 18 & 19 & 20 & 21 \\ \hdashline
3 & 9 & 14 & 18 & \x & 22 & 23 & 24 \\
4 & 10 & 15 & 19 & 22 & \x & 25 & 26 \\
5 & 11 & 16 & 20 & 23 & 25 & \x & 27 \\
6 & 12 & 17 & 21 & 24 & 26 & 27 & \x 
    \end{array}\right),\\ 
    P_1 &= 
{\left(\begin{array}{cc}
    \pi_{AD,2}(P_0^{(0,0)}) & (P_1^{(1,0)})^T \\
    \pi_{AD,1}^{2}(P_0^{(1,0)}) & \pi_{AD,2}(P_0^{(1,1)})
\end{array}\right)} =
\left(\begin{array}{cccc:cccc}
\x & 0 & 1 & 7 & 3 & 4 & 5 & 15 \\
0 & \x & 2 & 8 & 9 & 10 & 18 & 12 \\
1 & 2 & \x & 13 & 14 & 6 & 16 & 17 \\
7 & 8 & 13 & \x & 11 & 19 & 20 & 21 \\ \hdashline
3 & 9 & 14 & 11 & \x & 22 & 23 & 25 \\
4 & 10 & 6 & 19 & 22 & \x & 24 & 26 \\
5 & 18 & 16 & 20 & 23 & 24 & \x & 27 \\
15 & 12 & 17 & 21 & 25 & 26 & 27 & \x
    \end{array}\right) 
\end{align*}
The dashed lines show the block-wise construction. $\{P_0^{(j,j)},P_1^{(j,j)}\}$, $j=0,1$, satisfies Construction C2, while $\{P_0^{(j,1-j)},P_1^{(j,1-j)}\}$ satisfies Construction C1 and $\{P_0^{(j,1-j)},P_1^{(1-j,j)}\}$, $j=0,1$ satisfies Construction T2. 
\end{example}

In the previous construction, the blocks of $P_*$ were $\tilde{I}_g$, which has one integer per column. In the next block-wise construction, we will consider $P_*$ with diagonal blocks containing multiple integers per column and off-diagonal blocks as copies of $\tilde{I}_g$. 
This results in a higher integer density, and requires rotation of diagonal blocks in addition to the permutations employed in Lemma \ref{lm:bw1}. To describe $P_*$ precisely, we define a $2d\times 2d$ 2-PDA $T_{d^2}$ as follows:
\begin{equation}
    T_{d^2} = \begin{pmatrix}
    \tilde{I}_2(0)&\cdots&\tilde{I}_2(d-1)\\
    \tilde{I}_2(d)&\cdots&\tilde{I}_2(2d-1)\\
    \vdots&\vdots&\vdots\\
    \tilde{I}_2(d^2-d)&\cdots&\tilde{I}_2(d^2-1)
    \end{pmatrix},
\end{equation}
which can be alternatively defined as a basic regular lifting of the 1-PDA $J_d([d^2])$ by $\tilde{I}_2$, i.e. $T_{d^2}=B_{\tilde{I}_2}(J_d([d^2]))$. Cyclic rotations of the anti-diagonal of $T_{d^2}$ results in the set $\mathcal{T}_{d^2}=\{\pi_{AD,1}^{2i}(T_{d^2}):i\in[d]\}$ of $d$ PDAs, which are Blackburn-compatible w.r.t. $\tilde{I}_{2d}$ by an extension of Construction C1. 

For positive integers $g$, $d$ such that $d^2\mid g$, let $P_b^{(g,d^2)}=\begin{pmatrix}
I_d(0)&\cdots&\x_{d\times d}\\
\vdots&\ddots&\vdots\\
\x_{d\times d}&\cdots&I_d(g/d^2-1)
\end{pmatrix}$ be a $g/d\times g/d$ diagonal $d$-PDA.

$P_*^{(g,d^2)}={L}_{\mathcal{T}_{d^2},\tilde{I}_{2d}}(P_b^{(g,d^2)})$ is obtained by a general regular lifting (Corollary \ref{co:regularTiling}) of the $d$-PDA $P_b^{(g,d^2)}$ using the $d$ Blackburn-compatible PDAs in $\mathcal{T}_{d^2}$. The diagonal blocks of $P_*^{(g,d^2)}$ are 2-PDAs with $\x$s at the even diagonals, i.e. at $\{(i,j): i-j=0\mod 2\}$. Off-diagonal blocks of $P_*$ are $\tilde{I}_{2d}$, which has $\x$s at all even diagonals. For example, with $g=4$ and $d=2$, we get
\begin{align}
    P_*^{(4,2^2)} &= 
{\left(\begin{array}{cc}
    T_{d^2} & \tilde{I}_4(t_0) \\
    \tilde{I}_4(t_1) & \pi_{AD,1}^{2}(T_{d^2})
\end{array}\right)} =
\left(\begin{array}{cccc:cccc}
\x & t_2 & \x & t_4 & \x & \x & \x & t_0 \\
t_2 & \x & t_4 & \x & \x & \x & t_0 & \x \\
\x & t_3 & \x & t_5 & \x & t_0 & \x & \x \\
t_3 & \x & t_5 & \x & t_0 & \x & \x & \x \\ \hdashline
\x & \x & \x & t_1 & \x & t_2 & \x & t_3 \\
\x & \x & t_1 & \x & t_2 & \x & t_3 & \x \\
\x & t_1 & \x & \x & \x & t_4 & \x & t_5 \\
t_1 & \x & \x & \x & t_4 & \x & t_5 & \x
    \end{array}\right).
    \label{eq:Pstarg4d2}
\end{align}
\begin{lemma}[Construction BW2]\label{lm:bw2}
Let $g$ and $d$ be positive integers such that $d^2\mid g$. There exist a set $\mathcal{P}$ of $d$ $(2g,2g,1,g(2g-1))$ PDAs that are Blackburn-compatible w.r.t $P_*^{(g,d^2)}={L}_{\mathcal{T}_{d^2},\tilde{I}_{2d}}(P_b^{(g,d^2)})$, such that $L_{\mathcal{P},P_*}(I_{d})$ is a $2d$-regular $(2gd,2gd,2gd-3g+g/d+2d-d^2,g(3g-g/d-2d+d^2))$ PDA.
\end{lemma}
\begin{proof}
Consider the $2g\times 2g$ $2$-PDA $P_0$ defined in the proof of Lemma \ref{lm:bw1} and its $2g/d\times 2g/d$ blocks $P_0^{(j,k)}$, $j,k\in[d]$. Construct $P_i$, $i=0,\ldots,d-1$ as a blockwise concatenation of $P_i^{(j,k)}$ defined as follows.
\begin{align*}
    P_i^{(j,k)}=\begin{cases}
    \pi^i_{AD,2}(P_0^{(\tilde{j},\tilde{j})}), & j=k, \tilde{j}=j-i\text{ mod } d,\\
    \pi^{2i}_{AD,1}(P_0^{(j,k)}), & j>k\\
    (P_i^{(k,j)})^T, & j<k.
    \end{cases}
\end{align*}
The main difference, when compared to Construction BW1, is the rotation of the diagonal blocks.

We claim that $\mathcal{P}=\{P_0,\ldots,P_{d-1}\}$ is the desired set. Within $P_i$, Blackburn compatibility is satisfied by Construction T1. Between $P_i$ and $P_{i'}$, we consider the diagonal case and two off-diagonal cases separately: (1) Two diagonal blocks $P_i^{(j,j)}$ and $P_{i'}^{(j',j')}$ share the same set of integers only when $j\ne j'$, which means that the mirrored locations fall in an off-diagonal block $\tilde{I}_{2d}$ of $P_*$. So, by Construction C1, Blackburn compatibility is satisfied for diagonal blocks. (2) Two off-diagonal blocks $P_i^{(j,k)}$ and $P_{i'}^{(k,j)}$, $j\ne k$, share the same set of integers and have mirrored locations falling on diagonal blocks of $P_*$, which have $\x$s on all even diagonals. So, by Construction T2, Blackburn compatibility is satisfied. (3) Two off-diagonal blocks $P_i^{(j,k)}$ and $P_{i'}^{(j,k)}$, $j\ne k$, share the same set of integers, and the mirrored locations fall in an off-diagonal block $\tilde{I}_{2d}$ of $P_*$. So, by Construction C1, Blackburn compatibility is satisfied.
\end{proof}
\begin{example}
Let $g=4$ and $d=2$. $P_*$ is given in \eqref{eq:Pstarg4d2}.
\begin{align*}
    P_0 &= 
{\left(\begin{array}{cc}
    P_0^{(0,0)} & P_0^{(0,1)} \\
    P_0^{(1,0)} & P_0^{(1,1)}
\end{array}\right)} =
\left(\begin{array}{cccc:cccc}
\x & 0 & 1 & 2 & 3 & 4 & 5 & 6 \\
0 & \x & 7 & 8 & 9 & 10 & 11 & 12 \\
1 & 7 & \x & 13 & 14 & 15 & 16 & 17 \\
2 & 8 & 13 & \x & 18 & 19 & 20 & 21 \\ \hdashline
3 & 9 & 14 & 18 & \x & 22 & 23 & 24 \\
4 & 10 & 15 & 19 & 22 & \x & 25 & 26 \\
5 & 11 & 16 & 20 & 23 & 25 & \x & 27 \\
6 & 12 & 17 & 21 & 24 & 26 & 27 & \x 
    \end{array}\right),
\end{align*}
\begin{align*}
    P_1 &= 
{\left(\begin{array}{cc}
    \pi_{AD,2}(P_0^{(1,1)}) & (P_1^{(1,0)})^T \\
    \pi_{AD,1}^{2}(P_0^{(1,0)}) & \pi_{AD,2}(P_0^{(0,0)})
\end{array}\right)} =
\left(\begin{array}{cccc:cccc}
\x & 22 & 23 & 25 & 3 & 4 & 5 & 15 \\
22 & \x & 24 & 26 & 9 & 10 & 18 & 12 \\
23 & 24 & \x & 27 & 14 & 6 & 16 & 17 \\
25 & 26 & 27 & \x & 11 & 19 & 20 & 21 \\ \hdashline
3 & 9 & 14 & 11 & \x & 0 & 1 & 7 \\
4 & 10 & 6 & 19 & 0 & \x & 2 & 8 \\
5 & 18 & 16 & 20 & 1 & 2 & \x & 13 \\
15 & 12 & 17 & 21 & 7 & 8 & 13 & \x
    \end{array}\right).
\end{align*}
\end{example}
For the next blockwise construction, we consider $P_*=B_{H_{g/d}}(I_d)$ to be a $g\times g$ PDA constructed by basic lifting of $I_d$ using $H_{g/d}([\frac{g}{2d}(\frac{g}{d}-1)])$, for integers $g,d$ such that $d|g$. Therefore, it has $H_{g/d}$ repeated along the diagonal blocks and all-$\x$ blocks appearing in the off-diagonal blocks.

\begin{lemma}[Construction BW3]\label{lm:bw3}
For positive integers $g$ and $d$ such that $d\mid g$, there exist 
a set $\mathcal{P}$ of $d$ $2$-regular $(g,g,1,\allowbreak g(g-1)/2)$ PDAs that are Blackburn compatible w.r.t $P_*=B_{H_{g/d}}(I_{d})$
such that $L_{\mathcal{P},P_*}(I_d)$ is a $2d$-regular $(gd,gd,gd+d-2g+g/d,g(2g-d-g/d)/2)$ PDA.
\end{lemma}
\begin{proof}
Consider the $g\times g$ $2$-PDA $P_0=H_g([g(g-1)/2])$ and its $g/d\times g/d$ blocks $P_0^{(j,k)}$, $j,k\in[d]$. Construct $P_i$, $i=0,\ldots,d-1$ as a blockwise concatenation of $P_i^{(j,k)}$ defined as follows.
\begin{align*}
    P_i^{(j,k)}=\begin{cases}
    P_0^{(\tilde{j},\tilde{j})}, & j=k, \tilde{j}=j-i\text{ mod } d,\\
    P_0^{(j,k)}, & j\ne k.
    \end{cases}
\end{align*}
We claim that $\mathcal{P}=\{P_0,\ldots,P_{d-1}\}$ is the desired set, and skip the proof details, which are largely similar to the previous proof of Construction BW2.
\end{proof}
\begin{example}
Let $g=6$ and $d=2$.
\begin{align*}
    P_* &= 
{B_{H_{3}}(I_{2})} =
\left(\begin{array}{ccc:ccc}
\x & 0 & 1 & \x & \x & \x \\
0 & \x & 2 & \x & \x & \x \\
1 & 2 & \x & \x & \x & \x \\ \hdashline
\x & \x & \x & \x & 0 & 1 \\
\x & \x & \x & 0 & \x & 2 \\
\x & \x & \x & 1 & 2 & \x
    \end{array}\right),
\end{align*}
\begin{align*}
    P_0 &= 
{\left(\begin{array}{cc}
    P_0^{(0,0)} & P_0^{(0,1)} \\
    P_0^{(1,0)} & P_0^{(1,1)}
\end{array}\right)} =
\left(\begin{array}{ccc:ccc}
\x & 0 & 1 & 2 & 3 & 4 \\
0 & \x & 5 & 6 & 7 & 8 \\
1 & 5 & \x & 9 & 10 & 11 \\ \hdashline
2 & 6 & 9 & \x & 12 & 13 \\
3 & 7 & 10 & 12 & \x & 14 \\
4 & 8 & 11 & 13 & 14 & \x
    \end{array}\right),\\
    P_1 &= 
{\left(\begin{array}{cc}
    P_0^{(1,1)} & P_0^{(0,1)} \\
    P_0^{(1,0)} & P_0^{(0,0)}
\end{array}\right)} =
\left(\begin{array}{ccc:ccc}
\x & 12 & 13 & 2 & 3 & 4 \\
12 & \x & 14 & 6 & 7 & 8 \\
13 & 14 & \x & 9 & 10 & 11 \\ \hdashline
2 & 6 & 9 & \x & 0 & 1 \\
3 & 7 & 10 & 0 & \x & 5 \\
4 & 8 & 11 & 1 & 5 & \x 
    \end{array}\right).
\end{align*}
\end{example}
In the final block-wise construction, we will consider $P_*$ with diagonal blocks containing one $\x$ per column and off-diagonal blocks as copies of $\tilde{I}_g$. Letting $T=H_{2n}([n(2n-1)])$ be the $2n\times 2n$ 2-PDA as defined in Lemma \ref{lm:diagPDA}, we set
\begin{equation}
    P_*^{(n)}=\begin{pmatrix}
    T&\tilde{I}_{2n}(t_{0,1})&\ldots&\tilde{I}_{2n}(t_{0,n-1})\\
    \tilde{I}_{2n}(t_{1,0})&\pi_{AD,2}(T)&\ldots&\tilde{I}_{2n}(t_{1,n-1})\\
    \vdots&\vdots&\vdots&\vdots\\
    \tilde{I}_{2n}(t_{n-1,0})&\ldots&\tilde{I}_{2n}(t_{n-1,n-1})&\pi^{n-1}_{AD,2}(T)
    \end{pmatrix},
    \label{eq:Pstarn}
\end{equation}
where $t_{i,j}$ are integers. All the diagonal blocks of $P_*^{(n)}$ have $\x$s on the diagonals, all off-diagonal blocks have $\x$s on all even diagonals, and every column has $1+(n-1)(2n-1)=2n^2-3n+2$ $\x$s.

The construction of Blackburn-compatible PDAs is presented in the following Lemma.
\begin{lemma}\label{lm:bw4}
(Construction BW4) For an integer $n$, there exist a set $\mathcal{P}$ of $2n$ $(2n^2, 2n^2, 0, 4n^4)$ 1-PDAs that are Blackburn-compatible w.r.t. $P_*^{(n)}$ defined in \eqref{eq:Pstarn} such that  $L_{\mathcal{P},P_*^{(n)}}(I_{2n})$ is a $2n$-regular $(4n^3,4n^3,4n^3-8n^2+7n-2,2n^2(8n^2-7n+2))$ PDA.
\end{lemma}
\begin{proof}
Consider the $2n^2\times 2n^2$ 1-PDA $P_0$ defined as a blockwise concatenation of $2n\times 2n$ blocks. The $(j,k)$-th block, $j,k\in[n]$, $P_0^{(j,k)}=J_{2n}(S_{jk})$, where $S_{jk}=\{4n^2(nj+k),4n^2(nj+k+1)-1\}$ is a partition of the set $[4n^4]$ consisting of disjoint subsets of $4n^2$ consecutive integers.

Construct $P_i$, $i=1,\ldots,n-1$, as a blockwise concatenations of $2n\times 2n$ PDAs $P_i^{(j,k)}$, $j,k\in[n]$, defined as follows:
$$
P_i^{(j,k)} = \begin{cases}
\pi^2_{AD,1}(P_{i-1}^{(j-1\text{ mod }n,j-1\text{ mod }n)}),&j=k,\\
\pi^2_{AD,1}(P_{i-1}^{(j,k)}),&j\ne k.
\end{cases}
$$
When going from $P_{i-1}$ to $P_i$, a diagonal block is cyclically shifted down, and its anti-diagonal is cyclically shifted twice. An off-diagonal block has its anti-diagonal cyclically shifted twice.

Construct $\tilde{P}_i$, $i\in[n]$, as
$$\tilde{P}_i^{(j,k)} = \begin{cases}
(P_i^{(j,j)})^T,&j=k,\\
\pi_{AD,1}(P_i^{(j,k)}),&j\ne k.
\end{cases}$$
When going from $P_i$ to $\tilde{P}_i$, diagonal blocks are transposed. An off-diagonal block has its anti-diagonal cyclically shifted once.

We claim that $\mathcal{P}=\{P_0,\ldots,P_{n-1},\tilde{P}_0,\ldots,\tilde{P}_{n-1}\}$ is Blackburn-compatible w.r.t. $P_*^{(n)}$. To prove the claim, we will consider multiple cases where two blocks of the PDAs in $\mathcal{P}$ share integers. (1) $P_i^{(j,k)}$ and $P_{i'}^{(j,k)}$, $i\ne i'$, share integers if $j\ne k$. $P_i^{(j,j)}$ and $P_{i'}^{(j',j')}$, $i\ne i'$, share integers only if $j\ne j'$. In both cases, the mirrored locations are an off-diagonal block of $P_*$, which is $\tilde{I}_{2n}$. By Construction C1, Blackburn compatibility follows. The same argument holds when $P$ is replaced with $\tilde{P}$ in this case. (2) $P_i^{(j,j)}$ and $\tilde{P}_{i}^{(j,j)}$ share integers. Mirrored locations are on diagonal blocks of $P_*$, which have $\x$s on the diagonal. Blackburn compatibility follows by Construction T1. (3) $P_i^{(j,k)}$ and $\tilde{P}_{i'}^{(j,k)}$ share integers if $j\ne k$. $P_i^{(j,j)}$ and $\tilde{P}_{i'}^{(j',j')}$, $i\ne i'$, share integers only if $j\ne j'$. Mirrored locations are an off-diagonal block of $P_*$, which is $\tilde{I}_{2n}$. By Construction C1, Blackburn compatibility follows.  
\end{proof}

\begin{example}
For $n=2$, the following set of $8\times 8$ Blackburn PDAs can be constructed using BW4.
\begin{align*}
    P_* &= 
{\left(\begin{array}{cc}
    H_4([2,7]) & \tilde{I}_4(0) \\
    \tilde{I}_4(1) & \pi_{AD,2}(H_4([2,7]))
\end{array}\right)} =
\left(\begin{array}{cccc:cccc}
\x & 4 & 3 & 6 & \x & \x & \x & 0 \\
4 & \x & 2 & 5 & \x & \x & 0 & \x \\
3 & 2 & \x & 7 & \x & 0 & \x & \x \\
6 & 5 & 7 & \x & 0 & \x & \x & \x \\ \hdashline
\x & \x & \x & 1 & \x & 4 & 3 & 2 \\
\x & \x & 1 & \x & 4 & \x & 6 & 5 \\
\x & 1 & \x & \x & 3 & 6 & \x & 7 \\
1 & \x & \x & \x & 2 & 5 & 7 & \x 
    \end{array}\right)
\end{align*}
\begin{align*}
    P_0 &= 
{\left(\begin{array}{cc}
    P_0^{(0,0)} & P_0^{(0,1)} \\
    P_0^{(1,0)} & P_0^{(1,1)}
\end{array}\right)} =
\left(\begin{array}{cccc:cccc}
0 & 1 & 2 & 3 & 16 & 17 & 18 & 19 \\
4 & 5 & 6 & 7 & 20 & 21 & 22 & 23 \\
8 & 9 & 10 & 11 & 24 & 25 & 26 & 27 \\
12 & 13 & 14 & 15 & 28 & 29 & 30 & 31 \\ \hdashline
32 & 33 & 34 & 35 & 48 & 49 & 50 & 51 \\
36 & 37 & 38 & 39 & 52 & 53 & 54 & 55 \\
40 & 41 & 42 & 43 & 56 & 57 & 58 & 59 \\
44 & 45 & 46 & 47 & 60 & 61 & 62 & 63
    \end{array}\right),
\end{align*}
\begin{align*}
    \tilde{P}_0 &= 
{\left(\begin{array}{cc}
    (P_0^{(0,0)})^T & \pi_{AD,1}(P_0^{(0,1)}) \\
    \pi_{AD,1}(P_0^{(1,0)}) & (P_0^{(1,1)})^T
\end{array}\right)} =
\left(\begin{array}{cccc:cccc}
0 & 4 & 8 & 12 & 16 & 17 & 18 & 22 \\
1 & 5 & 9 & 13 & 20 & 21 & 25 & 23 \\
2 & 6 & 10 & 14 & 24 & 28 & 26 & 27 \\
3 & 7 & 11 & 15 & 19 & 29 & 30 & 31 \\ \hdashline
32 & 33 & 34 & 38 & 48 & 52 & 56 & 60 \\
36 & 37 & 41 & 39 & 49 & 53 & 57 & 61 \\
40 & 44 & 42 & 43 & 50 & 54 & 58 & 62 \\
35 & 45 & 46 & 47 & 51 & 55 & 59 & 63
    \end{array}\right),
\end{align*}
\begin{align*}
    P_1 &= 
{\left(\begin{array}{cc}
    \pi_{AD,1}^{2}(P_0^{(1,1)}) & \pi_{AD,1}^{2}(P_0^{(0,1)}) \\
    \pi_{AD,1}^{2}(P_0^{(1,0)}) & \pi_{AD,1}^{2}(P_0^{(0,0)})
\end{array}\right)} =
\left(\begin{array}{cccc:cccc}
48 & 49 & 50 & 57 & 16 & 17 & 18 & 25 \\
52 & 53 & 60 & 55 & 20 & 21 & 28 & 23 \\
56 & 51 & 58 & 59 & 24 & 19 & 26 & 27 \\
54 & 61 & 62 & 63 & 22 & 29 & 30 & 31 \\ \hdashline
32 & 33 & 34 & 41 & 0 & 1 & 2 & 9 \\
36 & 37 & 44 & 39 & 4 & 5 & 12 & 7 \\
40 & 35 & 42 & 43 & 8 & 3 & 10 & 11 \\
38 & 45 & 46 & 47 & 6 & 13 & 14 & 15
    \end{array}\right),
\end{align*}
\begin{align*}
    \tilde{P}_1 &= 
{\left(\begin{array}{cc}
    (P_1^{(0,0)})^T & \pi_{AD,1}(P_1^{(1,0)}) \\
    \pi_{AD,1}(P_1^{(1,0)}) & (P_1^{(1,1)})^T
\end{array}\right)} =
\left(\begin{array}{cccc:cccc}
48 & 52 & 56 & 54 & 16 & 17 & 18 & 28 \\
49 & 53 & 51 & 61 & 20 & 21 & 19 & 23 \\
50 & 60 & 58 & 62 & 24 & 22 & 26 & 27 \\
57 & 55 & 59 & 63 & 25 & 29 & 30 & 31 \\ \hdashline
32 & 33 & 34 & 44 & 0 & 4 & 8 & 6 \\
36 & 37 & 35 & 39 & 1 & 5 & 3 & 13 \\
40 & 38 & 42 & 43 & 2 & 12 & 10 & 14 \\
41 & 45 & 46 & 47 & 9 & 7 & 11 & 15
    \end{array}\right).
\end{align*}
\end{example}

\subsection{Tiling construction} 
The next lemma provides a tiling construction for regular Blackburn-compatible PDAs, and characterizes a tradeoff between the number of $\x$s per column and the number of integers.
\begin{lemma}[Tiling]\label{lm:BCboundZ}
For a positive integer $g$ and a divisor $d$ of $g$, there exists a set $\mathcal{P}$ of $d$ $(g,g,g-d,d^2)$ $(g/d)$-PDAs that are Blackburn compatible w.r.t. $I_g(t)$ (when $t$ does not appear in any PDA in $\mathcal{P}$). $L_{\mathcal{P},I_g}(I_d)$ is a $g$-regular $(dg,dg,dg-2d+1,d(2d-1))$ PDA.
\end{lemma}
\begin{proof}
Let $P_0,\ldots,P_{d-1}$ be $(d,d,0,d^2)$ $1$-PDAs Blackburn compatible w.r.t. $I_d(t)$ obtained using the first part of Lemma~\ref{lm:bc12pdas}. Each of these PDAs has $d$ integers per column.

Let $q=g/d$. Replace an integer $s$ in $P_i$ with $I_q(s)$ to obtain a $(g,g,g-d,d^2)$ $q$-PDA, which we denote as $\tilde{P}_i$. It is easy to see that $\tilde{\mathcal{P}}=\{\tilde{P}_0,\ldots,\tilde{P}_{d-1}\}$ is a set of PDAs Blackburn-compatible w.r.t. $I_g(t)$. 

$L\triangleq L_{\tilde{\mathcal{P}},I_g}(I_d)$ has $\tilde{P}_i$, $i=0,\ldots,d-1$, on the diagonal and $P_*=I_g$ as off-diagonal blocks. Each integer in $[d^2]$ occurs $q$ times in every $P_i$, which adds up to a total of $qd=g$ times in $L$. $I_g$, by definition, contains one integer appearing $g$ times. Each column of $L$ has $P_i$ with $d$ integers on the diagonal and $I_g$ with one integer on $d-1$ off-diagonal positions. So, each column has $d+d-1=2d-1$ integers, or $gd-(2d-1)$ $\x$s. So, we see that the parameters of $L$ are as claimed.
\end{proof}
The PDAs in the above lemma are visualized as ``tilings'' of the identity PDA as seen in the example below.
\begin{example}
Let $g=6$, $d=3$. Since $q=2$, we obtain the PDAs given below by lifting the Blackburn compatible PDAs from ~\eqref{eq:1gP3} with $I_2(t)$.
\begin{align}
\label{eq:2gP6}
    P_0 = \begin{pmatrix}
        I_2(0)&I_2(1)&I_2(2)\\
        I_2(3)&I_2(4)&I_2(5)\\
        I_2(6)&I_2(7)&I_2(8)
    \end{pmatrix},
    P_1 = \begin{pmatrix}
        I_2(8)&I_2(1)&I_2(2)\\
        I_2(3)&I_2(0)&I_2(5)\\
        I_2(6)&I_2(7)&I_2(4)
    \end{pmatrix},
    P_2 = \begin{pmatrix}
        I_2(4)&I_2(1)&I_2(2)\\
        I_2(3)&I_2(8)&I_2(5)\\
        I_2(6)&I_2(7)&I_2(0)
    \end{pmatrix}.
\end{align}
$L_{\{P_0,P_1,P_2\},I_6}(I_3)$ is an $(18,18,13,15)$ 6-PDA.
\end{example}
Using integer-disjoint copying, the $d$ PDAs obtained from Lemma \ref{lm:BCboundZ} can be replicated to obtain a set of $md$ $(g,g,g-d,d^2)$ $(g/d)$-PDAs. This idea is captured in the following corollary.
\begin{corollary}[Tiling]\label{cor:BCboundZ}
For positive integers $g,b$ and $d=\gcd(g,b)$, there exists a set $\mathcal{P}$ of $b$ $(g,g,g-d,d^2)$ $(g/d)$-PDAs that are Blackburn compatible w.r.t. $I_g(t)$ (when $t$ does not appear in any PDA in $\mathcal{P}$). $L_{\mathcal{P},I_g}(I_b)$ is a $g$-regular $(bg,bg,bg-b-d+1,b(b+d-1))$ PDA.
\end{corollary}
\begin{proof}
Obtain $d$ PDAs using Lemma \ref{lm:BCboundZ}. Since $b$ is a multiple of $d$, the $d$ PDAs can be replicated with disjoint sets of integers to obtain $b$ PDAs with the same parameters.

The parameters of $L_{\mathcal{P},I_g}(I_b)$ are easy to verify.
\end{proof}
\begin{example}
For $n=6$, $b=4$, we have $d=2$ and $q=3$. Using Lemma~\ref{lm:bc12pdas} obtain PDAs $Q_0'$ and $Q_1'$ which are Blackburn-compatible w.r.t. $I_2(t)$. Obtain $Q_2'$ and $Q_3'$ by replacing integers in $Q_0'$ and $Q_1'$ with a new disjoint set. Lift these PDAs using $I_3(t)$ to obtain the required set $\mathcal{Q}=\{Q_0,Q_1,Q_2,Q_3\}$ shown below.
 \begin{align}
\label{eq:4gP6}
    Q_0 = \begin{pmatrix}
        I_3(0)&I_3(1)\\
        I_3(2)&I_3(3)
    \end{pmatrix},
    Q_1 = \begin{pmatrix}
        I_3(3)&I_3(1)\\
        I_3(2)&I_3(0)
    \end{pmatrix},
    Q_2 = \begin{pmatrix}
        I_3(4)&I_3(5)\\
        I_3(6)&I_3(7)
    \end{pmatrix},
    Q_3 = \begin{pmatrix}
        I_3(7)&I_3(5)\\
        I_3(6)&I_3(4)
    \end{pmatrix}.
\end{align}
Note that $L(\mathcal{Q},I_6)$ is a 6-regular $(24,24,19,20)$ PDA.
\end{example}

\subsection{Recursive construction} 
Repeated application of lifting can result in larger Blackburn-compatible PDAs starting from small-sized base PDAs. The following lemma is an important step in such recursive constructions.
\begin{lemma}\label{lem:recursive}
    Let $\mathcal{P}=\{P_0,\dots,P_{g-1}\}$ be a set of $g$ $n\times n$ PDAs Blackburn compatible w.r.t. a PDA $P_*$ with $\mathcal{S}$ as its set of integers. For sets of disjoint integers $\mathcal{S}_j,j\in[g(g-1)]\cup \{t\}$, with $|\mathcal{S}_j|=|\mathcal{S}|$, let $P_*(\mathcal{S}_j)$ indicate the PDA constructed by replacing integers in $P_*$ with those in $\mathcal{S}_j$. Let $\pi(l)=(l+1)\;\text{mod}\;g$ be the cyclic shift permutation on $[g]$. For $i\in[g]$, let 
\begin{equation}
    P_{i}' = \left(\begin{array}{cccc}
        P_{\pi^i(1)} & P_*(\mathcal{S}_0) & \dots & P_*(\mathcal{S}_{g-2})  \\
        P_*(\mathcal{S}_{g-1}) & P_{\pi^i(2)} & \dots & P_*(\mathcal{S}_{2g-3})\\
        \vdots & & \ddots & \vdots \\
        P_*(\mathcal{S}_{(g-1)^2}) &  & \dots & P_{\pi^i(g)}
    \end{array}\right),\quad
    P_{\x}' = \left(\begin{array}{cccc}
        P_*(\mathcal{S}_t) & \x_{n\times n} & \dots & \x_{n\times n} \\
        \x_{n\times n} & P_*(\mathcal{S}_t) & \dots & \x_{n\times n} \\
        \vdots &  & \ddots & \vdots \\
        \x_{n\times n} & \x_{n\times n} & \dots & P_*(\mathcal{S}_t)
    \end{array}\right).
\label{eq:recursive}
\end{equation}
Then, $\mathcal{P}'=\{P'_0,\ldots,P'_{g-1}\}$ is a set of $g$ $ng\times ng$ PDAs  Blackburn-compatible w.r.t. $P'_{\x}$. If $L_{\mathcal{P},P_*}(I_g)$ is a $g_b$-regular $ng\times ng$ PDA, then $L_{\mathcal{P}',P'_*}(I_g)$ is a $gg_b$-regular $ng^2\times ng^2$ PDA.
\end{lemma}
\begin{proof}
A $P_*(S_j)$ block appears off-diagonal at the same location in all $P'_i$. This does not violate Blackburn compatibility w.r.t. $P'_*$, which has all-$\x$ arrays as off-diagonal blocks.  

A block that appears in the diagonal of $P_i'$ at its $k$-th position will be $P_{\pi^i(k)}$ and that of $P_j'$ will be  $P_{\pi^j(k)}$. Since $\pi^i(k)\ne \pi^j(k)$ for $i\ne j$, these blocks do not violate Blackburn compatibility with $P'_*$.

Finally, we prove the regularity claim. If $L_{\mathcal{P}}(I_g(s))$ is $g_b$-regular, then $P_*(\mathcal{S})$ is $g_b$-regular, and so is $P_*(\mathcal{S}_i)$ for every $i$. So, $P'_*$ is $gg_b$-regular, which implies that $L_{\mathcal{P}'}(I_g(s))$ is $gg_b$-regular in the off-diagonal blocks. Since every $P'_i$ is a diagonal-block-permuted version of $L_{\mathcal{P}}(I_g(s))$, $P'_i$ (which is a diagonal block in $L_{\mathcal{P}'}(I_g(s))$) is $g_b$-regular as well. Since the same set of integers appear in every $P'_i$, the diagonal blocks of $L_{\mathcal{P}'}(I_g(s))$ together are $gg_b$-regular.  
\end{proof}
We will now illustrate how to apply the above lemma repeatedly to construct larger PDAs. Consider the two $1$-PDAs, denoted $A_2(x)$ and $A'_2(x)$, which are Blackburn-compatible w.r.t. $I_2(t)$ for $t\notin\{x,x+1,x+2,x+3\}$ and obtained using Lemma \ref{lm:bc12pdas} (first part) for $n=2$. They are given by
\begin{equation}
    A_{2}(x) = \left(\begin{array}{cc}
        x & x+1 \\
        x+2 & x+3
    \end{array}\right),
    \,
    A_{2}'(x) = \left(\begin{array}{cc}
        x+3 & x+1 \\
        x+2 & x
    \end{array}\right).\label{eq:A2}
\end{equation}
\begin{lemma}\label{lem:A2r_compatibility}
Consider the following $2^r \times 2^r$ PDAs for $r> 1$.
\begin{equation}
\resizebox{.9\hsize}{!}{
    $A_{2^r}(x) = \left(\begin{array}{cc}
        I_{2^{r-1}}(x) & A_{2^{r-1}}(x+2) \\
        A_{2^{r-1}}'(x+2) & I_{2^{r-1}}(x+1) 
    \end{array}\right),
    \,
    A_{2^r}'(x) = \left(\begin{array}{cc}
        I_{2^{r-1}}(x+1) & A_{2^{r-1}}(x+2) \\
        A_{2^{r-1}}'(x+2) & I_{2^{r-1}}(x) 
    \end{array}\right)$,}\label{eq:A2r}
\end{equation} 
where $A_2(x)$ and $A_2'(x)$ are as defined in \eqref{eq:A2}. Let $\mathcal{S}$ be the set of integers in $A_{2^r}(x)$ and $A_{2^r}'(x)$. For $r\ge 1$ and $t \not\in \mathcal{S}$, $\mathcal{A}=\{A_{2^r}(x),A_{2^r}'(x)\}$ is a set of two Blackburn-compatible PDAs w.r.t. $I_{2^r}(t)$, and $L_{\mathcal{A}}(I_2(s))$ is a $2^r$-regular $(2^{r+1},2^{r+1},2^{r+1}-r-2,2r+4)$ PDA.
\end{lemma}
\begin{proof}
Start with $P_*=I_2$, $P_2=A_2(x)$, $P_1=A_2'(x)$ and apply Lemma~\ref{lem:recursive} $(r-1)$ times recursively.
\end{proof}

{A summary of the above constructions is given in Table~\ref{tab:summary}.} The final parameters of a lifted PDA can be calculated using Corollary~\ref{cor:genreg}.
\afterpage{%
    \clearpage%
    \thispagestyle{empty}%
    \begin{landscape}%
        \centering %
\vspace*{\fill}
        \captionof{table}{Summary of constructions of Blackburn compatible PDAs.}\label{tab:summary}
\begin{adjustbox}{max width=\linewidth}
\hspace{-1.5cm}
\begin{tabular}{cccccccccccccc}
    \toprule
    {\textbf{Construction}} & {\textbf{Condition}} & {\textbf{$g_b$}} & {\textbf{$K_c$}} & {\textbf{$f_c$}} & {\textbf{$P_i$}} & $Z_i$ & $S_i$ & $g_i$ & {\textbf{$P_*$}} & $Z_*$ & $S_*$ & $g_*$ \\ \midrule
    C1 &  & $g$ & $g$ & $g$ & $\pi_{D,1}^{i}(J_g([g^2])), i\in [g]$ & $0$ & $g^2$ & $1$ & $I_g(t)$ & $g-1$ & $1$ & $g$ \\\\ 
    C2 &  & $g$ & $2g$ & $2g$ & $\pi_{D,2}^{i}(G_{2g}([g(2g-1)])), i\in [g]$ & $1$ & $g(2g-1)$ & $2$ & $I_{2g}(t)$ & $2g-1$ & $1$ & $2g$ \\\\
    T1 &  & $2$ & $g$ & $g$ & $\{J_g([g^2],J_g([g^2]^T)\}$ & $0$ & $g^2$ & $1$ & $H_g([\frac{g(g-1)}{2}])$ & $1$ & $\frac{g(g-1)}{2}$ & $2$ \\\\ 
    T2 &  & $2$ & $g$ & $g$ \\\\
    BW1 & $d|g$ & $d$ & $2g$ & $2g$ & \makecell[c]{$P_0=H_{2g}([g(2g-1)])$\\$P_i^{(j,k)}=\begin{cases}
    \pi_{AD,2}^i(P_0^{(j,j)}), & j=k,\\
    \pi_{AD,1}^{2i}(P_0^{(j,k)}), & j>k,\\
    (P_i^{(k,j)})^T, & j<k.
    \end{cases}$\\where $i\in[d]\setminus \{0\}$ and $j,k\in[d]$} & $1$ & $g(2g-1)$ & $2$ & $B_{\tilde{I}_{2d}}(J_{g/d})$ & $2g-\frac{g}{d}$ & $(\frac{g}{d})^2$ & $2d$ \\\\ 
    BW2 & $d^2|g$ & $d$ & $2g$ & $2g$ & \makecell[c]{$P_0=H_{2g}([g(2g-1)])$\\$P_i^{(j,k)}=\begin{cases}
    \pi^i_{AD,2}(P_0^{(\tilde{j},\tilde{j})}), & j=k, \tilde{j}=j-i\text{ mod } d,\\
    \pi^{2i}_{AD,1}(P_0^{(j,k)}), & j>k\\
    (P_i^{(k,j)})^T, & j<k.
    \end{cases}$\\where $i\in[d]\setminus \{0\}$ and $j,k\in[d]$} & $1$ & $g(2g-1)$ & $2$ & ${L}_{\mathcal{T}_{d^2},\tilde{I}_{2d}}(P_b^{(g,d^2)})$ & $2g-\frac{g}{d}-d+1$ & $\frac{g}{d}(\frac{g}{d}+d-1)$ & $2d$ \\\\
    BW3 & $d|g$ & $d$ & $g$ & $g$ & \makecell[c]{$P_0=H_g([g(g-1)/2])$\\$P_i^{(j,k)}=\begin{cases}
    P_0^{(\tilde{j},\tilde{j})}, & j=k, \tilde{j}=j-i\text{ mod } d,\\
    P_0^{(j,k)}, & j\ne k.
    \end{cases}$\\where $i\in[d]\setminus \{0\}$ and $j,k\in[d]$} & $1$ & $\frac{g(g-1)}{2}$ & $2$ & $B_{H_{g/d}}(I_{d})$ & $g+1-\frac{g}{d}$ & $\frac{g}{2d}(\frac{g}{d}-1)$ & $2d$ \\\\ 
    BW4 &  & $2n$ & $2n^2$ & $2n^2$ & \makecell[c]{$P_0=J_{2n^2}([4n^4])$\\$P_i^{(j,k)} = \begin{cases}
\pi^2_{AD,1}(P_{i-1}^{(j-1\text{ mod }n,j-1\text{ mod }n)}),&j=k,\\
\pi^2_{AD,1}(P_{i-1}^{(j,k)}),&j\ne k.
\end{cases}$\\$\tilde{P}_i^{(j,k)} = \begin{cases}
(P_i^{(j,j)})^T,&j=k,\\
\pi_{AD,1}(P_i^{(j,k)}),&j\ne k.
\end{cases}$\\where $i\in[n]\setminus \{0\}$ and $j,k\in[n]$} & $0$ & $4n^4$ & $1$ & ${L}_{\{\pi^{i}_{AD,2}(H_{2n}): i\in[n]\},\tilde{I}_{2n}}(I_n)$ & $2n^2-3n+2$ & $n(3n-2)$ & $2n$ \\\\
Tiling & $d|g$ & $d$ & $g$ & $g$ & $B_{I_{g/d}}(\pi_{D,1}^{i}(J_d)), i\in [d]$ & $g-d$ & $d^2$ & $\frac{g}{d}$ & $I_g(t)$ & $g-1$ & $1$ & $g$ \\\\
$2^r$-lifting & & $2$ & $2^r$ & $2^r$ & $\{A_{2^r}(x),A_{2^r}'(x)\}$ & $2^r-r-1$ & $2+2r$ & $2^{r-1}$ & $I_{2^r}(t)$ & $2^r-1$ & $1$ & $2^r$ \\
\bottomrule
\end{tabular}
\end{adjustbox}
\vspace*{\fill}
    \end{landscape}
    \clearpage
}

\subsection{Randomized construction}
For positive integers $b$, $r$, $e$, $\alpha$ and $\eta$, we propose a randomised algorithm $\tRnd^{b,r}_{e,\alpha,\eta}$ that, when successful, will construct $\eta r \times \alpha r$ PDAs $P_i$, $i=0,\ldots,b-1$, satisfying the following conditions:
\begin{itemize}
    \item every $P_i$ contains $e$ $\x$s per column,
    \item the set $\{P_0,\ldots,P_{b-1}\}$ is Blackburn-compatible w.r.t. $P_*=\begin{pmatrix}
        I_r(t_{1,1})&\cdots&I_r(t_{1,\alpha})\\
        \vdots&\ddots&\vdots\\
        I_r(t_{\eta,1})&\cdots&I_r(t_{\eta,\alpha})
    \end{pmatrix}$, where $t_{j,k}$ are integers not occurring in the $P_i$'s,
    \item $L_{\{P_0,\ldots,P_{b-1}\},P_*}(I_b)$ is $r$-regular.
\end{itemize}
 A psuedocode for the random construction is given in Algorithm~\ref{alg:random}. 
\begin{algorithm}[H]
\caption{$\tRnd^{b,r}_{e,\alpha,\eta}$}\label{alg:random}
\begin{algorithmic}[1]
\State \textbf{function} CHECK($[P_i]_{x,y}\gets v$): \textbf{return} TRUE, \textbf{if} setting $(x,y)$-th position of $P_i$ as $v$ does not violate PDA, number of $\x$s per column and Blackburn compatibility conditions, \textbf{else} \textbf{return} FALSE 
\State \textbf{define} $\text{FREE}(P_i,v)=\{(x,y): [P_i]_{x,y}=\x, \text{CHECK}([P_i]_{x,y}\gets v)=\text{TRUE}\}$ 
\State $P_i\gets$ all-$\x$ for all $i$, $s\gets b\alpha(\eta r-e)$ ($s$: number of integers occurring in all $P_i$)
\State $i=0$ (start with $P_0$)
\For{$v\in\{1,\ldots,s\}$}
    \Loop $\ r$ times
        \State \textbf{if} $\text{FREE}(P_i,v)$ is empty \textbf{then} declare FAILURE and \textbf{exit}
        \State For $(x,y)\in \text{FREE}(P_i,v)$, $\text{PENALTY}(x,y)=N_r+N_c+w_r+w_c$, where (1) $N_r$ and $N_c$ are the number of rows and columns of all other $P_j$, $j\ne i$, invalidated under Blackburn compatibility by setting $(x,y)$-th position of $P_i$ as $v$, respectively, (2) $w_r$ and $w_c$ are the number of integers in the $x$-th row and $y$-th column of $P_i$, respectively. 
        \State $(x^*,y^*)=\arg\min_{(x,y)\in \text{FREE}(P_i,v)}\text{PENALTY}(x,y)$ (if multiple, pick one at random)
        \State Set $[P_i]_{x^*,y^*}\gets v$, Move to next $i=i+1\mod b$
    \EndLoop
\EndFor
\end{algorithmic}
\end{algorithm}
The algorithm cycles through the $P_i$, starting with $P_0$, and adds integer $v$, one at a time, starting with $v=1$. Each integer is added $r$ times. Free locations for adding $v$ in $P_i$ are identified, and a penalty term is calculated for each free location. The penalty tends to favour locations that minimize ``wasting" of cells in other $P_j$, $j\ne i$, and those that improve column and row spread of the integers in $P_i$. If no free locations are found at any point, the algorithm fails. Since all conditions are maintained throughout, the algorithm outputs the required PDAs, if successful.

Results of successful runs of $\tRnd^{b,r}_{e,\alpha,\eta}$ are given in Table \ref{tab:rndres}.

\begin{table}[hbt]
    \centering
    \caption{Examples of successful random generation of $b$ $\eta r\times\alpha r$ PDAs with $e$ $\x$s per column in each PDA and each integer occurring $r$ times across all PDAs.}
    \label{tab:rndres}
    \begin{tabular}{|c|c|c|c|}
    \hline
        $\eta$ & $\alpha$ & $(r,e)$               & $b$  \\
        \hline
1 & 1 & \makecell[c]{(3,2), (4,1), (5,3), (6,4), (8,5), (10,8), (12,9), (15,14), (16,14),\\(20,18), (25,24), (30,28), (32,30), (50,48), (125,124)} & 2 \\ 
1 & 1 & (4,2), (5,3), (8,6), (10,8), (20,18), (25,23), (50,49) & 3 \\ 
1 & 1 & (3,0), (4,0), (5,0), (6,3), (10,6), (12,9), (25,22) & 5 \\ 
1 & 1 & (3,1), (5,1), (6,1), (25,21) & 10 \\ 
1 & 1 & (3,2) & 20 \\ 
1 & 1 & (5,1) & 50 \\ \hline 
2 & 1 & (5,6), (10,16), (25,47), (125,248) & 2 \\ 
2 & 1 & (5,5), (10,15), (25,46), (50,98) & 3 \\ 
2 & 1 & (5,1), (10,11), (25,44), (50,98) & 5 \\ 
2 & 1 & (5,2), (25,42) & 10 \\ 
2 & 1 & (5,4), (10,12) & 25 \\ 
2 & 1 & (5,1) & 50 \\ \hline 
4 & 1 & (5,12), (10,32), (25,96), (50,196), (125,496) & 2 \\ 
4 & 1 & (5,2), (10,21), (25,89), (50,194) & 5 \\ 
4 & 1 & (5,3), (25,83) & 10 \\ 
4 & 1 & (5,8), (10,21) & 25 \\ 
4 & 1 & (5,2) & 50 \\        \hline
    \end{tabular}
\end{table}
From the table, we observe that the randomized algorithm succeeds for a wide range of parameters of interest.
\section{Results}\label{sec:results}
We present lifted PDAs and corresponding coded caching schemes using the Blackburn-compatible PDAs constructed in the previous sections. To bring out the versatility of the lifting procedure, we present lifted PDAs constructed for a given number of users $K$ and a wide range of memory sizes and rates.
\subsection{$K$ is a power of $2$}
Starting with 2-PDAs, we consider lifting to obtain PDAs with a coding gain of $2^r$, $r=2,3,\ldots$.
\begin{theorem}[$2^r$-lifting]\label{lem:lifting2r}
\begin{enumerate}[nosep]
    \item Given a $2$-regular $(K_b,f_b,Z_b,\mathcal{S}_b)$ PDA $P_b$, there exists an $2^r$-regular $(2^rK_b,2^rf_b,(2^r-r-1)f_b\allowbreak+rZ_b,(2+2r)|\mathcal{S}_b|+K_bZ_b)$ PDA.
    \item For coded caching with $K$ users, if $2^r\mid K$ for an integer $r$, the memory-rate pair $\big(\frac{N}{K}(K(1-2^{-r}(r+1))+r),2^{-2r}K(r+1)\allowbreak-2^{-r}r\big)$ is achievable with linear subpacketizattion.
\end{enumerate}
\end{theorem}
\begin{proof}
For the first part, use Corollary~\ref{co:regularTiling} to lift the given 2-regular base PDA $P_b$ using the Blackburn-compatible PDAs obtained from Lemma \ref{lem:A2r_compatibility} as constituent PDAs. Since $I_{2^r}(t)$ has $(2^r-1)$ $\x$'s per column and $A_{2^r}(x_i)$, $A'_{2^r}(x_i)$ have $(2^r-r-1)$ $\x$'s per column, the number of $\x$ in each column of the lifted PDA is $(Z)(2^r-1)+(f_b-Z_b)(2^r-r-1)=(2^r-r-1)f+rZ$. The other parameter values are easy to establish.

For the second part, let $q=\frac{K}{2^r}$. Construct a $(q,q,1,q(q-1)/2)$ $2$-PDA using Lemma~\ref{lm:diagPDA}. Lift this PDA using the above first part of the theorem to obtain a $(K,K,K(1-(r+1)2^{-r})+r,2^{-2r}K^2(r+1)-2^{-r}Kr)$ $2^r$-PDA. This would result in the memory-rate pair as claimed.
\end{proof}
If $n$ is a multiple of certain specific powers of 2, then the following theorem provides recursive lifting constructions for $n\times n$ PDAs.
\begin{theorem}\label{th:nested2g}
For $q,r \in \mathbb{Z}^+$ and $q\ge 2$, if $n=2^{\frac{r(r+1)}{2}-1} q$, then there exists a $2^r$-regular $(n,n,Z_r,S_r)$ PDA where $Z_r$ and $S_r$ are given by the following recursive relations:
\begin{align*}
    Z_i &= (2^{i}-2)Z_{i-1}+n_{i-1},\qquad\qquad\quad Z_1=1 \\
    S_i &= \frac{2^i(2^i-1)}{2}S_{i-1}+n_{i-1}Z_{i-1},\qquad S_1=\frac{q(q-1)}{2}.
\end{align*}
with $n_i=2^{\frac{i(i+1)}{2}-1} q$.
\end{theorem}
\begin{proof}
Let $P_1$ be a $2$-regular $(q,q,1,\frac{q(q-1)}{2})$ PDA constructed using Lemma~\ref{lm:diagPDA}. For $i>1$, let $L_i(\cdot)$ denote the lifting using Corollary \ref{co:regularTiling} of a $2^{i-1}$-regular PDA to a $2^i$-regular PDA using the set of $2^{i-1}$ $(2^i,2^i,1,2^{i-1}(2^i-1))$ PDAs obtained using C2. 

Let $P_i=L_i(P_{i-1})))$. Now, $P_{i-1}$ is an $(n_{i-1},n_{i-1},Z_{i-1},S_{i-1})$ PDA with $n_1=q$. Each integer in $P_{i-1}$ is replaced by a $(2^i,2^i,1,2^{i-1}(2^i-1))$ PDA (1 $\x$ per column) and each $\x$ is replaced by $I_{2^i}$ ($2^i-1$ $\x$s per columns). So, in one column of the lifted PDA, there are $Z_{i-1}(2^i-1)+(n_{i-1}-Z_{i-1})=(2^i-2)Z_{i-1}+n_{i-1}$. Clearly, $n_i=2^i n_{i-1}=2^{2+\cdots+(i-1)}q=2^{\frac{i(i-1)}{2}-1}q$.

Then $P_r$ is a $2^r$-regular $(n,n,Z_r,S_r)$ PDA as claimed. 
\end{proof}
The above theorem helps to obtain integer-dense PDAs since the base PDA and constituent PDAs are 2-PDAs with a single $\x$ per column. It can be used to obtain $8$-regular PDAs when $n$ is divisible by $32$, $16$-regular PDAs when $n$ is divisible by 512 and so on.

The memory-rate and memory-subpacketization tradeoff for $K=64$ using the lifting schemes BW2, C2 and $2^r$-lifting is compared with some known schemes in Figs.~\ref{fig:K64} and \ref{fig:K64sub}. The line between the achievable memory-rate pairs in Fig.~\ref{fig:K64} is obtained using memory sharing. In addition, memory-rate tradeoff of lifting schemes for $K=256$ is shown in Fig. \ref{fig:K256}. 
Minimum values of $Z$ obtained for each coding gain $g$ are highlighted in red and are labelled using $(Z,g)$.
For the points highlighted in solid red, the constructions are provided in Table~\ref{tab:samplecon}.

\begin{figure}[htb]
\begin{subfigure}[b]{0.4\textwidth}
    	\definecolor{cadmiumgreen}{rgb}{0.0, 0.42, 0.24}
	\definecolor{mycolor1}{rgb}{0.80000,0.94700,0.99100}%
    \definecolor{wine}{HTML}{882255}
\begin{tikzpicture}[scale=1.0]

\begin{axis}[%
width=3in,
height=3in,
at={(0.758in,0.481in)},
scale only axis,
xmin=0,
xmax=0.9,
xlabel style={font=\color{white!15!black}},
xlabel={Normalized memory},
ymin=0,
ymax=40,
ylabel style={font=\color{white!15!black}},
ylabel={Rate},
axis background/.style={fill=white},
legend style={legend cell align=left, align=left, draw=white!15!black}
]
\addplot [color=red,very thick]
  table[row sep=crcr]{%
0	64\\
0.015625	31.5\\
0.1875	13\\
0.5	4\\
0.75	1\\
0.890625	0.21875\\
0.984375	0.015625\\
};
\addlegendentry{BW2, C2 and $2^r$-lifting}
\addplot[only marks, mark=*, mark options={}, mark size=1.000pt, draw=mycolor1, fill=mycolor1] table[row sep=crcr]{%
x	y\\
0	64\\
0.015625	31.5\\
0.03125	31\\
0.0625	30\\
0.125	28\\
0.15625	27\\
0.171875	26.5\\
0.1875	26\\
0.1875	13\\
0.21875	25\\
0.21875	12.5\\
0.234375	24.5\\
0.234375	12.25\\
0.25	24\\
0.25	12\\
0.265625	23.5\\
0.28125	23\\
0.28125	11.5\\
0.296875	22.5\\
0.296875	11.25\\
0.3125	22\\
0.3125	11\\
0.328125	21.5\\
0.328125	10.75\\
0.34375	21\\
0.34375	10.5\\
0.359375	20.5\\
0.375	20\\
0.375	10\\
0.390625	19.5\\
0.40625	19\\
0.421875	18.5\\
0.421875	9.25\\
0.4375	18\\
0.4375	9\\
0.453125	17.5\\
0.46875	17\\
0.46875	8.5\\
0.484375	16.5\\
0.5	16\\
0.5	8\\
0.5	4\\
0.515625	7.75\\
0.53125	7.5\\
0.53125	3.75\\
0.546875	7.25\\
0.546875	3.625\\
0.5625	7\\
0.5625	3.5\\
0.578125	3.375\\
0.59375	3.25\\
0.609375	6.25\\
0.625	6\\
0.625	3\\
0.65625	5.5\\
0.65625	2.75\\
0.671875	2.625\\
0.6875	2.5\\
0.703125	4.75\\
0.71875	2.25\\
0.734375	4.25\\
0.734375	2.125\\
0.75	4\\
0.75	2\\
0.75	1\\
0.765625	1.875\\
0.765625	0.9375\\
0.78125	1.75\\
0.796875	0.8125\\
0.8125	1.5\\
0.8125	0.75\\
0.828125	1.375\\
0.84375	0.625\\
0.859375	0.5625\\
0.875	1\\
0.875	0.5\\
0.890625	0.4375\\
0.890625	0.21875\\
0.90625	0.375\\
0.90625	0.1875\\
0.921875	0.3125\\
0.921875	0.15625\\
0.9375	0.25\\
0.9375	0.125\\
0.953125	0.09375\\
0.96875	0.0625\\
0.984375	0.015625\\
1	0\\
};
\addlegendentry{Other proposed constructions}
    \addlegendimage{no markers,thick,green}
    \addlegendentry{M-N scheme \cite{maddah2014fundamental}}

\addplot [color=magenta,thick]
  table[row sep=crcr]{%
0	64.2413162705667\\
0.0313253012048193	30.9689213893967\\
0.0626506024096386	19.963436928702\\
0.125301204819277	11.1334552102377\\
0.250602409638554	5.37477148080439\\
0.501204819277108	1.91956124314442\\
0.860240963855422	0.255941499085923\\
0.921686746987952	0.127970749542962\\
0.953012048192771	0.0895795246800789\\
0.968674698795181	0.0511882998171817\\
1	0\\
};
\addlegendentry{Tang \textit{et al.} \cite{tang2018coded}}
    \addlegendimage{no markers,blue,thick}
    \addlegendentry{Grouping, c=8 \cite{shanmugam2016finite}}
    \addlegendimage{no markers,wine,thick}
    \addlegendentry{Grouping, c=4 \cite{shanmugam2016finite}}
	\draw[green,thick] (0,64) foreach \x [evaluate=\x as \r using (64-\x)/(1+\x)] in {1,...,64} { -- (\x/64,\r) };
	\draw[blue,thick] (0,64) foreach \x [evaluate=\x as \r using 8*(8-\x)/(1+\x)] in {1,...,8} { -- (\x/8,\r) };
	\draw[wine,thick] (0,64) foreach \x [evaluate=\x as \r using 4*(16-\x)/(1+\x)] in {1,...,16} { -- (\x/16,\r) };
\addplot [color=cadmiumgreen,only marks,thick]
  table[row sep=crcr]{%
0.875	0.88889\\
};
\addlegendentry{$(64,8,1)$-BIBD \cite{agrawal2019coded}}
\node[label={[outer sep=-2pt]45:\tiny{$(1,2)$}}] at (axis cs: 0.015625,31.5) {} ;
\node[label={[outer sep=-2pt]180:\tiny{$(12,4)$}}] at (axis cs: 0.1875,13) {} ;
\path[<->, draw] (axis cs: 0.5,4) to[out = 100, in = 300]
        (axis cs: 0.49, 20) node[above] {\tiny{$(32,8)$}};
\path[<->, draw] (axis cs: 0.75,1) to[out = 160, in = 265]
        (axis cs: 0.75, 20) node[above] {\tiny{$(48,16)$}};
\path[<->, draw] (axis cs: 0.890625,0.21875) to[out = 90, in = 265]
        (axis cs: 0.85, 15) node[above] {\tiny{$(57,32)$}};
\path[<->, draw] (axis cs: 0.984375,0.015625) to[out = 80, in = 280]
        (axis cs: .95, 8) node[above] {\tiny{$(63,64)$}};
\addplot[only marks, mark options={solid,draw=red,fill=red}] table[row sep=crcr]{%
x	y\\
0	64\\
0.015625	31.5\\
0.1875	13\\
0.5	4\\
0.75	1\\
0.890625	0.21875\\
0.984375	0.015625\\
};

\end{axis}
\end{tikzpicture}%
    \caption{Rate vs memory, $K=64$.}
    \label{fig:K64}
\end{subfigure}\hfill
\begin{subfigure}[b]{0.4\textwidth}
\definecolor{mycolor1}{rgb}{0.80000,0.94700,0.99100}%
\begin{tikzpicture}[scale=1.0]

\begin{axis}[%
clip mode=individual,
clip=true,
width=3in,
height=3.6in,
at={(0.758in,0.481in)},
xmin=0,
xmax=1.02,
ymin=0,
ymax=150,
axis background/.style={fill=white},
xlabel = {Normalized memory},
ylabel = {Rate},
legend style={legend cell align=left, align=left, draw=white!15!black}
]
    \addplot [domain=0:256, samples=257, black,dashed]({x/256},{(256-x)/(1+x)});
    \addlegendentry{MN-scheme}
    \addplot[black,dotted] coordinates {(0,256) (1,0)};
    \addlegendentry{Uncoded caching}
\addplot[only marks, mark=*, mark options={}, mark size=1.000pt, draw=mycolor1, fill=mycolor1] table[row sep=crcr]{%
x	y\\
0.00390625	127.5\\
0.0078125	127\\
0.015625	126\\
0.03125	124\\
0.0625	120\\
0.0703125	119\\
0.078125	118\\
0.08984375	116.5\\
0.09375	116\\
0.1015625	115\\
0.1015625	57.5\\
0.109375	114\\
0.1171875	113\\
0.1171875	56.5\\
0.12109375	56.25\\
0.125	112\\
0.125	56\\
0.12890625	111.5\\
0.1328125	111\\
0.13671875	110.5\\
0.140625	110\\
0.140625	55\\
0.1484375	109\\
0.15234375	54.25\\
0.15625	108\\
0.15625	54\\
0.1640625	107\\
0.1640625	53.5\\
0.171875	106\\
0.171875	53\\
0.17578125	105.5\\
0.17578125	52.75\\
0.1796875	52.5\\
0.1875	104\\
0.1875	52\\
0.19140625	103.5\\
0.1953125	51.5\\
0.19921875	102.5\\
0.203125	102\\
0.203125	51\\
0.2109375	101\\
0.2109375	50.5\\
0.21875	100\\
0.21875	50\\
0.22265625	99.5\\
0.22265625	49.75\\
0.2265625	99\\
0.234375	98\\
0.234375	49\\
0.2421875	97\\
0.24609375	96.5\\
0.24609375	48.25\\
0.25	96\\
0.25	48\\
0.25390625	95.5\\
0.2578125	95\\
0.2578125	47.5\\
0.26171875	94.5\\
0.26171875	47.25\\
0.265625	94\\
0.265625	47\\
0.26953125	93.5\\
0.2734375	93\\
0.2734375	46.5\\
0.28125	92\\
0.28125	46\\
0.28515625	91.5\\
0.2890625	91\\
0.29296875	90.5\\
0.296875	90\\
0.296875	45\\
0.30078125	89.5\\
0.30078125	44.75\\
0.3046875	89\\
0.30859375	88.5\\
0.3125	88\\
0.3125	44\\
0.3125	22\\
0.31640625	87.5\\
0.31640625	43.75\\
0.3203125	87\\
0.32421875	86.5\\
0.328125	86\\
0.328125	43\\
0.33203125	85.5\\
0.3359375	85\\
0.33984375	84.5\\
0.34375	84\\
0.34375	42\\
0.34375	21\\
0.34765625	83.5\\
0.3515625	83\\
0.3515625	41.5\\
0.35546875	82.5\\
0.359375	82\\
0.359375	20.5\\
0.36328125	81.5\\
0.3671875	81\\
0.3671875	20.25\\
0.37109375	80.5\\
0.375	80\\
0.375	40\\
0.37890625	79.5\\
0.3828125	79\\
0.3828125	39.5\\
0.3828125	19.75\\
0.38671875	78.5\\
0.38671875	39.25\\
0.390625	78\\
0.390625	19.5\\
0.39453125	77.5\\
0.3984375	77\\
0.3984375	38.5\\
0.3984375	19.25\\
0.40234375	76.5\\
0.40625	76\\
0.40625	19\\
0.41015625	75.5\\
0.41015625	37.75\\
0.4140625	75\\
0.4140625	18.75\\
0.41796875	74.5\\
0.41796875	18.625\\
0.421875	74\\
0.421875	37\\
0.42578125	73.5\\
0.42578125	18.375\\
0.4296875	73\\
0.43359375	72.5\\
0.43359375	36.25\\
0.4375	72\\
0.4375	36\\
0.4375	18\\
0.44140625	71.5\\
0.4453125	71\\
0.4453125	35.5\\
0.44921875	70.5\\
0.453125	70\\
0.453125	17.5\\
0.45703125	69.5\\
0.45703125	34.75\\
0.4609375	69\\
0.4609375	17.25\\
0.46484375	68.5\\
0.46484375	34.25\\
0.46875	68\\
0.46875	34\\
0.46875	17\\
0.47265625	67.5\\
0.47265625	16.875\\
0.4765625	67\\
0.4765625	16.75\\
0.48046875	66.5\\
0.484375	66\\
0.484375	33\\
0.484375	16.5\\
0.48828125	65.5\\
0.4921875	65\\
0.4921875	32.5\\
0.49609375	64.5\\
0.49609375	16.125\\
0.5	64\\
0.5	32\\
0.5	16\\
0.50390625	31.75\\
0.5078125	31.5\\
0.5078125	15.75\\
0.51171875	15.625\\
0.515625	31\\
0.515625	15.5\\
0.51953125	15.375\\
0.5234375	15.25\\
0.52734375	30.25\\
0.53125	30\\
0.53125	15\\
0.53515625	14.875\\
0.5390625	29.5\\
0.546875	29\\
0.546875	14.5\\
0.55078125	28.75\\
0.5625	28\\
0.5625	14\\
0.56640625	13.875\\
0.5703125	13.75\\
0.57421875	27.25\\
0.578125	13.5\\
0.5859375	26.5\\
0.58984375	13.125\\
0.59375	13\\
0.59765625	25.75\\
0.6015625	12.75\\
0.60546875	25.25\\
0.609375	25\\
0.609375	12.5\\
0.6171875	12.25\\
0.62109375	24.25\\
0.625	24\\
0.625	12\\
0.625	6\\
0.62890625	11.875\\
0.6328125	23.5\\
0.6328125	11.75\\
0.63671875	11.625\\
0.640625	11.5\\
0.64453125	22.75\\
0.64453125	11.375\\
0.64453125	5.6875\\
0.6484375	11.25\\
0.6484375	5.625\\
0.65625	22\\
0.65625	11\\
0.66015625	10.875\\
0.66015625	5.4375\\
0.6640625	10.75\\
0.6640625	5.375\\
0.66796875	21.25\\
0.671875	10.5\\
0.671875	5.25\\
0.67578125	10.375\\
0.6796875	20.5\\
0.68359375	10.125\\
0.68359375	5.0625\\
0.6875	20\\
0.6875	10\\
0.69140625	19.75\\
0.69140625	9.875\\
0.6953125	4.875\\
0.69921875	19.25\\
0.69921875	9.625\\
0.703125	19\\
0.703125	9.5\\
0.703125	4.75\\
0.70703125	9.375\\
0.70703125	4.6875\\
0.7109375	18.5\\
0.7109375	9.25\\
0.71484375	18.25\\
0.71875	18\\
0.71875	9\\
0.71875	4.5\\
0.72265625	17.75\\
0.72265625	8.875\\
0.7265625	17.5\\
0.7265625	4.375\\
0.73046875	17.25\\
0.734375	17\\
0.734375	8.5\\
0.734375	4.25\\
0.73828125	16.75\\
0.73828125	8.375\\
0.7421875	8.25\\
0.7421875	4.125\\
0.74609375	8.125\\
0.75	16\\
0.75	8\\
0.75	4\\
0.75390625	7.875\\
0.75390625	3.9375\\
0.7578125	7.75\\
0.76171875	7.625\\
0.76171875	3.8125\\
0.765625	7.5\\
0.765625	3.75\\
0.76953125	7.375\\
0.7734375	7.25\\
0.7734375	3.625\\
0.77734375	7.125\\
0.78125	7\\
0.78125	3.5\\
0.78515625	3.4375\\
0.7890625	6.75\\
0.7890625	3.375\\
0.79296875	6.625\\
0.796875	3.25\\
0.80078125	6.375\\
0.80078125	3.1875\\
0.8046875	6.25\\
0.8046875	3.125\\
0.80859375	3.0625\\
0.8125	6\\
0.8125	3\\
0.8125	1.5\\
0.81640625	5.875\\
0.8203125	5.75\\
0.8203125	2.875\\
0.82421875	2.8125\\
0.82421875	1.40625\\
0.828125	5.5\\
0.828125	2.75\\
0.83203125	2.6875\\
0.83203125	1.34375\\
0.8359375	2.625\\
0.8359375	1.3125\\
0.83984375	5.125\\
0.84375	2.5\\
0.84765625	4.875\\
0.84765625	1.21875\\
0.8515625	4.75\\
0.8515625	1.1875\\
0.85546875	4.625\\
0.85546875	2.3125\\
0.859375	4.5\\
0.859375	2.25\\
0.859375	1.125\\
0.86328125	4.375\\
0.86328125	1.09375\\
0.8671875	4.25\\
0.8671875	2.125\\
0.87109375	2.0625\\
0.875	4\\
0.875	2\\
0.875	1\\
0.87890625	1.9375\\
0.8828125	1.875\\
0.8828125	0.9375\\
0.88671875	1.8125\\
0.88671875	0.90625\\
0.890625	1.75\\
0.890625	0.875\\
0.89453125	0.84375\\
0.8984375	1.625\\
0.8984375	0.8125\\
0.90234375	1.5625\\
0.90625	1.5\\
0.90625	0.75\\
0.91015625	1.4375\\
0.9140625	1.375\\
0.9140625	0.6875\\
0.9140625	0.34375\\
0.91796875	0.65625\\
0.91796875	0.328125\\
0.921875	1.25\\
0.921875	0.625\\
0.92578125	1.1875\\
0.92578125	0.296875\\
0.9296875	0.5625\\
0.9296875	0.28125\\
0.93359375	1.0625\\
0.93359375	0.53125\\
0.9375	1\\
0.9375	0.5\\
0.9375	0.25\\
0.94140625	0.46875\\
0.94140625	0.234375\\
0.9453125	0.4375\\
0.9453125	0.21875\\
0.94921875	0.203125\\
0.953125	0.375\\
0.953125	0.1875\\
0.95703125	0.34375\\
0.9609375	0.3125\\
0.9609375	0.15625\\
0.96484375	0.140625\\
0.96484375	0.0703125\\
0.96875	0.25\\
0.96875	0.125\\
0.96875	0.0625\\
0.97265625	0.109375\\
0.97265625	0.0546875\\
0.9765625	0.09375\\
0.9765625	0.046875\\
0.98046875	0.078125\\
0.98046875	0.0390625\\
0.984375	0.0625\\
0.984375	0.03125\\
0.98828125	0.0234375\\
0.9921875	0.015625\\
0.99609375	0.00390625\\
1	0\\
};
\node[label={[outer sep=-2pt]45:\tiny{$(1,2)$}}] at (axis cs: 0.00390625,127.5) {} ;
\node[label={[outer sep=-2pt]90:\tiny{$(24,4)$}}] at (axis cs: 0.09375, 58) {} ;
\node[label={[outer sep=-2pt]180:\tiny{$(80,8)$}}] at (axis cs: 0.3125,22) {} ;
\path[->, draw] (axis cs: 0.625,6) to[out = 100, in = 260]
        (axis cs: 0.5, 75) node[above] {\tiny{$(160,16)$}};
\path[->, draw] (axis cs: 0.8125,1.5) to[out = 90, in = 280]
        (axis cs: 0.6, 55) node[above] {\tiny{$(208,32)$}};
\path[->, draw] (axis cs: 0.9140625,0.34375) to[out = 100, in = 280]
        (axis cs: 0.65, 90) node[above] {\tiny{$(234,64)$}};
\path[->, draw] (axis cs: 0.96484375,0.0703125) to[out = 90, in = 280]
        (axis cs: 0.85, 55) node[above] {\tiny{$(247,128)$}};
\path[->, draw] (axis cs: 0.99609375,0.00390625) to[out = 85, in = 300]
        (axis cs: 0.95, 40) node[above] {\tiny{$(255,256)$}};
\addplot[only marks, mark options={solid,draw=red2,fill=red2}] table[row sep=crcr]{%
x	y\\
0.00390625	127.5\\
0.09375	58\\
0.3125	22\\
0.625	6\\
0.8125	1.5\\
0.9140625	0.34375\\
0.96484375	0.0703125\\
0.99609375	0.00390625\\%
    };
    \addplot[only marks, mark options={solid,draw=red,fill=red}] coordinates {(0.3125,22) (0.96484375, 0.0703125)};
\end{axis}
\end{tikzpicture}%
\caption{Rate vs memory, $K=256$.}
\label{fig:K256}
\end{subfigure}

\begin{subfigure}[b]{0.4\textwidth}
	\definecolor{cadmiumgreen}{rgb}{0.0, 0.42, 0.24}
	\definecolor{mycolor1}{rgb}{0.80000,0.94700,0.99100}%
    \definecolor{wine}{HTML}{882255}
\begin{tikzpicture}[scale=1.0]

\begin{axis}[%
width=3in,
height=3in,
at={(0.758in,0.481in)},
scale only axis,
xmin=0,
xmax=1,
xlabel style={font=\color{white!15!black}},
xlabel={Normalized memory},
ymode=log,
ymin=1,
ymax=1e+10,
yminorticks=true,
ylabel style={font=\color{white!15!black}},
ylabel={Subpacketization},
axis background/.style={fill=white},
legend style={at={(0.5,0.98)},anchor=north,legend cell align=left, align=left, draw=white!15!black}
]
\addplot [only marks, mark=*, mark size=1.500pt, mark options={}, draw=red, fill=red,very thick]
  table[row sep=crcr]{%
x	y\\
0.015625	64\\
0.1875	64\\
0.5	64\\
0.75	64\\
0.890625	64\\
0.984375	64\\
};
\addlegendentry{BW2, C2 and $2^r$-lifting}

\addplot [only marks, mark=*, mark size=0.500pt, mark options={}, draw=mycolor1, fill=mycolor1,very thick]
  table[row sep=crcr]{%
x	y\\
0.015625	64\\
0.03125	64\\
0.0625	64\\
0.125	64\\
0.15625	64\\
0.171875	64\\
0.1875	64\\
0.1875	64\\
0.21875	64\\
0.21875	64\\
0.234375	64\\
0.234375	64\\
0.25	64\\
0.25	64\\
0.265625	64\\
0.28125	64\\
0.28125	64\\
0.296875	64\\
0.296875	64\\
0.3125	64\\
0.3125	64\\
0.328125	64\\
0.328125	64\\
0.34375	64\\
0.34375	64\\
0.359375	64\\
0.375	64\\
0.375	64\\
0.390625	64\\
0.40625	64\\
0.421875	64\\
0.421875	64\\
0.4375	64\\
0.4375	64\\
0.453125	64\\
0.46875	64\\
0.46875	64\\
0.484375	64\\
0.5	64\\
0.5	64\\
0.5	64\\
0.515625	64\\
0.53125	64\\
0.53125	64\\
0.546875	64\\
0.546875	64\\
0.5625	64\\
0.5625	64\\
0.578125	64\\
0.59375	64\\
0.609375	64\\
0.625	64\\
0.625	64\\
0.65625	64\\
0.65625	64\\
0.671875	64\\
0.6875	64\\
0.703125	64\\
0.71875	64\\
0.734375	64\\
0.734375	64\\
0.75	64\\
0.75	64\\
0.75	64\\
0.765625	64\\
0.765625	64\\
0.78125	64\\
0.796875	64\\
0.8125	64\\
0.8125	64\\
0.828125	64\\
0.84375	64\\
0.859375	64\\
0.875	64\\
0.875	64\\
0.890625	64\\
0.890625	64\\
0.90625	64\\
0.90625	64\\
0.921875	64\\
0.921875	64\\
0.9375	64\\
0.9375	64\\
0.953125	64\\
0.96875	64\\
0.984375	64\\
};
\addlegendentry{Other proposed constructions}

\addplot [only marks, mark=triangle*, mark size=1.500pt, mark options={}, draw=green, fill=green,thick]
  table[row sep=crcr]{%
0.015625	64\\
0.03125	2016\\
0.046875	41664\\
0.0625	635376\\
0.078125	7624512\\
0.09375	74974368\\
0.109375	621216192\\
0.125	4426165368\\
0.140625	27540584512\\
0.15625	151473214816\\
0.171875	743595781824\\
0.1875	3284214703056\\
0.203125	13136858812224\\
0.21875	47855699958816\\
0.234375	159518999862720\\
0.25	488526937079580\\
0.265625	1.37937017528352e+15\\
0.28125	3.60168879101808e+15\\
0.296875	8.71987812562272e+15\\
0.3125	1.96197257826511e+16\\
0.328125	4.11079968779356e+16\\
0.34375	8.03474484432379e+16\\
0.359375	1.46721427592e+17\\
0.375	2.50649105469666e+17\\
0.390625	4.01038568751466e+17\\
0.40625	6.01557853127199e+17\\
0.421875	8.46636978475316e+17\\
0.4375	1.11877029298524e+18\\
0.453125	1.3888182947403e+18\\
0.46875	1.62028801053035e+18\\
0.484375	1.77709007606554e+18\\
0.5	1.83262414094259e+18\\
0.515625	1.77709007606554e+18\\
0.53125	1.62028801053035e+18\\
0.546875	1.3888182947403e+18\\
0.5625	1.11877029298524e+18\\
0.578125	8.46636978475316e+17\\
0.59375	6.01557853127199e+17\\
0.609375	4.01038568751466e+17\\
0.625	2.50649105469666e+17\\
0.640625	1.46721427592e+17\\
0.65625	8.03474484432379e+16\\
0.671875	4.11079968779356e+16\\
0.6875	1.96197257826511e+16\\
0.703125	8.71987812562272e+15\\
0.71875	3.60168879101808e+15\\
0.734375	1.37937017528352e+15\\
0.75	488526937079580\\
0.765625	159518999862720\\
0.78125	47855699958816\\
0.796875	13136858812224\\
0.8125	3284214703056\\
0.828125	743595781824\\
0.84375	151473214816\\
0.859375	27540584512\\
0.875	4426165368\\
0.890625	621216192\\
0.90625	74974368\\
0.921875	7624512\\
0.9375	635376\\
0.953125	41664\\
0.96875	2016\\
0.984375	64\\
};
\addlegendentry{M-N scheme}

\addplot [only marks, mark=pentagon*,
mark size=1.500pt, mark options={}, draw=magenta, fill=magenta,thick]
  table[row sep=crcr]{%
0.0313253012048193	31.5562888865655\\
0.0626506024096386	258.918275433181\\
0.125301204819277	4166.23783430122\\
0.250602409638554	67038.6733532926\\
0.501204819277108	67038.6733532926\\
0.860240963855422	392780.490908326\\
0.921686746987952	56649.9158052173\\
0.953012048192771	7510.78893144068\\
0.968674698795181	1952.88386314809\\
};
\addlegendentry{Tang \textit{et al.}}

\addplot [only marks, mark=square*, mark size=1.500pt, mark options={}, draw=blue, fill=blue,thick]
  table[row sep=crcr]{%
0.125	8\\
0.25	28\\
0.375	56\\
0.5	70\\
0.625	56\\
0.75	28\\
0.875	8\\
};
\addlegendentry{Grouping, c=8}

\addplot [only marks, mark=diamond*, mark size=1.500pt, mark options={}, draw=wine, fill=wine, thick]
  table[row sep=crcr]{%
0.0625	16\\
0.125	120\\
0.1875	560\\
0.25	1820\\
0.3125	4368\\
0.375	8008\\
0.4375	11440\\
0.5	12870\\
0.5625	11440\\
0.625	8008\\
0.6875	4368\\
0.75	1820\\
0.8125	560\\
0.875	120\\
0.9375	16\\
};
\addlegendentry{Grouping, c=4}

\addplot [only marks, mark=*, mark size=1.500pt, mark options={}, draw=cadmiumgreen, fill=cadmiumgreen,thick,only marks]
  table[row sep=crcr]{%
0.875	72\\
};
\addlegendentry{$(64,8,1)$-BIBD}

\addplot [only marks, mark=*, mark size=1.500pt, mark options={}, draw=red, fill=red,very thick]
  table[row sep=crcr]{%
x	y\\
0.015625	64\\
0.1875	64\\
0.5	64\\
0.75	64\\
0.890625	64\\
0.984375	64\\
};

\end{axis}
\end{tikzpicture}%
\caption{Subpacketization, $K=64$.}
\label{fig:K64sub}
\end{subfigure}
    \caption{Results: $K$ is a power of 2.}
    \label{fig:comparisonPlot}
\end{figure}
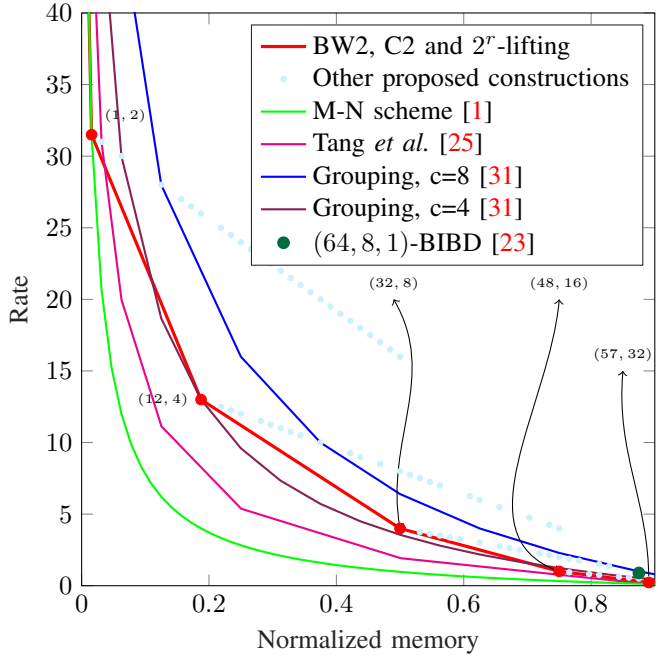
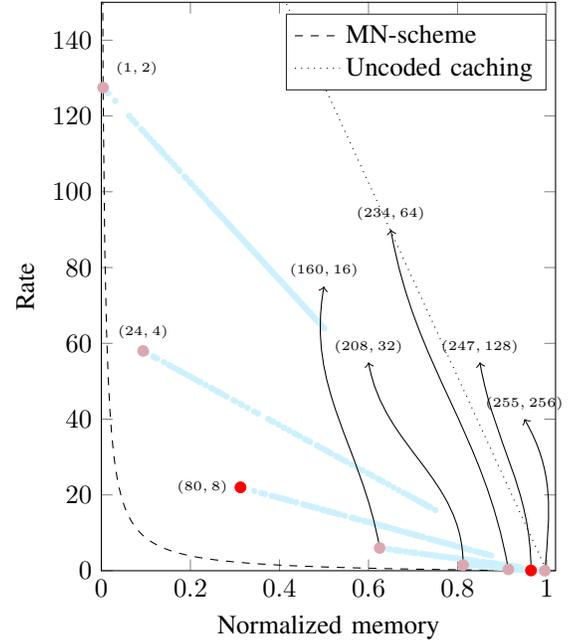
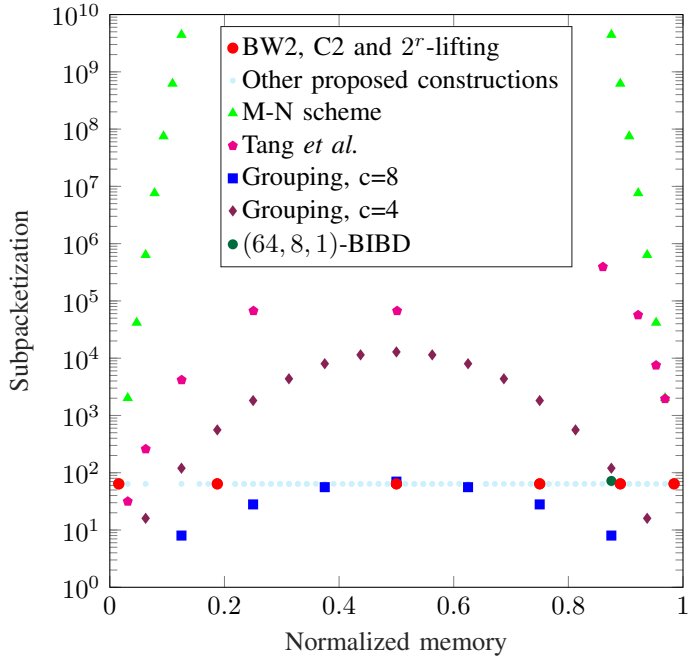
We see that the lifting scheme has better rate than the grouping scheme with $c=8$ for all memory. For memory ratios $0.71<M/N<1$, lifting has better rate than grouping with $c=4$. In Fig.~\ref{fig:K64sub}, we see that for most values of cache memory ratio, the subpacketization of lifting schemes is noticeably better than other comparable schemes, except for the grouping scheme with $c=8$.
Also, our schemes are flexible to provide a wide range of intermediate points without an increase in subpacketization.

\subsection{$K$ with many divisors}
When $K$ has many small divisors, lifting can be applied in multiple ways. A good approach is to consider as many possibilities of lifting as possible and find constructions that achieve the best tradeoff between cache memory and rate. In this section, we present such results for $K=24, 60, 240, 250$. The same procedures can be employed for other such values of $K$ to generate PDAs. 

The memory-rate tradeoffs of various different lifting constructions for $K=24,60,240,250$ are shown in Fig.~\ref{fig:lifting24}.
\begin{figure}[htb]
\begin{subfigure}{0.5\textwidth}
    \centering
    \definecolor{mycolor1}{rgb}{0.80000,0.94700,0.99100}%
\begin{tikzpicture}[scale=1.0]

\begin{axis}[%
clip mode=individual,
clip=false,
width=3in,
height=3in,
at={(0.758in,0.481in)},
xmin=0,
xmax=1.02,
ymin=0,
ymax=25,
ytick = {6,12,18,24},
axis background/.style={fill=white},
xlabel = {$\frac{M}{N}$},
ylabel = {R},
legend style={legend cell align=left, align=left, draw=white!15!black}
]
    \addplot [domain=0:24, samples=25, black,dashed]({x/24},{(24-x)/(1+x)});
    \addlegendentry{MN-scheme}
    \addplot[black,dotted] coordinates {(0,24) (1,0)};
    \addlegendentry{Uncoded caching}
\addplot[only marks, mark=*, mark options={}, mark size=1.000pt, draw=mycolor1,fill=mycolor1] table[row sep=crcr]{%
x	y\\
0	24\\
0.0416666666666667	11.5\\
0.0833333333333333	11\\
0.125	10.5\\
0.166666666666667	10\\
0.208333333333333	9.5\\
0.25	9\\
0.291666666666667	8.5\\
0.291666666666667	4.25\\
0.333333333333333	8\\
0.333333333333333	5.33333333333333\\
0.333333333333333	4\\
0.375	7.5\\
0.375	5\\
0.375	3.75\\
0.416666666666667	7\\
0.416666666666667	4.66666666666667\\
0.416666666666667	3.5\\
0.458333333333333	6.5\\
0.5	6\\
0.5	4\\
0.5	3\\
0.541666666666667	2.75\\
0.583333333333333	3.33333333333333\\
0.583333333333333	2.5\\
0.625	2.25\\
0.625	1.5\\
0.625	1.125\\
0.666666666666667	2.66666666666667\\
0.666666666666667	2\\
0.666666666666667	1\\
0.708333333333333	1.75\\
0.708333333333333	1.16666666666667\\
0.708333333333333	0.875\\
0.75	1.5\\
0.75	1\\
0.75	0.75\\
0.791666666666667	0.833333333333333\\
0.791666666666667	0.625\\
0.833333333333333	0.666666666666667\\
0.833333333333333	0.5\\
0.833333333333333	0.333333333333333\\
0.875	0.375\\
0.875	0.25\\
0.916666666666667	0.166666666666667\\
0.916666666666667	0.125\\
0.958333333333333	0.0416666666666667\\
0.302083333333333	5.58333333333333\\
0.354166666666667	5.16666666666667\\
0.364583333333333	7.625\\
0.458333333333333	4.33333333333333\\
0.5	6\\
0.59375	1.625\\
0.614583333333333	2.3125\\
0.677083333333333	1.29166666666667\\
0.729166666666667	1.08333333333333\\
0.75	1.5\\
0.864583333333333	0.270833333333333\\
1	0\\
};
\node[label={[outer sep=-2pt]45:\tiny{$(1,2)$}}] at (axis cs: 0.0416666666666667,11.5) {} ;
\node[label={[outer sep=-2pt]140:\tiny{$(9,3)$}}] at (axis cs: 0.302083333333333, 5.58333333333333) {} ;
\node[label={[outer sep=-2pt]180:\tiny{$(7,4)$}}] at (axis cs: 0.29, 4.25) {} ;
\path[->, draw] (axis cs: 0.59375, 1.625) to[out = 100, in = 280]
        (axis cs: 0.5, 7) node[above] {\tiny{$(15,6)$}};
\path[->, draw] (axis cs: 0.625, 1.125) to[out = 40, in = 270]
        (axis cs: 0.6, 8) node[above] {\tiny{$(15,8)$}};
\path[->, draw] (axis cs: 0.8333, 0.3333) to[out = 100, in = 280]
        (axis cs: 0.7, 5) node[above] {\tiny{$(20,12)$}};
\path[->, draw] (axis cs: 0.916666666666667, 0.125) to[out = 80, in = 260]
        (axis cs: 0.8, 7) node[above] {\tiny{$(22,16)$}};
\path[->, draw] (axis cs: 0.958333333333333, 0.0416666666666667) to[out = 70, in = 280]
        (axis cs: 0.9, 4.5) node[above] {\tiny{$(23,24)$}};
\addplot[only marks, mark options={solid,draw=red2,fill=red2}] table[row sep=crcr]{%
x	y\\
0	24\\
0.0416666666666667	11.5\\
0.302083333333333	5.58333333333333\\
0.291666666666667	4.25\\
0.59375	1.625\\
0.625	1.125\\
0.833333333333333	0.333333333333333\\
0.916666666666667	0.125\\
0.958333333333333	0.0416666666666667\\ 
    };
    \addplot[only marks, mark options={solid,draw=red,fill=red}] coordinates {(0.29, 4.25) (0.625, 1.125)};
\end{axis}
\end{tikzpicture}%
    \caption{$K=24$}
    \label{fig:first}
\end{subfigure}%
~
\begin{subfigure}{0.5\textwidth}
    \centering
    \definecolor{mycolor1}{rgb}{0.80000,0.94700,0.99100}%
\begin{tikzpicture}[scale=1.0]

\begin{axis}[%
clip mode=individual,
clip=false,
width=3in,
height=3in,
at={(0.758in,0.481in)},
xmin=0,
xmax=1.02,
ymin=0,
ymax=67,
axis background/.style={fill=white},
xlabel = {$\frac{M}{N}$},
ylabel = {R},
legend style={legend cell align=left, align=left, draw=white!15!black}
]
    \addplot [domain=0:60, samples=61, black,dashed]({x/60},{(60-x)/(1+x)});
    \addplot[black,dotted] coordinates {(0,60) (1,0)};
\addplot[only marks, mark=*, mark options={}, mark size=1.000pt, draw=mycolor1, fill=mycolor1] table[row sep=crcr]{%
x	y\\
0	60\\
0.0166666666666667	29.5\\
0.0333333333333333	29\\
0.05	28.5\\
0.0666666666666667	28\\
0.0833333333333333	27.5\\
0.1	27\\
0.133333333333333	26\\
0.166666666666667	25\\
0.183333333333333	12.25\\
0.2	24\\
0.216666666666667	11.75\\
0.233333333333333	23\\
0.233333333333333	15.3333333333333\\
0.25	22.5\\
0.25	11.25\\
0.266666666666667	22\\
0.266666666666667	14.6666666666667\\
0.266666666666667	11\\
0.283333333333333	10.75\\
0.3	21\\
0.3	10.5\\
0.333333333333333	20\\
0.333333333333333	10\\
0.35	19.5\\
0.35	13\\
0.35	9.75\\
0.366666666666667	19\\
0.366666666666667	12.6666666666667\\
0.366666666666667	9.5\\
0.383333333333333	18.5\\
0.4	18\\
0.4	12\\
0.4	9\\
0.416666666666667	17.5\\
0.416666666666667	11.6666666666667\\
0.416666666666667	8.75\\
0.433333333333333	17\\
0.433333333333333	11.3333333333333\\
0.45	16.5\\
0.466666666666667	16\\
0.466666666666667	10.6666666666667\\
0.466666666666667	8\\
0.483333333333333	15.5\\
0.5	15\\
0.5	10\\
0.5	7.5\\
0.516666666666667	7.25\\
0.533333333333333	9.33333333333333\\
0.533333333333333	7\\
0.533333333333333	4.66666666666667\\
0.533333333333333	3.5\\
0.55	6.75\\
0.55	4.5\\
0.566666666666667	5.2\\
0.566666666666667	3.25\\
0.583333333333333	6.25\\
0.6	6\\
0.6	4.8\\
0.6	4\\
0.6	3\\
0.616666666666667	7.66666666666667\\
0.616666666666667	5.75\\
0.616666666666667	4.6\\
0.633333333333333	7.33333333333333\\
0.633333333333333	5.5\\
0.633333333333333	4.4\\
0.633333333333333	2.75\\
0.65	4.2\\
0.65	3.5\\
0.666666666666667	6.66666666666667\\
0.666666666666667	5\\
0.666666666666667	4\\
0.666666666666667	2.5\\
0.683333333333333	4.75\\
0.683333333333333	3.8\\
0.683333333333333	3.16666666666667\\
0.7	3.6\\
0.7	3\\
0.7	2.25\\
0.716666666666667	4.25\\
0.733333333333333	4\\
0.733333333333333	3.2\\
0.733333333333333	2.66666666666667\\
0.733333333333333	2\\
0.75	3.75\\
0.75	3\\
0.75	2.5\\
0.75	1.5\\
0.766666666666667	2.33333333333333\\
0.766666666666667	1.75\\
0.766666666666667	1.4\\
0.383333333333333	18.5\\
0.383333333333333	12.3333333333333\\
0.783333333333333	1.08333333333333\\
0.8	2.4\\
0.8	1.5\\
0.8	1.2\\
0.8	1\\
0.4	18\\
0.816666666666667	1.83333333333333\\
0.816666666666667	1.1\\
0.833333333333333	1.66666666666667\\
0.833333333333333	1.25\\
0.833333333333333	1\\
0.416666666666667	17.5\\
0.85	0.9\\
0.85	0.75\\
0.425	17.25\\
0.866666666666667	0.8\\
0.866666666666667	0.666666666666667\\
0.866666666666667	0.5\\
0.433333333333333	11.3333333333333\\
0.883333333333333	0.7\\
0.883333333333333	0.583333333333333\\
0.883333333333333	0.466666666666667\\
0.883333333333333	0.35\\
0.9	0.6\\
0.9	0.4\\
0.9	0.3\\
0.45	8.25\\
0.916666666666667	0.416666666666667\\
0.916666666666667	0.25\\
0.933333333333333	0.266666666666667\\
0.933333333333333	0.2\\
0.933333333333333	0.166666666666667\\
0.466666666666667	10.6666666666667\\
0.95	0.15\\
0.95	0.1\\
0.966666666666667	0.0666666666666667\\
0.966666666666667	0.05\\
0.483333333333333	10.3333333333333\\
0.983333333333333	0.0166666666666667\\
0.5	15\\
0.5	7.5\\
0.525	7.125\\
0.525	5.7\\
0.270833333333333	14.5833333333333\\
0.55	5.4\\
0.283333333333333	21.5\\
0.575	5.1\\
0.291666666666667	14.1666666666667\\
0.6	4.8\\
0.608333333333333	7.83333333333333\\
0.616666666666667	7.66666666666667\\
0.616666666666667	4.6\\
0.625	4.5\\
0.633333333333333	4.4\\
0.65	4.2\\
0.666666666666667	6.66666666666667\\
0.666666666666667	4\\
0.333333333333333	20\\
0.333333333333333	13.3333333333333\\
0.675	3.9\\
0.7	3.6\\
0.7	1.8\\
0.708333333333333	1.75\\
0.354166666666667	12.9166666666667\\
0.716666666666667	2.83333333333333\\
0.716666666666667	1.7\\
0.725	3.3\\
0.733333333333333	2.66666666666667\\
0.75	3.75\\
0.75	1.5\\
0.766666666666667	1.4\\
0.775	1.35\\
0.8	1.2\\
0.4	18\\
0.816666666666667	1.1\\
0.825	1.05\\
0.833333333333333	1.66666666666667\\
0.833333333333333	1\\
0.416666666666667	17.5\\
0.85	0.9\\
0.858333333333333	0.566666666666667\\
0.433333333333333	11.3333333333333\\
0.883333333333333	0.466666666666667\\
0.883333333333333	0.35\\
0.9	0.4\\
0.9	0.3\\
0.908333333333333	0.366666666666667\\
0.916666666666667	0.25\\
0.458333333333333	10.8333333333333\\
0.925	0.225\\
0.466666666666667	10.6666666666667\\
0.941666666666667	0.116666666666667\\
0.95	0.1\\
0.483333333333333	5.16666666666667\\
0.5	15\\
0.5	7.5\\
0.5	5\\
0.5125	4.875\\
0.566666666666667	4.33333333333333\\
0.583333333333333	6.25\\
0.616666666666667	7.66666666666667\\
0.616666666666667	4.6\\
0.633333333333333	4.4\\
0.6375	4.35\\
0.645833333333333	3.54166666666667\\
0.65	4.2\\
0.666666666666667	6.66666666666667\\
0.666666666666667	4\\
0.666666666666667	3.33333333333333\\
0.670833333333333	3.95\\
0.691666666666667	3.7\\
0.7	3.6\\
0.729166666666667	2.70833333333333\\
0.733333333333333	2.66666666666667\\
0.75	3.75\\
0.75	1.5\\
0.766666666666667	1.4\\
0.783333333333333	1.08333333333333\\
0.8	2.4\\
0.8	1.2\\
0.816666666666667	1.1\\
0.833333333333333	1.66666666666667\\
0.833333333333333	1\\
0.833333333333333	0.833333333333333\\
0.845833333333333	0.925\\
0.85	0.9\\
0.870833333333333	0.516666666666667\\
0.883333333333333	0.466666666666667\\
0.891666666666667	0.433333333333333\\
0.9	0.6\\
0.9	0.4\\
0.9	0.3\\
0.916666666666667	0.25\\
0.945833333333333	0.108333333333333\\
0.95	0.1\\
0.5	15\\
0.60625	4.725\\
0.6375	4.35\\
0.81875	1.0875\\
1	0\\
};
\node[label={[outer sep=-2pt]45:\tiny{$(1,2)$}}] at (axis cs: 0.01666667, 29.5) {} ;
\node[label={[outer sep=-2pt]180:\tiny{$(14,3)$}}] at (axis cs: 0.2333333, 15.33333) {} ;
\node[label={[outer sep=-2pt]180:\tiny{$(11,4)$}}] at (axis cs: 0.1833333, 12.25) {} ;
\path[->, draw] (axis cs: 0.525, 5.7) to[out = 160, in = 0]
        (axis cs: 0.2, 9) node[left] {\tiny{$(\frac{63}{120},5)$}};
\path[->, draw] (axis cs: 0.4833333, 5.166667) to[out = 180, in = 0]
        (axis cs: 0.2, 5.5) node[left] {\tiny{$(\frac{116}{240},6)$}};
\path[->, draw] (axis cs: 0.5333333, 3.5) to[out = 180, in = 0]
        (axis cs: 0.2, 2.5) node[left] {\tiny{$(32,8)$}};
\path[->, draw] (axis cs: 0.7, 1.8) to[out = 160, in = 310]
        (axis cs: 0.25, 25) node[above] {\tiny{$(\frac{84}{120},10)$}};
\path[->, draw] (axis cs: 0.7833333, 1.083333) to[out = 150, in = 270]
        (axis cs: 0.35, 30) node[above] {\tiny{$(47,12)$}};
\path[->, draw] (axis cs: 0.8583333, 0.566667) to[out = 150, in = 270]
        (axis cs: 0.5, 25) node[above] {\tiny{$(\frac{103}{120},15)$}};
\path[->, draw] (axis cs: 0.866667, 0.5) to[out = 120, in = 290]
        (axis cs: 0.62, 15) node[above] {\tiny{$(52,16)$}};
\path[->, draw] (axis cs: 0.8833333, 0.35) to[out = 95, in = 300]
        (axis cs: 0.75, 10) node[above] {\tiny{$(53,20)$}};
\path[->, draw] (axis cs: 0.9333333, 0.166667) to[out = 90, in = 270]
        (axis cs: .65, 40) node[above] {\tiny{$(56,24)$}};
\path[->, draw] (axis cs: 0.941667, 0.116667) to[out = 90, in = 270]
        (axis cs: .75, 35) node[above] {\tiny{$(\frac{113}{120},30)$}};
\path[->, draw] (axis cs: 0.966667, 0.05) to[out = 90, in = 270]
        (axis cs: .85, 30) node[above] {\tiny{$(58,40)$}};
\path[->, draw] (axis cs: 0.9833333, 0.0166667) to[out = 90, in = 270]
        (axis cs: .95, 25) node[above] {\tiny{$(59,60)$}};
\addplot[only marks, mark options={solid,draw=red2,fill=red2}] table[row sep=crcr]{%
x	y\\
0	60\\
0.0166667	29.5\\
0.2333333	15.33333\\
0.1833333	12.25\\
0.525	5.7\\
0.4833333	5.166667\\
0.5333333	3.5\\
0.7	1.8\\
0.7833333	1.083333\\
0.8583333	0.566667\\
0.866667	0.5\\
0.8833333	0.35\\
0.9333333	0.166667\\
0.941667	0.116667\\
0.966667	0.05\\
0.9833333	0.0166667\\
};
    \addplot[only marks, mark options={solid,draw=red,fill=red}] coordinates {(0.7833333, 1.083333) (0.8833333, 0.35)};
\end{axis}
\end{tikzpicture}%
    \caption{$K=60$}
    \label{fig:second_a}
\end{subfigure}\\
\begin{subfigure}{0.5\textwidth}
    \centering
    \input{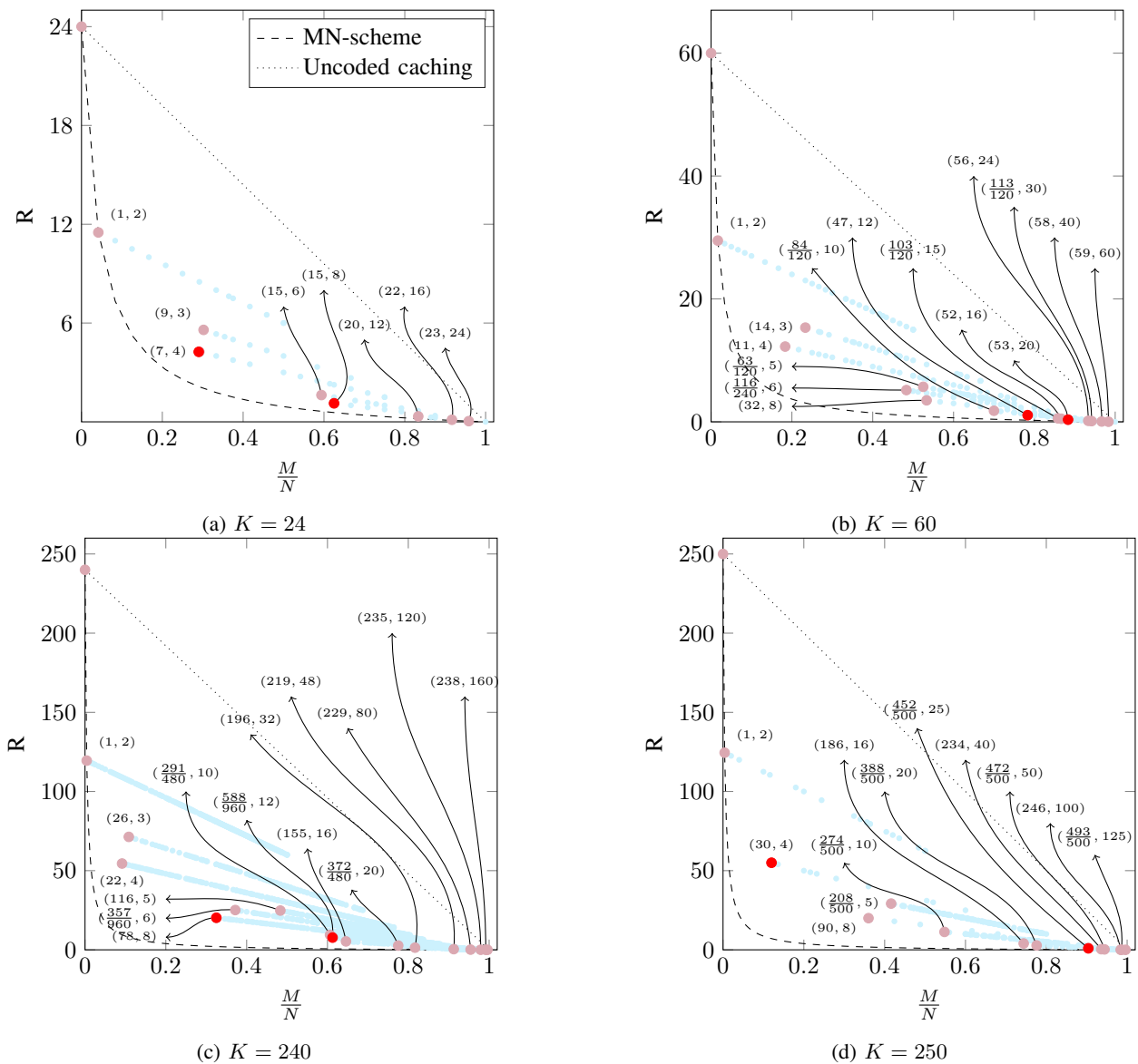}
    \caption{$K=240$}
    \label{fig:third_a}
\end{subfigure}
~
\begin{subfigure}{0.5\textwidth}
    \centering
    \definecolor{mycolor1}{rgb}{0.80000,0.94700,0.99100}%
\begin{tikzpicture}[scale=1.0]

\begin{axis}[%
clip mode=individual,
clip=false,
width=3in,
height=3in,
at={(0.758in,0.481in)},
xmin=0,
xmax=1.02,
ymin=0,
ymax=260,
axis background/.style={fill=white},
xlabel = {$\frac{M}{N}$},
ylabel = {R},
legend style={legend cell align=left, align=left, draw=white!15!black}
]
    \addplot [domain=0:250, samples=251, black,dashed]({x/250},{(250-x)/(1+x)});
    \addplot[black,dotted] coordinates {(0,250) (1,0)};
\addplot[only marks, mark=*, mark options={}, mark size=1.000pt, draw=mycolor1, fill=mycolor1] table[row sep=crcr]{%
x	y\\
0	250\\
0.004	124.5\\
0.008	124\\
0.02	122.5\\
0.04	120\\
0.1	112.5\\
0.116	110.5\\
0.12	55\\
0.136	54\\
0.18	102.5\\
0.2	100\\
0.2	50\\
0.216	49\\
0.232	48\\
0.244	94.5\\
0.28	45\\
0.34	82.5\\
0.356	80.5\\
0.36	40\\
0.36	20\\
0.404	74.5\\
0.42	72.5\\
0.424	18\\
0.436	70.5\\
0.448	27.6\\
0.48	26\\
0.484	64.5\\
0.488	16\\
0.496	25.2\\
0.5	62.5\\
0.504	31\\
0.512	24.4\\
0.52	30\\
0.544	22.8\\
0.56	22\\
0.576	21.2\\
0.592	20.4\\
0.6	25\\
0.6	10\\
0.604	19.8\\
0.608	19.6\\
0.612	9.7\\
0.616	12\\
0.62	19\\
0.64	18\\
0.644	17.8\\
0.648	8.8\\
0.324	84.5\\
0.656	17.2\\
0.66	8.5\\
0.672	16.4\\
0.676	8.1\\
0.68	10\\
0.68	8\\
0.692	15.4\\
0.696	7.6\\
0.7	15\\
0.704	14.8\\
0.712	7.2\\
0.356	80.5\\
0.72	14\\
0.364	79.5\\
0.736	13.2\\
0.74	13\\
0.74	6.5\\
0.744	4\\
0.752	12.4\\
0.756	6.1\\
0.76	6\\
0.768	11.6\\
0.772	5.7\\
0.776	5.6\\
0.784	10.8\\
0.788	10.6\\
0.792	10.4\\
0.792	5.2\\
0.792	2.6\\
0.8	10\\
0.804	4.9\\
0.808	4.8\\
0.808	2.4\\
0.82	4.5\\
0.416	29.2\\
0.836	4.1\\
0.84	4\\
0.84	2\\
0.424	28.8\\
0.852	3.7\\
0.856	1.8\\
0.872	1.6\\
0.88	3\\
0.884	2.9\\
0.896	2.6\\
0.448	27.6\\
0.9	2.5\\
0.904	1.2\\
0.912	0.88\\
0.92	1\\
0.924	0.76\\
0.928	0.72\\
0.936	0.64\\
0.936	0.4\\
0.94	0.6\\
0.944	0.56\\
0.472	26.4\\
0.948	0.52\\
0.948	0.26\\
0.952	0.24\\
0.96	0.4\\
0.964	0.18\\
0.968	0.16\\
0.972	0.14\\
0.98	0.1\\
0.984	0.04\\
0.492	25.4\\
0.988	0.024\\
0.992	0.016\\
0.496	25.2\\
0.996	0.004\\
0.5	62.5\\
0.504	24.8\\
0.506	24.7\\
0.512	24.4\\
0.52	24\\
0.528	23.6\\
0.53	23.5\\
0.534	23.3\\
0.544	22.8\\
0.548	22.6\\
0.548	11.3\\
0.56	22\\
0.56	11\\
0.576	21.2\\
0.584	20.8\\
0.59	20.5\\
0.592	20.4\\
0.594	20.3\\
0.6	10\\
0.604	19.8\\
0.604	9.9\\
0.608	19.6\\
0.612	9.7\\
0.616	9.6\\
0.618	19.1\\
0.62	19\\
0.624	18.8\\
0.624	9.4\\
0.628	9.3\\
0.632	18.4\\
0.64	18\\
0.644	17.8\\
0.648	8.8\\
0.65	17.5\\
0.656	17.2\\
0.66	17\\
0.664	16.8\\
0.664	8.4\\
0.668	16.6\\
0.672	16.4\\
0.672	8.2\\
0.676	8.1\\
0.68	16\\
0.692	15.4\\
0.694	7.65\\
0.696	7.6\\
0.698	15.1\\
0.7	15\\
0.704	14.8\\
0.712	7.2\\
0.356	80.5\\
0.72	7\\
0.724	6.9\\
0.728	6.8\\
0.364	79.5\\
0.74	13\\
0.744	12.8\\
0.744	6.4\\
0.752	12.4\\
0.756	6.1\\
0.76	12\\
0.76	6\\
0.764	5.9\\
0.768	5.8\\
0.772	5.7\\
0.774	5.65\\
0.776	11.2\\
0.776	2.8\\
0.78	5.5\\
0.784	10.8\\
0.786	10.7\\
0.788	10.6\\
0.79	10.5\\
0.792	5.2\\
0.798	5.05\\
0.804	4.9\\
0.808	2.4\\
0.82	4.5\\
0.824	2.2\\
0.83	4.25\\
0.852	3.7\\
0.856	1.8\\
0.864	1.7\\
0.888	2.8\\
0.444	27.8\\
0.892	2.7\\
0.448	27.6\\
0.904	0.96\\
0.906	0.94\\
0.912	0.88\\
0.918	0.82\\
0.924	0.76\\
0.928	0.72\\
0.464	26.8\\
0.93	0.7\\
0.932	0.68\\
0.936	0.64\\
0.94	0.6\\
0.942	0.58\\
0.471	26.45\\
0.944	0.56\\
0.944	0.28\\
0.472	26.4\\
0.948	0.52\\
0.95	0.5\\
0.952	0.24\\
0.956	0.22\\
0.96	0.2\\
0.964	0.18\\
0.966	0.17\\
0.972	0.14\\
0.98	0.1\\
0.492	25.4\\
0.986	0.028\\
0.988	0.024\\
0.99	0.02\\
0.992	0.016\\
0.496	25.2\\
0.5	62.5\\
0.512	24.4\\
0.513	24.35\\
0.528	23.6\\
0.534	23.3\\
0.548	22.6\\
0.555	22.25\\
0.56	22\\
0.562	21.9\\
0.576	21.2\\
0.588	20.6\\
0.59	20.5\\
0.594	20.3\\
0.6	10\\
0.604	19.8\\
0.608	19.6\\
0.612	9.7\\
0.618	19.1\\
0.62	19\\
0.628	9.3\\
0.631	18.45\\
0.632	18.4\\
0.64	18\\
0.644	17.8\\
0.644	8.9\\
0.648	8.8\\
0.656	17.2\\
0.66	17\\
0.664	16.8\\
0.668	16.6\\
0.672	16.4\\
0.672	8.2\\
0.676	8.1\\
0.683	15.85\\
0.692	15.4\\
0.694	7.65\\
0.695	15.25\\
0.696	15.2\\
0.696	7.6\\
0.698	15.1\\
0.7	15\\
0.712	7.2\\
0.728	6.8\\
0.735	13.25\\
0.746	6.35\\
0.752	12.4\\
0.756	6.1\\
0.76	6\\
0.764	5.9\\
0.765	5.875\\
0.768	11.6\\
0.772	11.4\\
0.774	5.65\\
0.776	11.2\\
0.78	11\\
0.781	5.475\\
0.784	10.8\\
0.786	10.7\\
0.788	10.6\\
0.788	5.3\\
0.792	10.4\\
0.792	5.2\\
0.798	5.05\\
0.804	4.9\\
0.808	2.4\\
0.82	4.5\\
0.825	4.375\\
0.83	4.25\\
0.848	3.8\\
0.856	1.8\\
0.86	1.75\\
0.864	1.7\\
0.885	2.875\\
0.444	27.8\\
0.892	2.7\\
0.911	0.89\\
0.912	0.88\\
0.918	0.82\\
0.924	0.76\\
0.928	0.72\\
0.464	26.8\\
0.932	0.68\\
0.936	0.64\\
0.939	0.61\\
0.94	0.6\\
0.471	26.45\\
0.472	26.4\\
0.946	0.54\\
0.947	0.53\\
0.952	0.24\\
0.96	0.4\\
0.964	0.18\\
0.965	0.175\\
0.966	0.17\\
0.968	0.16\\
0.973	0.135\\
0.4885	25.575\\
0.978	0.11\\
0.98	0.1\\
0.492	25.4\\
0.988	0.024\\
0.992	0.016\\
0.5	62.5\\
0.513	24.35\\
0.534	23.3\\
0.548	22.6\\
0.555	22.25\\
0.562	21.9\\
0.576	21.2\\
0.59	20.5\\
0.6145	19.275\\
0.618	19.1\\
0.644	17.8\\
0.644	8.9\\
0.656	17.2\\
0.657	17.15\\
0.672	8.2\\
0.683	15.85\\
0.6895	15.525\\
0.692	15.4\\
0.695	15.25\\
0.698	15.1\\
0.735	6.625\\
0.746	6.35\\
0.765	5.875\\
0.772	11.4\\
0.7725	5.6875\\
0.774	5.65\\
0.776	11.2\\
0.78	11\\
0.781	5.475\\
0.784	10.8\\
0.911	0.89\\
0.918	0.82\\
0.9355	0.645\\
0.936	0.64\\
0.939	0.61\\
0.472	26.4\\
0.4885	25.575\\
0.492	25.4\\
0.5	62.5\\
0.513	24.35\\
0.534	23.3\\
0.548	22.6\\
0.562	21.9\\
0.59	20.5\\
0.6145	19.275\\
0.644	17.8\\
0.672	8.2\\
0.683	15.85\\
0.68625	15.6875\\
0.6895	15.525\\
0.765	5.875\\
0.7725	5.6875\\
0.774	5.65\\
0.776	11.2\\
0.781	5.475\\
0.784	10.8\\
0.918	0.82\\
0.939	0.61\\
0.4885	25.575\\
0.513	24.35\\
0.513	24.35\\
1	0\\
};
\node[label={[outer sep=-2pt]45:\tiny{$(1,2)$}}] at (axis cs: 0.004,124.5) {} ;
\node[label={[outer sep=-2pt]90:\tiny{$(30,4)$}}] at (axis cs: 0.12,55) {} ;
\node[label={[outer sep=-2pt]180:\tiny{$(\frac{208}{500},5)$}}] at (axis cs: 0.416,29.2) {} ;
\node[label={[outer sep=-2pt]190:\tiny{$(90,8)$}}] at (axis cs: 0.36, 20) {} ;
\path[->, draw] (axis cs: 0.548,11.3) to[out = 100, in = 300]
        (axis cs: 0.3, 55) node[above] {\tiny{$(\frac{274}{500},10)$}};
\path[->, draw] (axis cs: 0.744,4) to[out = 130, in = 270]
        (axis cs: 0.3, 120) node[above] {\tiny{$(186,16)$}};
\path[->, draw] (axis cs: 0.776,2.8) to[out = 110, in = 280]
        (axis cs: 0.4, 100) node[above] {\tiny{$(\frac{388}{500},20)$}};
\path[->, draw] (axis cs: 0.904,0.96) to[out = 140, in = 290]
        (axis cs: 0.48, 140) node[above] {\tiny{$(\frac{452}{500},25)$}};
\path[->, draw] (axis cs: 0.936,0.4) to[out = 130, in = 280]
        (axis cs: 0.6, 120) node[above] {\tiny{$(234,40)$}};
\path[->, draw] (axis cs: 0.944,0.28) to[out = 110, in = 270]
        (axis cs: 0.71, 100) node[above] {\tiny{$(\frac{472}{500},50)$}};
\path[->, draw] (axis cs: 0.984,0.04) to[out = 100, in = 270]
        (axis cs: 0.81, 80) node[above] {\tiny{$(246,100)$}};
\path[->, draw] (axis cs: 0.986,0.028) to[out = 85, in = 290]
        (axis cs: 0.92, 60) node[above] {\tiny{$(\frac{493}{500},125)$}};
\addplot[only marks, mark options={solid,draw=red2,fill=red2}] table[row sep=crcr]{%
x	y\\
0	250\\
0.004	124.5\\
0.12	55\\
0.416	29.2\\
0.36	20\\
0.548	11.3\\
0.744	4\\
0.776	2.8\\
0.904	0.96\\
0.936	0.4\\
0.944	0.28\\
0.984	0.04\\
0.986	0.028\\
0.996	0.004\\
    };
    \addplot[only marks, mark options={solid,draw=red,fill=red}] coordinates {(0.12,55) (0.904,0.96)};
\end{axis}
\end{tikzpicture}%
    \caption{$K=250$}
    \label{fig:third}
\end{subfigure}
\caption{Memory-rate tradeoff for $K=24,64,240,250$.}
\label{fig:lifting24}
\end{figure}
For $K=24$, memory-rate tradeoff of the coded caching schemes obtained from our constructions is compared with uncoded and Maddah-Ali-Niesen (MN) schemes. 
Our PDA-based schemes are close to the MN scheme at different parts of the rate versus memory trade-off curve.
Some selected points in the plots are highlighted in darker shade of red. For these selected points, we provide a complete description of the lifting construction in the table below. The notation $(K_b,f_b)_{Z_b}^{g_b}\xrightarrow[X]{(m,n)_{Z_c,Z_*}} (K,f)_{Z}^{g}$ denotes the lifting of a $g_b$-regular $(K_b,f_b,Z_b,\frac{K_b(f_b-Z_b)}{g_b})$ PDA to a $g$-regular $(K,f,Z,\frac{K(f-Z)}{g})$ PDA using a set of $m\times n$ Blackburn-compatible PDAs having $Z_c$ $\x$'s per column and $P_*$ having $Z_*$ $\x$'s per column. $X$ denotes the construction method.

\medskip
\begin{table}[]
    \centering
    \caption{Lifting constructions of some of the PDAs from Fig.~\ref{fig:comparisonPlot} and \ref{fig:lifting24}.}
    \label{tab:samplecon}
\begin{tabular}{|c|l|l|}
\hline
    $K$ & Example 1 & Example 2 \\
    \hline
    24  & $(4,4)_{1}^{2}\xrightarrow[BW3]{(6,6)_{1,4}} (24,24)_{7}^{4}$ & $(3,3)_{1}^{2}\xrightarrow[2^r \mathrm{ Lifting}]{(8,8)_{4,7}} (24,24)_{15}^{8}$\\
    \hline
    60  & $(5,5)_{1}^{2}\xrightarrow[\mathrm{Randomized}]{(12,12)_{9,11}} (60,60)_{47}^{12}$ & $(3,3)_{1}^{2}\xrightarrow[C2]{(4,4)_{1,3}} (12,12)_{5}^{4}\xrightarrow[\mathrm{Basic Lifting}]{(5,5)_{4,5}} (60,60)_{53}^{20}$\\
    \hline
    64  & $(8,8)_{1}^{2}\xrightarrow[BW2]{(8,8)_{1,5}} (64,64)_{12}^{4}$ & $(2,2)_{1}^{2}\xrightarrow[C2]{(4,4)_{1}} (8,8)_{4}^{4}\xrightarrow[C2]{(8,8)_{1}} (64,64)_{32}^{8}$\\
    \hline
    240 & $(5,5)_{1}^{2}\xrightarrow[BW3]{(6,6)_{1,4}} (30,30)_{8}^{4}\xrightarrow[C2]{(8,8)_{1,7}} (240,240)_{78}^{8}$ & $(4,4)_{1}^{2}\xrightarrow[\mathrm{Eq. \scriptsize{\eqref{eq:12by3}}}]{(3,12)_{3,8}} (12,48)_{17}^{3}\xrightarrow[\mathrm{Basic Lifting}]{(4,4)_{1,4}}(48,192)_{99}^{6}$\\
    &&\multicolumn{1}{r|}{$\xrightarrow[\mathrm{Basic Lifting}]{(5,5)_{1,5}} (240,960)_{588}^{12}$}\\
    \hline
    250 & $(5,5)_{1}^{2}\xrightarrow[BW3]{(50,50)_{1,26}} (250,250)_{30}^{4}$ & $(2,2)_{1}^{2}\xrightarrow[\mathrm{Eq. \scriptsize{\eqref{eq:10by5}}}]{(5,10)_{5,8}} (10,20)_{13}^{5}\xrightarrow[\mathrm{Tiling}]{(25,25)_{20,24}} (250,500)_{452}^{25}$\\
    \hline
    256 & $(4,4)_{1}^{2}\xrightarrow[BW2]{(8,8)_{1,5}} (32,32)_{8}^{4}\xrightarrow[C2]{(8,8)_{1,7}} (256,256)_{80}^{8}$ & $(2,2)_{1}^{2}\xrightarrow[2^r \mathrm{ Lifting}]{(128,128)_{120,127}} (256,256)_{247}^{128}$\\
    \hline
\end{tabular}
\end{table}

Table~\ref{tab:samplecon} provides a sample of how a multitude of lifting sequences are possible when $K$ has many small divisors. For instance, in Example 1 for $K=240$, a $2$-regular $5\times 5$ base PDA is first lifted to a $4$-regular $30\times 30$ PDA, which is in turn lifted to an $8$-regular $240 \times 240$ PDA. The two liftings use Blackburn-compatible PDAs from Constructions BW3 and C2, respectively.

\subsection{Randomized construction: $f=K$, $f=2K$, $f=4K$}
To obtain subpacketization as a small multiple of the number of users, randomized construction of Blackburn-compatible PDAs shown in Algorithm~\ref{alg:random} can be used. For $K=250$ and $K=256$, Table~\ref{tab:random} shows some of the resulting lifting constructions. 
\begin{table}[htb]
    \centering
    \caption{Randomized construction.}\label{tab:random}
\begin{tabu}{ccccccc}
    \toprule
    $K$ & $f$ & Gain & Construction && $M/N$ & $R$\\
    \midrule 
250 & 250 & 125 & $(2,2)_{1}^{2}\xrightarrow{(125,125)_{124,124}}(250,250)_{248}^{125}$ && 0.992 & 0.016 \\ 
250 & 250 & 25 & $(2,2)_{1}^{2}\xrightarrow{(5,5)_{3,4}}(10,10)_{7}^{5}\xrightarrow{(25,25)_{22,24}}(250,250)_{234}^{25}$ && 0.936 & 0.64 \\ 
250 & 250 & 10 & $(5,5)_{1}^{2}\xrightarrow{(5,5)_{3,4}}(25,25)_{16}^{5}\xrightarrow{(10,10)_{6,9}}(250,250)_{198}^{10}$ && 0.792 & 5.2 \\ 
250 & 250 & 5 & \makecell[c]{$(2,2)_{1}^{2}\xrightarrow{(5,5)_{3,4}}(10,10)_{7}^{5}\xrightarrow{(5,5)_{0,4}}$\phantom{000000}\\\phantom{000000}$(50,50)_{28}^{5}\xrightarrow{(5,5)_{0,4}}(250,250)_{112}^{5}$} && 0.448 & 27.6 \\ 
256 & 256 & 4 & $(64,64)_{2}^{2}\xrightarrow{(4,4)_{1,3}}(256,256)_{66}^4$ && 0.2578 & 47.5\\
250 & 500 & 125 & $(2,4)_{2}^{2}\xrightarrow{(125,125)_{124,124}}(250,500)_{496}^{125}$ && 0.992 & 0.016 \\ 
250 & 500 & 25 & $(2,2)_{1}^{2}\xrightarrow{(5,10)_{6,8}}(10,20)_{14}^{5}\xrightarrow{(25,25)_{22,24}}(250,500)_{468}^{25}$ && 0.936 & 0.64 \\ 
250 & 500 & 10 & $(5,5)_{4}^{5}\xrightarrow{(5,5)_{0,4}}(25,25)_{16}^{5}\xrightarrow{(10,20)_{11,18}}(250,500)_{387}^{10}$ && 0.774 & 5.65 \\ 
250 & 500 & 5 & \makecell[c]{$(2,2)_{1}^{2}\xrightarrow{(5,10)_{6,8}}(10,20)_{14}^{5}\xrightarrow{(5,5)_{0,4}}$\phantom{000000}\\\phantom{000000}$(50,100)_{56}^{5}\xrightarrow{(5,5)_{0,4}}(250,500)_{224}^{5}$} && 0.448 & 27.6 \\ 
256 & 512 & 4 & $(64,128)_{2}^{2}\xrightarrow{(4,4)_{1,3}}(256,512)_{132}^4$ && 0.2578 & 47.5\\
250 & 1000 & 125 & $(2,8)_{4}^{2}\xrightarrow{(125,125)_{124,124}}(250,1000)_{992}^{125}$ && 0.992 & 0.016 \\ 
250 & 1000 & 25 & $(2,4)_{2}^{2}\xrightarrow{(5,10)_{6,8}}(10,40)_{28}^{5}\xrightarrow{(25,25)_{22,24}}(250,1000)_{936}^{25}$ && 0.936 & 0.64 \\ 
250 & 1000 & 10 & $(5,5)_{4}^{5}\xrightarrow{(5,5)_{0,4}}(25,25)_{16}^{5}\xrightarrow{(10,40)_{21,36}}(250,1000)_{765}^{10}$ && 0.765 & 5.875 \\ 
250 & 1000 & 5 & \makecell[c]{$(2,4)_{2}^{2}\xrightarrow{(5,10)_{6,8}}(10,40)_{28}^{5}\xrightarrow{(5,5)_{0,4}}$\phantom{000000}\\\phantom{000000}$(50,200)_{112}^{5}\xrightarrow{(5,5)_{0,4}}(250,1000)_{448}^{5}$} && 0.448 & 27.6 \\ 
256 & 1024 & 4 & $(64,256)_{2}^{2}\xrightarrow{(4,4)_{1,3}}(256,1024)_{264}^4$ && 0.2578 & 47.5\\ 
\bottomrule
    \end{tabu}
\end{table}
The memory-rate tradeoffs are shown in Fig.~\ref{fig:random}.

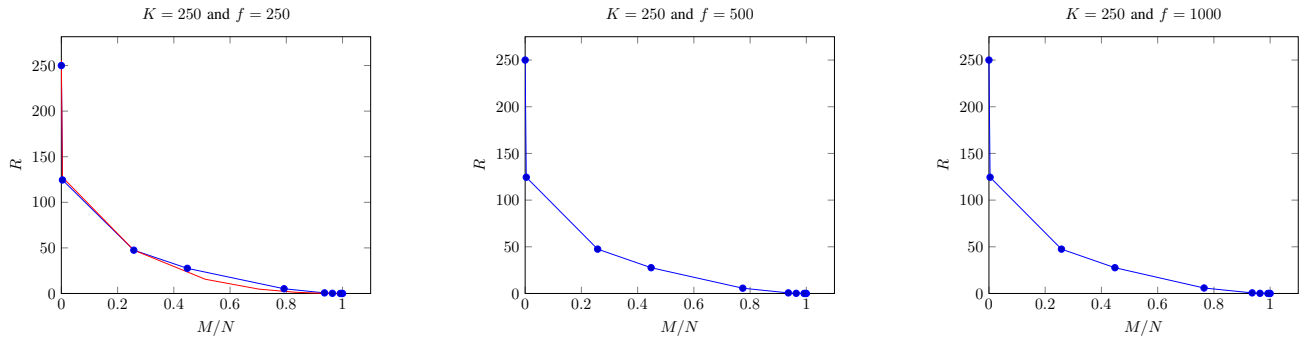
\begin{figure}[!htb]
\minipage{0.32\textwidth}
    \begin{tikzpicture}[scale=0.6]
      \begin{axis}[
        title={$K=250$ and $f=250$},
        xmin=0,
        ymin=0,
        ylabel=$R$,
        xlabel=$M/N$,]
        \addplot coordinates {(1,0) (0.996, 0.004) (0.992, 0.016) (0.964, 0.18) (0.936, 0.64) (0.792, 5.2) (0.448, 27.6) (0.2578125, 47.5) (0.004, 124.5) (0, 250)};
\addplot [color=red]
  table[row sep=crcr]{%
0	256\\
0.00390625	127.5\\
0.2578125	47.5\\
0.51171875	15.625\\
0.703125	4.75\\
0.83203125	1.34375\\
0.9140625	0.34375\\
0.96484375	0.0703125\\
0.99609375	0.00390625\\
1	0\\
};
      \end{axis}
    \end{tikzpicture}
\endminipage\hfill
\minipage{0.32\textwidth}
    \begin{tikzpicture}[scale=0.6]
      \begin{axis}[
        title={$K=250$ and $f=500$},
        xmin=0,
        ymin=0,
        ylabel=$R$,
        xlabel=$M/N$,]
        \addplot coordinates {(1,0) (0.996, 0.004) (0.992, 0.016) (0.964, 0.18) (0.936, 0.64) (0.774, 5.65) (0.448, 27.6) (0.2578125, 47.5) (0.004, 124.5) (0, 250)};
      \end{axis}
    \end{tikzpicture}
\endminipage\hfill
\minipage{0.32\textwidth}%
    \begin{tikzpicture}[scale=0.6]
      \begin{axis}[
        title={$K=250$ and $f=1000$},
        xmin=0,
        ymin=0,
        ylabel=$R$,
        xlabel=$M/N$,]
        \addplot coordinates {(1,0) (0.996, 0.004) (0.992, 0.016) (0.964, 0.18) (0.936, 0.64) (0.765, 5.875) (0.448, 27.6) (0.2578125, 47.5) (0.004, 124.5) (0, 250)};
      \end{axis}
    \end{tikzpicture}
\endminipage
\caption{Memory-rate tradeoff for coded caching schemes from PDAs lifted using randomly constructed Blackburn compatible PDAs. The red curve added for comparison in the first plot is the MR tradeoff obtained using $2^r$-lifting for $K=f=256$.}\label{fig:random}
\end{figure}
We observe that a wide variety of memory vs rate tradeoffs are obtained by lifting with Blackburn compatible PDAs obtained using the randomized algorithm. For $K=f=250$, we have added the memory-rate tradeoff obtained by deterministic $2^r$-lifting for $K=f=256$ (red line) for comparison. We see that the randomized method provides tradeoffs that are comparable with the deterministic one. 

To summarize, in this results section, we have clearly demonstrated the versatility of the lifting construction. For a given number of users, we have shown how the idea of lifting can readily provide multiple lifting constructions for PDAs offering a range of tradeoffs between cache memory size and rate at very low subpacketization.

\section{Conclusion}\label{sec:conc}
We propose several constructions for coded caching schemes with subpacketization linear with the number of users using the framework of placement delivery arrays. We presented a general scheme to construct PDAs with a coding gain of 2. 
We introduced the notion of Blackburn compatibility of PDAs and used this concept for a several lifting constructions of PDAs with higher coding gains. We showed that Blackburn-compatible PDAs can be built from existing sets of Blackburn compatible PDAs through our blockwise and recursive constructions. We also proposed an algorithm to randomly construct Blackburn compatible PDAs for any arbitrary setting. In many regimes, our lifting constructions are shown to perform better compared to other existing schemes for lower subpacketization.


\begin{thebibliography}{10}
	\providecommand{\url}[1]{#1}
	\csname url@samestyle\endcsname
	\providecommand{\newblock}{\relax}
	\providecommand{\bibinfo}[2]{#2}
	\providecommand{\BIBentrySTDinterwordspacing}{\spaceskip=0pt\relax}
	\providecommand{\BIBentryALTinterwordstretchfactor}{4}
	\providecommand{\BIBentryALTinterwordspacing}{\spaceskip=\fontdimen2\font plus
		\BIBentryALTinterwordstretchfactor\fontdimen3\font minus
		\fontdimen4\font\relax}
	\providecommand{\BIBforeignlanguage}[2]{{%
			\expandafter\ifx\csname l@#1\endcsname\relax
			\typeout{** WARNING: IEEEtran.bst: No hyphenation pattern has been}%
			\typeout{** loaded for the language `#1'. Using the pattern for}%
			\typeout{** the default language instead.}%
			\else
			\language=\csname l@#1\endcsname
			\fi
			#2}}
	\providecommand{\BIBdecl}{\relax}
	\BIBdecl
	
	\bibitem{maddah2014fundamental}
	M.~A. Maddah-Ali and U.~Niesen, ``Fundamental limits of caching,'' \emph{IEEE
		Transactions on Information Theory}, vol.~60, no.~5, pp. 2856--2867, 2014.
	
	\bibitem{yu2017exact}
	Q.~Yu, M.~A. Maddah-Ali, and A.~S. Avestimehr, ``The exact rate-memory tradeoff
	for caching with uncoded prefetching,'' \emph{IEEE Transactions on
		Information Theory}, vol.~64, no.~2, pp. 1281--1296, 2017.
	
	\bibitem{sengupta2015improved}
	A.~Sengupta, R.~Tandon, and T.~C. Clancy, ``Improved approximation of
	storage-rate tradeoff for caching via new outer bounds,'' in \emph{2015 IEEE
		International Symposium on Information Theory (ISIT)}.\hskip 1em plus 0.5em
	minus 0.4em\relax IEEE, 2015, pp. 1691--1695.
	
	\bibitem{ghasemi2017improved}
	H.~Ghasemi and A.~Ramamoorthy, ``Improved lower bounds for coded caching,''
	\emph{IEEE Transactions on Information Theory}, vol.~63, no.~7, pp.
	4388--4413, 2017.
	
	\bibitem{wang2016new}
	C.-Y. Wang, S.~H. Lim, and M.~Gastpar, ``A new converse bound for coded
	caching,'' in \emph{2016 Information Theory and Applications Workshop
		(ITA)}.\hskip 1em plus 0.5em minus 0.4em\relax IEEE, 2016, pp. 1--6.
	
	\bibitem{yu2018characterizing}
	Q.~Yu, M.~A. Maddah-Ali, and A.~S. Avestimehr, ``Characterizing the rate-memory
	tradeoff in cache networks within a factor of 2,'' \emph{IEEE Transactions on
		Information Theory}, vol.~65, no.~1, pp. 647--663, 2018.
	
	\bibitem{wan2016optimality}
	K.~Wan, D.~Tuninetti, and P.~Piantanida, ``On the optimality of uncoded cache
	placement,'' in \emph{2016 IEEE Information Theory Workshop (ITW)}.\hskip 1em
	plus 0.5em minus 0.4em\relax IEEE, 2016, pp. 161--165.
	
	\bibitem{maddah2015decentralized}
	M.~A. Maddah-Ali and U.~Niesen, ``Decentralized coded caching attains
	order-optimal memory-rate tradeoff,'' \emph{IEEE/ACM Trans. on Networking
		(TON)}, vol.~23, no.~4, pp. 1029--1040, 2015.
	
	\bibitem{niesen2016coded}
	U.~Niesen and M.~A. Maddah-Ali, ``Coded caching with nonuniform demands,''
	\emph{IEEE Transactions on Information Theory}, vol.~63, no.~2, pp.
	1146--1158, 2016.
	
	\bibitem{karamchandani2016hierarchical}
	N.~Karamchandani, U.~Niesen, M.~A. Maddah-Ali, and S.~N. Diggavi,
	``Hierarchical coded caching,'' \emph{IEEE Transactions on Information
		Theory}, vol.~62, no.~6, pp. 3212--3229, 2016.
	
	\bibitem{chen2016fundamental}
	Z.~Chen, P.~Fan, and K.~B. Letaief, ``Fundamental limits of caching: Improved
	bounds for users with small buffers,'' \emph{IET Communications}, vol.~10,
	no.~17, pp. 2315--2318, 2016.
	
	\bibitem{tian2018caching}
	C.~Tian and J.~Chen, ``Caching and delivery via interference elimination,''
	\emph{IEEE Transactions on Information Theory}, vol.~64, no.~3, pp.
	1548--1560, 2018.
	
	\bibitem{gomez2018fundamental}
	J.~G{\'o}mez-Vilardeb{\'o}, ``Fundamental limits of caching: Improved
	rate-memory tradeoff with coded prefetching,'' \emph{IEEE Transactions on
		Communications}, vol.~66, no.~10, pp. 4488--4497, 2018.
	
	\bibitem{sengupta2014fundamental}
	A.~Sengupta, R.~Tandon, and T.~C. Clancy, ``Fundamental limits of caching with
	secure delivery,'' \emph{IEEE Transactions on Information Forensics and
		Security}, vol.~10, no.~2, pp. 355--370, 2014.
	
	\bibitem{ravindrakumar2016fundamental}
	V.~Ravindrakumar, P.~Panda, N.~Karamchandani, and V.~Prabhakaran, ``Fundamental
	limits of secretive coded caching,'' in \emph{2016 IEEE International
		Symposium on Information Theory (ISIT)}.\hskip 1em plus 0.5em minus
	0.4em\relax IEEE, 2016, pp. 425--429.
	
	\bibitem{wan2020coded}
	K.~Wan and G.~Caire, ``On coded caching with private demands,'' \emph{IEEE
		Transactions on Information Theory}, vol.~67, no.~1, pp. 358--372, 2020.
	
	\bibitem{aravind2020subpacketization}
	V.~R. Aravind, P.~K. Sarvepalli, and A.~Thangaraj, ``Subpacketization in coded
	caching with demand privacy,'' in \emph{2020 National Conference on
		Communications (NCC)}.\hskip 1em plus 0.5em minus 0.4em\relax IEEE, 2020, pp.
	1--6, extended version at arXiv:1909.10471.
	
	\bibitem{kamath2020demand}
	S.~Kamath, J.~Ravi, and B.~K. Dey, ``Demand-private coded caching and the exact
	trade-off for $n=k=2$,'' in \emph{2020 National Conference on Communications
		(NCC)}.\hskip 1em plus 0.5em minus 0.4em\relax IEEE, 2020, pp. 1--6.
	
	\bibitem{shangguan2018centralized}
	C.~Shangguan, Y.~Zhang, and G.~Ge, ``Centralized coded caching schemes: A
	hypergraph theoretical approach,'' \emph{IEEE Transactions on Information
		Theory}, vol.~64, no.~8, pp. 5755--5766, 2018.
	
	\bibitem{cheng2017coded}
	M.~Cheng, Q.~Yan, X.~Tang, and J.~Jiang, ``Coded caching schemes with low rate
	and subpacketizations,'' \emph{arXiv:1703.01548}, 2017.
	
	\bibitem{chittoor2020subexponential}
	H.~H.~S. Chittoor, P.~Krishnan, and K.~Sree, ``Subexponential and linear
	subpacketization coded caching via line graphs and projective geometry,''
	\emph{arXiv preprint arXiv:2001.00399}, 2020.
	
	\bibitem{yan2017placement}
	Q.~Yan, M.~Cheng, X.~Tang, and Q.~Chen, ``On the placement delivery array
	design for centralized coded caching scheme,'' \emph{IEEE Transactions on
		Information Theory}, vol.~63, no.~9, pp. 5821--5833, 2017.
	
	\bibitem{agrawal2019coded}
	S.~Agrawal, K.~S. Sree, and P.~Krishnan, ``Coded caching based on combinatorial
	designs,'' in \emph{2019 IEEE International Symposium on Information Theory
		(ISIT)}.\hskip 1em plus 0.5em minus 0.4em\relax IEEE, 2019, pp. 1227--1231.
	
	\bibitem{wanless2004partial}
	I.~M. Wanless \emph{et~al.}, ``A partial {L}atin squares problem posed by
	{B}lackburn,'' \emph{Bull. Inst. Comb. Appl}, vol.~42, pp. 76--80, 2004.
	
	\bibitem{tang2018coded}
	L.~Tang and A.~Ramamoorthy, ``Coded caching schemes with reduced
	subpacketization from linear block codes,'' \emph{IEEE Transactions on
		Information Theory}, vol.~64, no.~4, pp. 3099--3120, 2018.
	
	\bibitem{shanmugam2017coded}
	K.~Shanmugam, A.~M. Tulino, and A.~G. Dimakis, ``Coded caching with linear
	subpacketization is possible using {R}uzsa-{S}zem{\'e}redi graphs,'' in
	\emph{2017 IEEE International Symposium on Information Theory (ISIT)}.\hskip
	1em plus 0.5em minus 0.4em\relax IEEE, 2017, pp. 1237--1241.
	
	\bibitem{shanmugam2017unified}
	K.~Shanmugam, A.~G. Dimakis, J.~Llorca, and A.~M. Tulino, ``A unified
	{R}uzsa-{S}zemer{\'e}di framework for finite-length coded caching,'' in
	\emph{2017 51st Asilomar Conference on Signals, Systems, and
		Computers}.\hskip 1em plus 0.5em minus 0.4em\relax IEEE, 2017, pp. 631--635.
	
	\bibitem{yan2017bipartite}
	Q.~Yan, X.~Tang, Q.~Chen, and M.~Cheng, ``Placement delivery array design
	through strong edge coloring of bipartite graphs,'' \emph{IEEE Communications
		Letters}, vol.~22, no.~2, pp. 236--239, 2017.
	
	\bibitem{krishnan2018coded}
	P.~Krishnan, ``Coded caching via line graphs of bipartite graphs,'' in
	\emph{2018 IEEE Information Theory Workshop (ITW)}.\hskip 1em plus 0.5em
	minus 0.4em\relax IEEE, 2018, pp. 1--5.
	
	\bibitem{alon2012nearly}
	N.~Alon, A.~Moitra, and B.~Sudakov, ``Nearly complete graphs decomposable into
	large induced matchings and their applications,'' in \emph{Proceedings of the
		forty-fourth annual ACM symposium on Theory of computing}, 2012, pp.
	1079--1090.
	
	\bibitem{shanmugam2016finite}
	K.~Shanmugam, M.~Ji, A.~M. Tulino, J.~Llorca, and A.~G. Dimakis,
	``Finite-length analysis of caching-aided coded multicasting,'' \emph{IEEE
		Transactions on Information Theory}, vol.~62, no.~10, pp. 5524--5537, 2016.
	
	\bibitem{thorpe2003low}
	J.~Thorpe, ``Low-density parity-check ({LDPC}) codes constructed from
	protographs,'' \emph{IPN progress report}, vol.~42, no. 154, pp. 42--154,
	2003.
	
	\bibitem{zhong2020placement}
	X.~Zhong, M.~Cheng, and J.~Jiang, ``Placement delivery array based on
	concatenating construction,'' \emph{IEEE Communications Letters}, 2020.
	
	\bibitem{michel2019placement}
	J.~Michel and Q.~Wang, ``Placement delivery arrays from combinations of strong
	edge colorings,'' in \emph{2019 Ninth International Workshop on Signal Design
		and its Applications in Communications (IWSDA)}.\hskip 1em plus 0.5em minus
	0.4em\relax IEEE, 2019, pp. 1--5.
	
	\bibitem{aravind2020coded}
	V.~R. Aravind, P.~K. Sarvepalli, and A.~Thangaraj, ``Coded caching with demand
	privacy: Constructions for lower subpacketization and generalizations,''
	\emph{arXiv preprint arXiv:2007.07475v1}, 2020.
	
	\bibitem{Harary1969}
	F.~Harary, \emph{Graph Theory}.\hskip 1em plus 0.5em minus 0.4em\relax Reading,
	MA: Addison-Wesley, 1969.
	
	\bibitem{wallis2013one}
	W.~Wallis, \emph{One-Factorizations}, ser. Mathematics and Its
	Applications.\hskip 1em plus 0.5em minus 0.4em\relax Springer US, 2013.
	
	\bibitem{factorizationsurvey85}
	E.~Mendelsohn and A.~Rosa, ``One-factorizations of the complete graph - {A}
	survey,'' \emph{Journal of Graph Theory}, vol.~9, no.~1, pp. 43--65, 1985.
	
	\bibitem{AlspachWalecki08}
	B.~Alspach, ``The wonderful {W}alecki construction,'' \emph{Bulletin of the
		Institute of Combinatorics and its Applications}, vol.~52, 01 2008.
	
\end{thebibliography}

\end{document}